\documentclass[11pt]{article}
\usepackage{graphicx}
\usepackage{pgfplots}
\usepackage[colorlinks=true,citecolor=blue]{hyperref}
\usepackage{caption}
\usepackage[labelformat=simple]{subcaption}
\usepackage{natbib}
\usepackage{graphicx}
\usepackage{amsfonts}
\usepackage{amsmath}
\usepackage{amssymb}
\usepackage{url}
\usepackage{fancyhdr}
\usepackage{indentfirst}
\usepackage{enumerate}
\usepackage{titlesec}
\usepackage{amsthm}
\usepackage{dsfont}
\usepackage[misc]{ifsym}
\usetikzlibrary{patterns}

\theoremstyle{definition}

\newtheorem{example}{Example}

\theoremstyle{plain}
\newtheorem{theorem}{Theorem}
\newtheorem{lemma}{Lemma}
\newtheorem{proposition}{Proposition}
\newtheorem{corollary}{Corollary}
\theoremstyle{remark}
\newtheorem{remark}{Remark}

\usepackage{cases}

\def\laweq{\buildrel \mathrm{d} \over =}

\theoremstyle{definition}

\def\P{\mathbb{P}}
\def\p{\mathbb{P}}
\def\E{\mathbb{E}}

\def\R{\mathbb{R}}

\def\H{\mathcal{H}}

\def\X{\mathcal{X}}

\def\X{\mathcal{X}}

\def\d{\,\mathrm{d}}
\DeclareMathOperator*{\esssup}{ess\text{-}sup}
\DeclareMathOperator*{\essinf}{ess\text{-}inf}

\usepackage[onehalfspacing]{setspace}

\usepackage{bm}
\usepackage{tikz-qtree}
\usepackage{tikz}

\newcommand{\dsquare}{\mathop{  \square} \displaylimits}

\newcommand{\dboxplus}{\mathop{  \boxplus} \displaylimits}

\def\id{\mathds{1}}

\title{Risk sharing, measuring variability, and distortion riskmetrics}
\author{Jean-Gabriel Lauzier\thanks{Department of Economics, Memorial University of Newfoundland,  Canada. \Letter~{\url{jlauzier@mun.ca}}}\and Liyuan Lin\thanks{Department of Econometrics and Business Statistics, Monash University,  Australia. \Letter~{\url{liyuan.lin@monash.edu}} } \and  Ruodu Wang\thanks{Department of Statistics and Actuarial Science, University of Waterloo,  Canada. \Letter~{\url{wang@uwaterloo.ca}}}}
\date{\today}

\begin{document}
	\maketitle
	\begin{abstract}
We address the problem of sharing risk among agents with preferences modelled by a general class of comonotonic additive and law-invariant functionals that need not be either monotone or convex. 
Such functionals are called distortion riskmetrics, which include many statistical measures of risk and variability used in portfolio optimization and insurance. 
The set of Pareto-optimal allocations
is characterized under various settings of general or comonotonic risk sharing problems. We solve explicitly Pareto-optimal allocations among agents using the Gini deviation, the mean-median deviation, or the inter-quantile difference as the relevant variability measures. The latter is of particular interest, as optimal allocations are not comonotonic in the presence of inter-quantile difference agents;
instead, the optimal allocation features a  mixture of pairwise counter-monotonic structures, showing some patterns of extremal negative dependence. 

	\end{abstract}

 \noindent\rule{\textwidth}{0.5pt}

\section{Introduction} 

Anne, Bob and Carole are sharing a  random financial loss. 
After negotiating their respective expected returns,
each of them prefers to minimize a statistical measure of variability of their allocated risk. 
While agreeing on the distribution of the total loss, and that the variance is a poor metric of riskiness,  each of them has their own favourite tool for measuring risks. Anne, as an economics student, likes the Gini deviation (GD) because of its intuitive appearance as an economic index. Bob, as a computer science student, prefers the mean-median deviation (MMD) because it minimizes the mean absolute error. Finally,
Carole, as a statistics student, finds that an inter-quantile difference (IQD) is the most representative of her preference, as she does not worry about extreme events for this particular risk.\footnote{For an integrable random variable $X$, its GD is defined as $\E[|X-X'|]$ where $X'$ is an iid copy of $X$, its
MMD is defined as $\E[|X-m|]$ where $m$ is a median of $X$, and its IQD at level $\alpha\in (0,1/2)$ is the difference between the $\alpha$-quantile and the $(1-\alpha)$-quantile of $X$.
For precise definitions, see Section \ref{sec:2}.} How should Anne, Bob and Carole optimally share risks among themselves?

The reader familiar with risk sharing problems may immediately realize two notable features of such a problem. First, the preferences are not monotone, different from standard decision models in the literature. Second, and most crucially, Carole's preference is neither convex nor consistent with second-order stochastic dominance.
This alludes to the possibility of non-comonotonic Pareto-optimal allocations, in contrast to the comonotonic ones, which are well studied in the literature (e.g., \citealp{LM94, JST08, CDG12, R13}).
 
The GD, the MMD and the IQD are measures of distributional variability. Variability is used to characterize the concept of risk as in the classic work of \cite{MARK52} and \cite{RS70}. For this reason, we also call them \emph{riskmetrics}, which also include \emph{risk measures} in the literature, often associated with monotonicity (e.g., \citealp{FS16}). 
As the most popular measure of variability, 
the variance is known to be a coarse metric and it does not distinguish positive and negative deviations from the mean;\footnote{The latter criticism also applies to GD, MMD and IQD, but most of our results do not assume symmetry and can accommodate non-symmetric riskmetrics.} {\cite{EMS02} discussed various flaws of using variance and correlation in financial risk management.}
Anne's decision criterion has been proposed in \cite{SY84}, which considers an optimal portfolio problem \textit{\`a la} \cite{MARK52}, but with the variance replaced by the GD.\footnote{The authors use the term \textit{Gini's mean difference}.}
  Formally, the authors analyze the investor's problem $\min_{X} \mathrm{GD}(X)$ subject to $\E [X]\geq R$, for a given rate  $R\geq0 $ of return proportional to the investor's risk aversion.
As with mean-variance preferences (e.g., \citealp{M14, MMR13}), the decision criterion can thus be viewed as the problem of maximizing $\E[X]- \eta  \mathrm{GD}(X)$, for $\eta \ge 0$ being the Lagrange multiplier of the problem. While the decision criterion $\E[X]-\eta \mathrm{GD}(X)$ seems natural,  it is not monotone unless $\eta $ is less than or equal to one, in which case the investor's preference belongs to those of \cite{Y87}. The other measures MMD and IQD also have sound foundations and long history in statistics and its applications (\citealp[Chapter 6]{Y11}).  
Slightly different from  MMD, \cite{KY91} studied portfolio optimization using the mean-absolute deviation from the mean. 
Risk sharing problems with convex risk measures are well studied (e.g., ~\citealp{BE05}, \citealp{JST08}, \citealp{FS08} and \cite{RS14}), but the classes of riskmetrics mentioned above do not belong to convex risk measures in general.
When the riskmetric is a deviation measure, the risk sharing problem may become non-monotone, a feature that also appears in the model of \cite{MARK52}.

In this paper, we address the problem of sharing risk among agents that uses \emph{distortion riskmetrics} as their preferences. Distortion riskmetrics are evaluation functionals that are characterized by comonotonic additivity and law invariance \citep{WWW20a}. This rich family includes many measures of risk and variability,  and in particular, the mean, the GD, the MMD, the IQD, and their linear combinations.
Distortion riskmetrics are closely related to Choquet integrals and rank-dependent utilities widely used in decision theory (e.g., \citealp{Y87, S89, CD03}); for a comprehensive treatment of distortions in decisions and economics, see \cite{W10}. 
The combination of the mean and GD or that of the mean and MMD, as well as other distortion riskmetrics,
are used as premium principles in the insurance literature; see \cite{D90}.  Several variability measures within the class of distortion riskmetrics are studied by \cite{GMZ09}, \cite{FWZ17} and \cite{BFWW22}. 

While we analyze the general problem of sharing risk amongst distortion riskmetrics agents, non-monotone and non-cash-additive evaluation functionals receive greater attention. This is for a few reasons. First, the special case of sharing risk with cash-additive and law-invariant functionals is well studied, and more so when the functionals are monotone, but the general case is less understood. Second,   the formalism we introduce allows us to generalize the example above and consider individuals that analyze their risks with different variability measures.  
This is critical because we aim to understand how the act of measuring risk differently gives rise to incentives to trade it.
Also, it allows for non-monetary measurement of risk.  
Third, technically, relaxing monotonicity and convexity allows us to deal with maximization and minimization problems of risk in a unified framework.

The following simple example, by considering the GD and MMD agent only, illustrates the structure of a Pareto-optimal allocation as either an insurance or a financial contract.

\begin{example}\label{ex:intro}

Consider the problem of sharing a random loss $X$ between Anne ($A$) and Bob ($B$) only. Recall that Bob evaluates its allocation $X_B$ using the mean-median deviation  $\mathrm{MMD}(X_B)$. Similarly, Anne's allocation is $X_A$ which she evaluates with the Gini deviation $\mathrm{GD}(X_A)$. We will show (in Section \ref{sec:5}) that any Pareto-optimal allocation takes the form
$$X_A=X\wedge \ell - X\wedge  d, ~~~ X_B=X-X_A,$$
where $\ell \ge d$ will be specified later.   
We can interpret this as a situation where $X$ is  Bob's potential loss and Anne provides some degree of insurance for Bob. The contract (transfer function) is thus simply the random variable $X_A$. Notice that (i) when $\ell \ge X\ge d=0$ there is full insurance, (ii) when $\ell = d$ there is no insurance and (iii) for other choices of $\ell >d$, the contract is a simple deductible $d$ with an upper limit $\ell $. Further, we show that each Pareto-optimal allocation minimizes $\lambda \mathrm{MMD}(X_B)+ (1-\lambda)\mathrm{GD}(X_A)$ for some $\lambda \in[0,1]$. 
If we interpret $P=-X>0$ as the future (after one period) price of an asset that Bob owns one share, then case (iii)  can be achieved by Bob writing to Anne a bull call spread option with long strike price $-l$ and short strike price $-d$. 

\end{example}

The previous example is interesting because it confirms the intuition that the act of measuring risk differently leads to incentives to trade it. Yet, the ``shape" of the solution above is not surprising, as both the Gini deviation and mean-median deviation are convex order consistent functionals, and so exhibit risk aversion of \cite{RS70}.
Just as in the increasing distortion case, risk-minimizing (utility-maximizing) Pareto-optimal allocations are comonotonic when the distortion riskmetrics' distortion function is concave (convex),  because concavity  of the distortion function is equivalent to   convex order consistency.
 
The situation for IQD agents like Carole is more sophisticated.  The distortion function of IQD is discontinuous, non-concave, non-monotone, and takes value zero on both tails of the distribution. The preference induced by IQD does not correspond to a decision criterion typically considered in the literature,  whereas the preferences induced by quantiles, 
called quantile maximization, have been axiomatized by
 \cite{R10}. IQD is a standard measure of dispersion used in statistics such as in box plots, and its most popular special case in statistics is the inter-quartile difference, which measures the difference between the 25\% and 75\% quantiles of data.

The general theory of risk sharing between agents using distortion riskmetrics is laid out in Section \ref{sec:3}.
 A convenient feature of distortion riskmetrics is that they are convex order consistent if and only if the distortion function is concave (\citealp[Theorem 3]{WWW20a}). This enables the characterization of Pareto-optimal allocations for such agents using the comonotonic improvement, a notion introduced in \cite{LM94} to characterize the optimal sharing of risk among risk-averse expected utility maximizers; see also \cite{LM08} and \cite{R13}. 
Non-concave distortion functions lead to substantial challenges and to non-comonotonic optimal allocations,
with limited recent results obtained by \cite{ELW18} and \cite{W18} for some increasing distortion riskmetrics.

We study optimality within the subset of comonotonic allocations, which we refer to as the comonotonic risk sharing problem, for general distortion riskmetrics which are not necessarily convex in Section \ref{sec:4}. 
We show that 
the risk possibility set of distortion riskmetrics is always a convex set when restricted to the subset of comonotonic allocations. By the Hahn-Banach Theorem, we can always find comonotonic Pareto-optimal allocations by optimizing a linear combination of the agents' welfare. This simple but valuable result ``essentially comes for free” by the comonotonic additivity and positive homogeneity of distortion riskmetrics. In particular, it does not require the convexity of the evaluation functionals.
Moreover, this comonotonic setting allows us to easily incorporate heterogeneous beliefs as in the setting of \cite{ELMW20}, which we study in Section \ref{sec:7} for the interested reader.

With IQD agents, the set of optimal allocations can dramatically differ when defined on the whole set of allocations or the subset of comonotonic ones, as shown by results in Sections \ref{subsec:Unc_IQD} and \ref{sec:42}. 
We show the surprising result that Pareto-optimal allocations are precisely those which solve a sum optimality problem, which is not true for other variability measures such as GD or MMD.
Closed-form Pareto-optimal allocations are obtained, which can be decomposed as the sum of two pairwise counter-monotonic allocations. This observation complements the optimal allocations for quantile agents obtained by \cite{ELW18} which are pairwise counter-monotonic. 

Combining results obtained in Sections \ref{sec:3} and \ref{sec:4}, the general problem of sharing risks between IQD agents (like Carole) and agents with concave and symmetric distortion functions (like Anne and Bob) mentioned in the beginning of the paper is solved in Section \ref{sec:new5} and further illustrated in Section \ref{sec:5}.  We obtain a sum-optimal allocation 
which features a combination of comonotonicity and pairwise counter-monotonicity.  These two structures are, respectively, regarded as extremal positive and negative dependence concepts; see \cite{PW15} for an overview of these dependence concepts and \cite{LLW23} for a stochastic representation of pairwise counter-monotonicity. More specifically, there exists an event on which the obtained Pareto-optimal allocation is comonotonic, and two events on which the sum-optimal allocation is pairwise counter-monotonic. To the best of our knowledge, this is the first article to obtain such a type of sum-optimal or Pareto-optimal allocation. 
Moreover, none of our results relies on continuity of the distortion functions.
Section \ref{sec:7} extends our results on comonotonic risk sharing to the problem where agents have heterogeneous beliefs.
We conclude the paper in Section \ref{sec:8} with a few remarks, and all proofs are put in the appendices.

\section{Preliminaries}\label{sec:2}
\subsection{Distortion riskmetrics}
For a measurable space $(\Omega, \mathcal{F})$ and a finite set function $\nu:\mathcal F \rightarrow \mathbb{R}$ with $\nu(\varnothing)=0$, the signed \textit{Choquet integral} of a random variable $X:\Omega \rightarrow \mathbb{R}$ is the integral
\begin{align}
    \label{eq:0} \int X \d \nu =\int_0^\infty  \nu(X>x) \d x + \int_{-\infty}^{0} \left(\nu(X>x)- \nu (\Omega)\right)\d x,
\end{align}
provided these integrals are finite.
Let $n$ be a positive integer and let $[n]=\{1,\dots,n\}$.
 The random variables $X_1,\dots,X_n$ are \textit{comonotonic} if there exists a collection of increasing functions $f_i: \mathbb{R}\rightarrow \mathbb{R}$, $i\in [n]$, and a random variable $Z$ such that $X_i=f_i(Z)$ a.s. for all $i\in [n]$.
 Two random variables $X_1,X_2$ are \emph{counter-monotonic} if $X_1,-X_2$ are comonotonic. The random variables $X_1,\dots,X_n$ are \emph{pairwise counter-monotonic} if $X_i,X_j$ are counter-monotonic for each pair of distinct  $i,j$ \citep{PW15, LLW23}. Terms like ``increasing" or ``decreasing" are in the non-strict sense.

Assume that $(\Omega,\mathcal F,\p)$  is an atomless probability space where almost surely equal random variables are treated as identical. When $X(\omega)$ appears, it does not matter which version we choose.
Let $\X$ be a set of random variables on this space.
For simplicity, we assume throughout that $\X=L^\infty$,  the set of essentially bounded random variables, and we will inform the reader when a result can be extended to larger spaces.   
A \emph{distortion riskmetric}
 $\rho_h$ is the mapping from $\X$ to $\R$,
\begin{align}
\rho_h(X)=\int X \d\left( h\circ \P \right)=\int_0^\infty h(\P(X > x))\d x + \int_{-\infty}^{0} (h(\P(X > x))-h(1) )\d x,
\label{eq:1}
\end{align}
where  $h$ is in the set $ \H ^{\rm BV}$ of all possibly non-monotone \emph{distortion functions}, i.e., $$\H^{\rm BV}=\{h:  [0,1]\to\R \mid \mbox{$h$ is of bounded variation and }h(0)=0 \}.$$
Distortion riskmetrics are law-invariant versions of a general Choquet integral defined with regards to (possibly non-monotone) set functions; see Theorem 4.5 of \cite{MM04}.\footnote{To avoid any confusion, we refrain from using the term capacity, as those are typically defined as positive monotone set functions that are not necessarily additive. In fact, the set functions $h\circ \P$ we consider in the text need not be either positive or monotone for $h\in \H^{\rm BV}$.}

We now recall some properties of distortion riskmetrics that we use throughout. Any distortion riskmetric $\rho_h$ always satisfies the following four properties as a function $\rho:\X\to \R$.  
\begin{enumerate}
\item \emph{Law invariance:} $\rho(X)=\rho(Y)$ for $X\overset{\d}{=}Y$.
    \item \emph{Positive homogeneity:} $\rho(\lambda X)=\lambda \rho(X)$ for all $\lambda>0$ and $X\in\X$ with $\lambda X \in \X$.
    \item \emph{Comonotonic additivity:}   $\rho(X+Y)=\rho(X)+\rho(Y)$ whenever $X$ and $Y$ are comonotonic and $X + Y \in \X$.  
     \item[4.] \emph{Translation invariance:} $\rho(X+c)=\rho(X)+ c \rho(1)$ for all $c\in\R$ and $X\in\X$ with $X + c \in \X$. 
 \end{enumerate}
 As a special case of translation invariance with $\rho(1)=1$, 
  $\rho$ is \emph{cash additive} if $\rho(X+c)=\rho(X)+c$ for $c\in \R$ and $X\in \X$. 
For a distortion riskmetric $\rho_h$,
 cash additivity means $h(1)=1$. We also say \textit{location invariance} for $h(1)=0$ and \textit{reverse cash additivity} for $h(1)=-1$. 
 We note that although we use the general term ``cash additivity" as in the literature of risk measures, the values of random variables may be interpreted as non-monetary, such as carbon dioxide emissions, as long as they can be transferred between agents  in an additive fashion.

A distortion riskmetric $\rho_h$ may also satisfy the following properties depending on $h$.
A random variable $X$ is said to be dominated by a random variable $Y$ in \textit{convex order}, denoted by $X\le_{\rm{cx}}  Y$, if $\E[\phi(X)] \le \E[\phi(Y )]$ for every convex function $\phi: \R \to \R$, provided that both expectations exist (allowed to be infinite).  
\begin{enumerate}
\item[5.] \emph{Increasing monotonicity}: $\rho(X)\leq \rho(Y)$ whenever  $X\leq Y$.
\item[6.] \emph{Convex order consistency}: $\rho(X)\leq \rho(Y)$ whenever $X\le_{\rm{cx}} Y$.
    \item[7.] \emph{Subadditivity}: $\rho(X+Y)\leq \rho(X)+\rho(Y)$
 for every $X,Y \in \X$.
 \end{enumerate}
 We also say that $\rho$ is
\emph{monotone} if either $\rho$ or $-\rho$ is increasing.
  Increasing and cash-additive functionals are 
  \textit{monetary risk measure} (\citealp{FS16}) or  \textit{niveloids} \citep{CMMR14}; see also \cite{ADEH99} for coherent risk measures and \cite{CMMM11} for quasiconvex risk measures.  
 For a distortion riskmetric $\rho_h$, increasing monotonicity means that $h$ is increasing, and either subadditivity or   convex order consistency   is equivalent to the concavity of   $h$ by Theorem 3 of \cite{WWW20a}.

Distortion riskmetrics are precisely all law-invariant 
 and comonotonic-additive mappings satisfying some forms of continuity; see \cite{WWW20b} on $L^\infty$ and \cite{WWW20a} on general spaces.
The subset of increasing normalized distortion functions is denoted by $ \H^{\rm DT}$, that is,
 $$\H^{\rm DT}=\{h:  [0,1]\to\R  \mid \mbox{$h$ is increasing, $h(0)=0$ and $h(1)=1$} \}.$$ 
 
 If $h\in \H^{\rm DT}$, then $\rho_h$ is often called a \emph{dual utility} of \cite{Y87}.  
 Recall that a Yaari agent is strongly risk averse when the distortion function $h$ is concave \citep{Y87}. Hence, we slightly abuse nomenclature and simply say that a distortion riskmetric agent is risk averse when its distortion function is concave, regardless of whether it is increasing or not. This is consistent with the concept of increasing in risk introduced by \cite{RS70}.  With risk aversion, $\rho_h$ is a  \emph{spectral risk measure} \citep{A02} in risk management,  an important class of \emph{coherent} risk measures \citep{K01}. In insurance, it is also known as \textit{Wang's premium principle} \citep{W96}.

Any distortion riskmetric admits a quantile representation (Lemma 1 of \cite{WWW20a}; see the monotone case in Theorems 4 and 6 of \cite{DKLT12}). 
For a concise presentation of results,  
we define quantiles by counting losses \emph{from large to small}.\footnote{It will be clear from Theorem \ref{th:IQD} that this nontraditional choice of notation significantly simplifies the presentation of several results; this is also the case in \cite{ELW18}.} Formally, for $t\in [0,1]$, we define the left quantile of a random variable $X\in \X$ as 
$Q_t^-(X)=\inf\{x\in \R: \p(X\le x)\ge 1-t \}$,
and the  right quantile as
$ Q_t^+(X)=\inf\{x\in \R: \p(X\le x) > 1-t \}$, where $\inf \varnothing =\infty$.
Further let $\esssup=Q_0^-$ and $\essinf=Q_1^+$. 
The following integrals are in the sense of Lebesgue--Stieltjes.

\begin{lemma} \label{lem:qr}
For $h\in \H ^{\rm BV}$ and $X\in \X$ such that $\rho_h(X)$ is well-defined (it may take values $\pm \infty$),
    \begin{enumerate}[(i)] 
    \item if $h$ is right-continuous, then $\int X\d ( h\circ\P)=\int^1_0 Q_{t}^+(X)\d h(t)$;
    \item if $h$ is left-continuous, then $\int X\d ( h\circ\P)=\int^1_0 Q^{-}_{t}(X)\d h(t)$; 
    \item if $t\mapsto Q^-_{t}(X)$ is continuous on $(0,1)$, then $\int X\d (h\circ\P)=\int^1_0 Q^{-}_{t}(X)\d h(t)=\int^1_0 Q^+_{t}(X)\d h(t)$.
\end{enumerate}
\end{lemma}

There are two main advantages of working with non-monotone distortion functions.  First, as monotonicity is not assumed, results on maxima and minima are symmetric; we only need to analyze one of them. Second,  distortion riskmetrics include many more functionals in risk management, such as variability measures, which never have a monotone distortion function. Some properties of non-monotone risk functionals are studied by \cite{WW20}. We will make extensive use of three variability measures which appeared in the introduction. 
They are well defined on spaces larger than $L^\infty$, although we state our main results on $\X=L^\infty$.

The first measure of variability we use extensively is the Gini deviation  (GD)  $$\mathrm{GD}(X)=\frac{1}{2}\E [\vert X^*-X^{**}\vert]=\int X \d (h_{\mathrm{GD}}\circ \P)$$ for  $X \in L^1$, where  $X^*$, $X^{**}$ are any independent copies of $X$,  and $h_{\mathrm{GD}}(t)= t-t^2$ for $t\in[0,1]$. The specific choice of $X^*$, $X^{**}$ is not relevant. Its distortion function is depicted in Figure \ref{fig:distortion} (a).
As our second measure of variability, the mean-median deviation (MMD) is defined by $$\mathrm{MMD}(X)=\min_{x\in \mathbb R}\E [\vert X-x\vert ]=\E[\vert X-Q^-_{1/2}(X)\vert]=\int X \d (h_{\mathrm{MMD}} \circ \P)$$ for $X \in L^1$ and  $h_\mathrm{MMD}(t)=t \wedge (1-t)$, $t\in [0,1]$; see Figure \ref{fig:distortion} (b). The mean-median deviation is sometimes called the mean (or average) absolute deviation from the median and is well known for its statistical robustness.
Both the mean-median deviation and the Gini deviation have concave distortions and thus are convex order consistent.
Lastly, the inter-quantile difference (IQD) is defined by 
$$\mathrm{IQD}_\alpha(X) =  Q_\alpha^-(X) -  Q_{1-\alpha}^+(X)=\int X \d (h_{\mathrm{IQD}}\circ \P )$$
for $X \in L^0$ and $h_{\mathrm{IQD}}(t)=\id_{\{\alpha < t< 1-\alpha\}}$, $t \in [0,1]$ and $\alpha \in [0,1/2)$. See Figure \ref{fig:distortion} (c)  for its distortion function.
We further set $\mathrm{IQD}_\alpha=0$ for $\alpha \in [1/2,\infty)$, but this is only for the purpose of unifying the presentation of some results.
Our formulation of IQD is slightly different from the definition used by \cite{BFWW22} where $\mathrm{IQD}_\alpha$ is defined as 
$ Q_{\alpha}^+ -  Q_{1-\alpha}^-$, but this difference is minor. The two definitions coincide when the quantile function is continuous.
For $X\in \X$ and $\alpha \in [0,1/2)$, a convenient formula (see Theorem 1 of \cite{BFWW22})  is 
\begin{equation}
\label{eq:IQD-1}
\mathrm{IQD}_\alpha(X) = Q_\alpha^- (X)+  Q_\alpha^- (-X),
\end{equation}
and this is due to   $Q_{1-{\alpha}}^+(X) =-  Q_{\alpha}^-(-X)$.

Consider now a preference functional $\mathcal{I}$ of the form
$$\mathcal{I}(X)= \theta \E(X) + \gamma D(X)$$
for $\theta \geq 0$, $\gamma\in \R$ and $D(X)$ a variability measure. 
The version of $\mathcal{I}$ with $\theta =1$ and $\gamma<0$ is widely used in modern portfolio theory (as an objective to maximize). There, the random variable $X$ denotes the gains, the parameter $ \gamma$ indicates the degree of risk aversion and $D(X)$ is a variability measure chosen to replace the variance. This yields the ``Mean-$D$" preferences nomenclature common in the literature. The version of $\mathcal{I}$ with $X$ being a loss, $\theta \geq 1$ and $\gamma \ge 0$ is common in the insurance/reinsurance literature, where it is called a distortion-deviation premium principle. For instance, \cite{D90} suggests the premium principle $\theta=1$ and $D(X)=\mathrm{MMD}(X)$. The functional  $\mathcal{I}$ is a distortion riskmetric 
as long as $D$ is one,
and so we adopt the convention of denoting such functional by $\rho_h$ and interpreting $X$ as losses. 
Panels (d)-(f) of 
Figure \ref{fig:distortion} illustrate the distortion functions of $\E+\gamma D$.

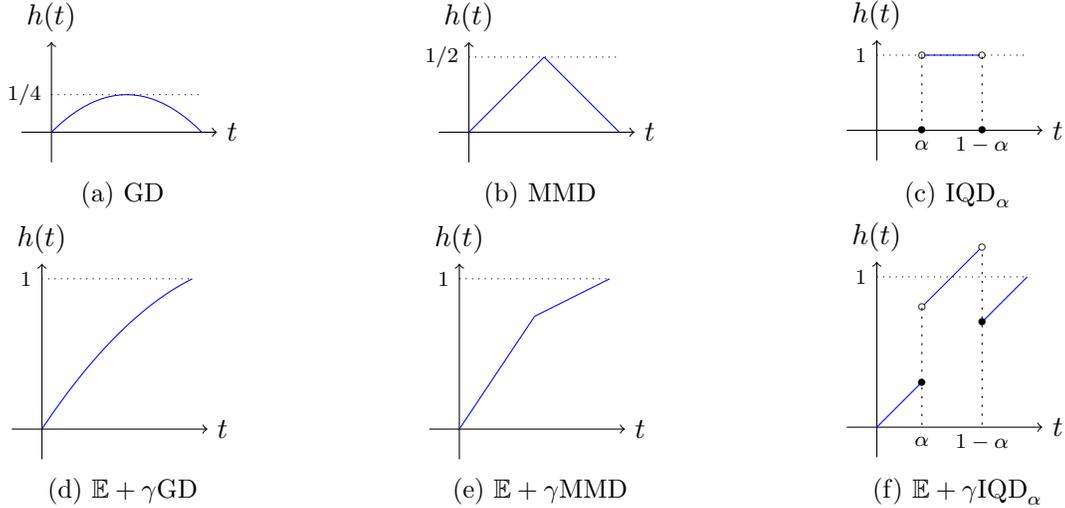
\begin{figure}[t]
\caption{Distortion functions for GD, MMD, IQD and $\E+\gamma D$, where $\gamma=1/2$}
\label{fig:distortion}
\begin{subfigure}[b]{0.33\textwidth}
\centering
\begin{tikzpicture}[scale=2]
\draw[->] (-0.2,0) --(1.1,0) node[right] {$t$};
\draw[->] (0,-0.2) --(0,0.6) node[above] {$h(t)$};
\draw[blue,domain =0:1] plot (\x ,{\x-\x*\x});
\node[left] at (0,0.25) {\scriptsize $1/4$};
 \draw[dotted] (0, 0.25)--(1,0.25);
\end{tikzpicture}
\caption{GD}
\end{subfigure}
\begin{subfigure}[b]{0.33\textwidth}
\centering
\begin{tikzpicture}[scale=2]
\draw[->] (-0.2,0) --(1.1,0) node[right] {$t$};
\draw[->] (0,-0.2) --(0,0.6) node[above] {$h(t)$};
 \draw[blue,domain =0:0.5] plot (\x ,{\x});
 \draw[blue,domain =0.5:1] plot (\x ,{1-\x});
\node[left] at (0,0.5) {\scriptsize $1/2$};
 \draw[dotted] (0, 0.5)--(1,0.5);
\end{tikzpicture}
\caption{MMD}
\end{subfigure}
\begin{subfigure}[b]{0.33\textwidth}
\centering
\begin{tikzpicture}[scale=2]
\draw[->] (-0.2,0) --(1.1,0) node[right] {$t$};
\draw[->] (0,-0.2) --(0,0.6) node[above] {$h(t)$};
 \draw[blue,domain =0.315:0.685] plot (\x ,{0.5});
 \draw[dash pattern={on 0.84pt off 2.51pt}] (0.3, 0.49)--(0.3,0);
  \draw[dash pattern={on 0.84pt off 2.51pt}]  (0.7, 0.49)--(0.7,0);
 
\node at (0.3, 0) {\tiny\textbullet};
\node at (0.7, 0) {\tiny\textbullet};
\node at (0.3, 0.5) {\tiny$\circ$};
\node at (0.7, 0.5) {\tiny$\circ$};
\node[below] at (0.3,0) {\scriptsize $\alpha$};
\node[below] at (0.7,0.02) {\scriptsize $1-\alpha$};
\node[left] at (0,0.5) {\scriptsize $1$};
 \draw[dotted] (0, 0.5)--(1,0.5);
\end{tikzpicture}
\caption{ $\mathrm{IQD}_\alpha$}
\end{subfigure}\\
\begin{subfigure}[b]{0.33\textwidth}
\centering
\begin{tikzpicture}[scale=2]
\draw[->] (-0.2,0) --(1.1,0) node[right] {$t$};
\draw[->] (0,-0.2) --(0,1.1) node[above] {$h(t)$};
\draw[blue,domain =0:1] plot (\x ,{\x+0.5*(\x-\x*\x)});
\node[left] at (0,1) {\scriptsize $1$};
 \draw[dotted] (0, 1)--(1,1);
\end{tikzpicture}
\caption{$\E+\gamma\mathrm{GD}$}
\end{subfigure}
\begin{subfigure}[b]{0.33\textwidth}
\centering
\begin{tikzpicture}[scale=2]
\draw[->] (-0.2,0) --(1.1,0) node[right] {$t$};
\draw[->] (0,-0.2) --(0,1.1) node[above] {$h(t)$};
 \draw[blue,domain =0:0.5] plot (\x ,{1.5*\x});
 \draw[blue,domain =0.5:1] plot (\x ,{0.5*(\x-0.5)+0.75});
\node[left] at (0,1) {\scriptsize $1$};
 \draw[dotted] (0, 1)--(1,1);
\end{tikzpicture}
\caption{ $\E+\gamma \mathrm{MMD}$}
\end{subfigure}
\begin{subfigure}[b]{0.33\textwidth}
\centering
\begin{tikzpicture}[scale=2]
\draw[->] (-0.2,0) --(1.1,0) node[right] {$t$};
\draw[->] (0,-0.2) --(0,1.1) node[above] {$h(t)$};
 \draw[blue,domain =0:0.3] plot (\x ,{\x});
 \draw[blue,domain =0.315:0.685] plot (\x ,{\x+0.5});
  \draw[blue,domain =0.7:1] plot (\x ,{\x});
 \draw[dash pattern={on 0.84pt off 2.51pt}](0.3, 0.795)--(0.3,0);
 \draw[dash pattern={on 0.84pt off 2.51pt}] (0.7, 1.15)--(0.7,0);
\node at (0.3, 0.3) {\tiny\textbullet};
\node at (0.7, 0.7) {\tiny\textbullet};
\node at (0.3, 0.8) {\tiny$\circ$};
\node at (0.7, 1.2) {\tiny$\circ$};
 
\node[below] at (0.3,0) {\scriptsize $\alpha$};
\node[below] at (0.7,0.02) {\scriptsize $1-\alpha$};
\node[left] at (0,1) {\scriptsize $1$};
 \draw[dotted] (0, 1)--(1,1);
\end{tikzpicture}
\caption{$\E+\gamma \mathrm{IQD}_\alpha$}
\end{subfigure}
\end{figure}

\subsection{Risk sharing problems}

There are $n$ agents sharing a total loss $X\in \X $.
Recall that, for all results,  we consider $\X=L^\infty$. 
Agents can have different preferences, which may be due to their own statistical modelling; see e.g., \cite{A22}.
Suppose that agent $i\in [n]$ has a preference modelled by a distortion riskmetric $\rho_{h_i}$  with smaller values preferred. Given $X\in \mathcal{X}$, we define the set of \emph{allocations} of $X$ as
\begin{align}
\mathbb{A}_n(X)=\left\{(X_1,\ldots,X_n)\in \mathcal{X}^n: \sum_{i=1}^nX_i=X\right\}. \label{eq:intro1}
\end{align} 
The \emph{inf-convolution}  $\dsquare_{i=1}^n \rho_{h_i}$ of $n$ distortion riskmetrics $\rho_{h_1},\dots,\rho_{h_n}$ is defined as
\begin{align*}
\dsquare_{{i=1}}^n\rho_{h_i}(X):=\inf\left\{\sum_{i=1}^n\rho_{h_i}(X_i): (X_1,\dots,X_n)\in \mathbb{A}_n(X) \right\},~~X\in \X.
\end{align*}
That is, the inf-convolution of $n$ distortion riskmetrics is the infimum over aggregate welfare for all possible allocations. 
For a general treatment of inf-convolution in risk sharing problems, see \cite{R13}. 
 
  Let  $\rho_{h_1},\dots,\rho_{h_n}$ be distortion riskmetrics and $X\in \X$. The allocation $(X_1,\dots,X_n)$ is \textit{sum optimal  in $\mathbb{A}_n(X)$} if $\dsquare_{{i=1}}^n\rho_{h_i}(X)=\sum_{i=1}^n\rho_{h_i}(X_i) $.
An allocation $(X_1,\ldots,X_n)\in\mathbb{A}_n(X)$ is     \emph{Pareto optimal  in $\mathbb{A}_n(X)$} if for any
$(Y_1,\ldots,Y_n)\in\mathbb{A}_n(X)$ satisfying $\rho_{h_i}(Y_i)\leq\rho_{h_i}(X_i)$ for all $i\in [n]$, we have $\rho_{h_i}(Y_i)=\rho_{h_i}(X_i)$ for all $i\in [n]$. Note that sum-optimal allocations are always Pareto optimal.

Part of our analysis is conducted for the constrained problem where 
the allocations are confined to 
the set of comonotonic allocations, that is,
 $$\mathbb{A}_n^+(X)=\left\{(X_1,\ldots,X_n)\in\mathbb{A}_n(X):~ X_1, \dots, X_n \mbox{ are comonotonic}\right\},
$$

We first make a useful observation about $\mathbb{A}_n^+(X)$ below, which is a simple generalization of \citet[Proposition 4.5]{D94} in the case of $n=2$. 
\begin{proposition}\label{pr:den}
The random variables $X_1,\dots,X_n$ are comonotonic if and only if  there exist   increasing functions $f_i:\R \rightarrow \R$ such that $X_i=f_i(\sum_{k=1}^n X_k)$ a.s., $i\in [n]$   and $\sum_{i=1}^n f_i(x)=x$ for $x\in \R$. 
\end{proposition}
Proposition \ref{pr:den} implies that if $(X_1, \dots, X_n) \in \mathbb{A}_n^+(X)$, then we can set $X=Z$ in the definition of comonotonicity while guaranteeing that $\sum_{i=1}^n X_i=X$ a.s.

The \emph{comonotonic inf-convolution} $\dboxplus_{i=1}^n \rho_{h_i}$ of distortion riskmetrics $\rho_{h_1},\dots,\rho_{h_n}$ is defined as
\begin{align*}
\dboxplus_{{i=1}}^n\rho_{h_i}(X):=\inf\left\{\sum_{i=1}^n\rho_{h_i}(X_i): (X_1,\dots,X_n)\in \mathbb{A}_n^+(X) \right\}.
\end{align*}

  Let  $\rho_{h_1},\dots,\rho_{h_n}$ be distortion riskmetrics and $X\in \X$. 
An allocation $(X_1,\dots,X_n)$ is \textit{sum optimal in} $\mathbb{A}_n^+(X)$ when $\dboxplus_{{i=1}}^n\rho_{h_i}(X)=\sum_{i=1}^n\rho_{h_i}(X_i) $.
An allocation $(X_1,\ldots,X_n)\in\mathbb{A}^+_n(X)$ is   \emph{Pareto optimal   in $\mathbb{A}^+_n(X)$} if for any
$(Y_1,\ldots,Y_n)\in\mathbb{A}^+_n(X)$ satisfying $\rho_{h_i}(Y_i)\leq\rho_{h_i}(X_i)$ for all $i\in [n]$, we have $\rho_{h_i}(Y_i)=\rho_{h_i}(X_i)$ for all $i\in [n]$. 
  
 The set of comonotonic allocations $\mathbb{A}_n^+(X)$ is a strict subset of the set of all possible allocations $\mathbb{A}_n(X)$. Hence, the sequel refers to the problem of sharing risk in $\mathbb{A}_n(X)$ and $\mathbb{A}_n^+(X)$ as \textit{unconstrained} and  \textit{comonotonic} risk sharing, respectively. 

\section{Unconstrained risk sharing}\label{sec:3}
This section tackles the unconstrained risk sharing problem. It is divided into two subsections. The first aims at providing general results and  
 the second subsection characterizes the unconstrained risk sharing problem with IQD agents. There, we show that sum-optimal allocations involve pairwise counter-monotonicity, an extreme form of negative dependence between the agents' risk.
\subsection{Pareto optimality, sum optimality, and comonotonic improvement}\label{sec:31}
In all results, we will always assume that agents have preferences modelled by $\rho_{h_1},\dots,\rho_{h_n}$ where $h_1,\dots,h_n\in \mathcal H^{\rm BV}$, with one exception which will be specified clearly. The value of $h(1)$ is important for a distortion riskmetric $\rho_h$ because, by translation invariance, it pins down the value attributed to a sure gain or loss.
 \begin{proposition}\label{prop:1}~ Let $X\in \X$. Then
\begin{enumerate}[(i)]
\item If a Pareto-optimal allocation in either $\mathbb A^+_n(X)$ or $\mathbb A_n(X)$ exists then  
  $h_i(1)$, $i\in [n]$, are all $0$, all positive, or all negative;
\item If $\dboxplus_{i=1}^n \rho_{h_i}(X)> -\infty$, then $h_1(1)=\dots=h_n(1)$.
  \end{enumerate}
 \end{proposition} 
The proof of Proposition \ref{prop:1} highlights the role of translation invariance. 
For (i), we thus assume by contradiction that $(X_1,\dots, X_n)$ is Pareto optimal but that $h_i(1)$, $i\in [n]$, are not all zero or all of the same sign. We can organize a (cash) transfer $(c_1,\dots, c_n)$ between agents such that $\sum_{i=1}^n c_i = 0$ and the allocation $(X_1+ c_1, \dots, X_n+ c_n)$ strictly improves upon $(X_1,\dots,X_n)$;  this should not be possible.
This condition implies that, in order for the risk-sharing problem to be meaningful,
all agents must agree on whether they like or dislike
an increase of their allocation.
In the former case, $X_1,\dots,X_n$ may represent a good like monetary gains, and in the latter case,  they may represent bad outcomes, like carbon dioxide emissions.
For (ii), when the value of $h(1)$ differs between agents, a similar type of transfer strictly reduces the sum of welfare $\sum_{i=1}^n\rho_{h_i}$, and so the value attained by the inf-convolution $\dboxplus_{i=1}^n \rho_{h_i}$ is arbitrarily small.

For $h\in \mathcal H^{\rm BV}$, we write $\tilde h = h/|h(1)|$ if $h(1)\ne 0$ and $\tilde h = h$ if $h(1)=0$. If $h(1)\ne 0$, then   $\tilde h(1)=\pm 1$. 
Note that replacing $h_i$ with its normalized version $\tilde h_i$ does not change the preference of agent $i$. Hence, we sometimes consider in our proofs distortion riskmetrics that are either all cash additive or all reverse cash additive. While this normalization does change the value attained by the inf-convolution, it is without loss of generality for characterizing Pareto optimality.

We now state our first theorem, a generalization of Proposition 1 of \cite{ELW18}  stated for monetary risk measures.
\begin{theorem}\label{th:1}
Suppose that   $h_i(1) \ne 0$ for some $i\in [n]$. 
  An allocation
 $(X_1,\dots,X_n)\in \mathbb A_n(X)$ is  Pareto optimal in $\mathbb A_n(X)$ if and only if   $\sum_{i=1}^n \rho_{\tilde h_i}(X_i) =\dsquare_{{i=1}}^n\rho_{\tilde h_i}(X)$.
\end{theorem}
 Theorem \ref{th:1} states that  Pareto optimality and sum optimality can be translated into each other via normalization whenever the distortion riskmetrics are not location invariant. 
 The picture for location-invariant distortion riskmetrics is, however, drastically different, as we only have one direction. 
 The next statement considers this setting. Its proof is straightforward and thus omitted.
\begin{proposition}\label{pr:1}
Suppose that $h_i(1)=0$ for all $i\in [n]$.
For an allocation $(X_1,\dots,X_n)\in \mathbb A_n(X)$, it holds that (i)$\Rightarrow$(ii):
\begin{enumerate}[(i)]
\item   $\sum_{i=1}^n \lambda_i  \rho_{h_i}(X_i) =\dsquare_{{i=1}}^n(\lambda_i \rho_{h_i})(X)$ for some $(\lambda_1,\dots,\lambda_n)\in (0,\infty)^n$;
\item $(X_1,\dots,X_n)$ is  Pareto optimal in $\mathbb A_n(X)$.
\end{enumerate}
\end{proposition}
The weights $(\lambda_1, \dots, \lambda_n)$ in (i) are often called \emph{Negishi weights} \citep{N60}. One might be interested in the converse statement of Proposition \ref{pr:1}, asking whether the Pareto optimality of $(X_1,\dots,X_n)$ implies the existence of a set of $(\lambda_1,\dots,\lambda_n)\in [0,\infty)^n\setminus \{\mathbf 0\}$ such that $\sum_{i=1}^n \lambda_i  \rho_{h_i}(X_i) =\dsquare_{{i=1}}^n(\lambda_i \rho_{h_i})(X)$. We see in this paper that this claim holds in three cases: when agents have $h_i(1)>0$ or $h_i(1)<0$ (Theorem \ref{th:1});  when all agents are IQD (Theorem \ref{th:IQD}); when they have concave distortion functions (a combination of Propositions \ref{pr:equivalence} and \ref{pr:2}). However, we do not know whether this holds true for general distortion functions with $h_1(1)=\dots=h_n(1)=0$; see also the discussion after Proposition \ref{pr:2}.

In view of Proposition \ref{pr:1}, we say that an allocation $(X_1,\dots,X_n)$ of $X$ is $\boldsymbol \lambda$-optimal if 
 \begin{align}\label{eq:sum-opt-lam}
\dsquare_{{i=1}}^n\rho_{\lambda_i h_i}(X)=\sum_{i=1}^n\rho_{ \lambda_i h_i}(X_i).
\end{align}
where $\boldsymbol \lambda= (\lambda_1,\dots,\lambda_n)$.
Clearly,  $\boldsymbol \lambda$-optimality is equivalent to sum optimality 
when the preferences are specified as  $(\lambda_1\rho_{h_1},\dots,\lambda_n\rho_{h_n})$,
and conversely, sum optimality is $(1,\dots,1)$-optimality. Therefore, we encounter no additional technical complications when solving either of them.

The following result follows from the well-known result of comonotonic improvement of \cite{LM94}\footnote{See \cite{R13} for this result on general spaces.} and the fact that distortion riskmetrics are convex order consistent when the distortion functions $h_i$ are concave (Theorem 3 of \cite{WWW20a}).
A \emph{comonotonic improvement} 
of $(X_1,\dots,X_n)\in \mathbb A_n(X)$ is a random vector $(Y_1,\dots,Y_n)\in \mathbb A^+_n(X)$ such that $Y_i\le_{\rm cx} X_i$ for all $i\in [n]$. 
Such a comonotonic improvement always exists for any $(X_1,\dots,X_n)$.
\begin{proposition}\label{pr:equivalence}
Suppose that $h_1,\dots,h_n$ are concave. It holds that
$\dsquare_{{i=1}}^n \rho_{h_i}=\dboxplus_{{i=1}}^n \rho_{h_i}$.
 Moreover,  for any $X\in \X$, if there exists a  Pareto-optimal allocation in $\mathbb A_n(X)$, then there exists a comonotonic Pareto-optimal allocation in $\mathbb A_n(X)$. 
\end{proposition}  

Next, we prove that if $h_1,\dots,h_n$ are strictly concave, then the set of optimal allocations in  $\mathbb A_n(X)$ is exactly that of those in  $\mathbb A^+_n(X)$. This is because comonotonic improvements lead to a strict increase in welfare when the probability distortions $h_i$ are strictly concave. We state this result formally in Corollary \ref{cor:rw1} as a consequence of the following ancillary lemma:

\begin{lemma}\label{lem:rw1}
For two random variables $X,Y\in \X$, the following are equivalent:
\begin{enumerate}[(i)]
\item $X\laweq Y$;
\item  $\rho_h(X) = \rho_h(Y)$ for all concave $h\in \H^{\rm BV}$;
\item  $\rho_h(X)\le \rho_h(Y)$ for all concave $h\in \H^{\rm BV}$, in which the 
equality holds for a strictly concave $h$;
\item $X\le_{\rm cx} Y$ and $\rho_h(X)=\rho_h(Y)$ for a strictly concave $h \in \H ^{\rm BV}$.  
\end{enumerate} 
\end{lemma}

\begin{corollary}\label{cor:rw1}
If $X\le_{\rm cx} Y$ and  $X\not\laweq Y$, then $\rho_h(X)<\rho_h(Y)$ for any strictly concave $h$.
\end{corollary}

\begin{remark}
The equivalence in Lemma \ref{lem:rw1} holds true for any random variables $X,Y$ with finite mean, by requiring  that $\rho_h(X)$ and $\rho_h(Y)$ are finite for  the strictly concave function $h$ in (iii) and (iv). This follows by noting that we did not use the boundedness  of $X$ and $Y$ in the proof. 
\end{remark}

\begin{proposition}\label{pr:como}
Suppose that $h_1,\dots,h_n$ are strictly concave and $X\in \X$.
\begin{enumerate}[(i)]
\item Every Pareto-optimal allocation in $\mathbb A_n(X)$ is comonotonic.
\item If for each $i\in [n]$, $h_i= a _i h_1$ for some $a_i>0$  then an allocation   is Pareto optimal  in   $\mathbb A_n(X)$ if and only if it is comonotonic. 
\end{enumerate} 
\end{proposition} 

Both Proposition \ref{pr:equivalence}  and \ref{pr:como}  can easily be extended to $L^p$ for $p\ge 1$ instead of $\X =L^\infty$ provided that every $\rho_{h_i}$ is finite when defined on $L^p$.\footnote{Conditions for the finiteness of $\rho_h$ on $L^p$ are provided in Proposition 1 of \cite{WWW20a}.}

\subsection{IQD agents and negatively dependent optimal allocations
}\label{subsec:Unc_IQD}

We characterize the sum-optimal allocations on general spaces when agents evaluate their risk with the IQD measure of variability. We start with the problem of sharing risk among $\mathrm{IQD}$ agents only. 
In this setting, agent $i\in[n]$ has $\mathrm{IQD}_{\alpha_i}$ as their preference where $\alpha_i \in[0,1/2)$. 

For a random variable $X$ on the probability space $(\Omega,\mathcal F,\p)$, 
we   define tail events as in \cite{WZ21}.
For $\beta \in [0,1]$, we say that an event $A\in \mathcal{F}$ is a \textit{right (resp.~left) $\beta$-tail event} of $X$ if $\P(A)=\beta$ and $X(\omega)\geq X(\omega')$ (resp.~$X(\omega)\le   X(\omega')$) holds for a.s.~all $\omega \in A$ and $\omega'\in A^c$, where $A^c$ stands for the complement of $A$.  One could also follow the definition of  low tail-event in \cite{JST08} to define the $\beta$-tail event; that is, $A$ a right (resp.~left) $\beta$-tail event if $\p(A)=\beta$ and $\essinf_A(X)\ge \esssup_{A^c}(X)$ (resp.~$\esssup_A(X)\le \essinf_{A^c}(X)$), where $\essinf_A(X) = \sup\{m \in \R : \p(X \ge m \mid A) = 1\}$ and $\esssup_A(X) = \inf\{m \in \R : \p(X \le m \mid B) = 1\}$. The two definitions are equivalent.

\begin{theorem}\label{th:IQD}
For $X\in \X$ and the IQD risk sharing problem in $\mathbb A_n(X)$ with $\alpha_1,\dots,\alpha_n\in [0,1/2)$, let $\alpha =\sum_{i=1}^n \alpha_i$.\begin{enumerate}[(i)]
\item An allocation of $X$ 
 is Pareto optimal if and only if  
it is sum optimal.
\item For  $\lambda_1,\dots,\lambda_n \ge 0$ and  $\lambda=\bigwedge_{i=1}^n \lambda_i$,
\begin{equation}
\label{eq:IQD} \dsquare_{i=1}^n (\lambda_i \mathrm{IQD}_{\alpha_i}) =  \left(\bigwedge_{i=1}^n \lambda_i \right) \mathrm{IQD}_{\sum_{i=1}^n \alpha_i}=  \lambda \mathrm{IQD}_{ \alpha}. \end{equation}
In particular,   $
\dsquare_{i=1}^n \mathrm{IQD}_{\alpha_i} = \mathrm{IQD}_\alpha
 $. 
\item  A class of Pareto-optimal allocations  of $X\in\X$ for IQD agents is given by
\begin{align} \label{eq:IQD_allocation}
    X_i=  (X-c)\id_{A_i\cup B_i}+ a_i (X-c)\left(1-\id_{A\cup B}\right)+c_i,~~~~i\in [n],
\end{align}
where, by setting $\beta =\alpha \wedge (1/2)$,
\begin{enumerate}[(a)]
\item  $\left \{A_i\right\}_{i=1}^n$ and $\left \{B_i\right\}_{i=1}^n$  are partitions of  a right $\beta$-tail event $A$ and a left $\beta$-tail event $B$ of $X$ with $A,B$ disjoint, respectively,  satisfying $\P(A_i)=\P(B_i)=\alpha_i \beta/\alpha$, $i \in [n]$; 
\item $a_i\ge 0$ for $i\in[n]$ and  $\sum_{i=1}^na_i=1$;
\item $c\in [Q^-_{1/2}(X),Q^+_{1/2}(X)]$ and $\sum_{i=1}^n c_i=c$.
 \end{enumerate} 
 \end{enumerate} 
\end{theorem}

\begin{remark}\label{rem:IQD}
The allocation \eqref{eq:IQD_allocation} satisfies 
  $\sum_{i=1}^n  \lambda_i \mathrm{IQD}_{\alpha_i}(X_i) =\dsquare_{i=1}^n (\lambda_i \mathrm{IQD}_{\alpha_i})(X)$ by setting
 $a_i=0$ for   $i\in[n]   $ such that $ \lambda_i>\lambda$.
 \end{remark}

 The surprising ingredient of Theorem \ref{th:IQD}, part (i) is that, for IQD agents,
 sum optimality is indeed equivalent to Pareto optimality, which is the case for cash-additive distortion riskmetrics (Theorem \ref{th:1}).
 However,   for general agents with $h_1(1)=\dots=h_n(1)=0$,
 Pareto optimality is not necessarily equivalent to sum optimality, because different choices of $(\lambda_1,\dots,\lambda_n)$ in Proposition \ref{pr:1} lead to
 different Pareto-optimal allocations, which are not necessarily sum optimal (see Proposition \ref{pr:1} as well as Section \ref{sec:5}).
 As a consequence of this result, 
  Pareto-optimal allocations for IQD agents are precisely 
those  for agents using the mean-risk preferences with risk measured by IQD,
 $$
 \rho_{h_i}(X_i)= \E[X_i] +\mathrm{IQD}_{\alpha_i}(X_i),~~~i\in[n],
 $$ 
 because both solve the same sum optimality problem by noting that $\sum_{i=1}^n \E[X_i] =\E[X] $ for any allocation $(X_1,\dots,X_n)$ of $X$.

In part (ii) of Theorem \ref{th:IQD}, we see that 
 the inf-convolution of several IQD is an IQD.
Related to this observation, \cite{ELW18} showed that the inf-convolution of several quantiles is again a quantile.

\begin{figure}[t]
\begin{center}
\caption{A Pareto-optimal allocation in \eqref{eq:IQD_allocation}, where 
the shaded area represents  the allocation to agent $1$ minus $c_1$, that is, $X_1-c_1= (X-c)\id_{A_i\cup B_i}+ a_i (X-c) \id_{(A\cup B^c)}$
} \label{Fig:Pairwise_counter-monotonicity}
\vspace{-0cm}
\begin{tikzpicture}[scale=1]

\draw [black] (0,0) -- (10,0);
   \node at (0.5,0.4)   { \footnotesize $B_1$}  ;         
      \node at (1.5,0.4)   {\footnotesize  $\dots$}  ;       
      \node at (1.5,1)   { \footnotesize $B$}  ;    
      \node at (2.5,0.4)   { \footnotesize $B_n$}  ;      
      \node at (9.5,-0.4)    {\footnotesize $A_1$}; 
      \node at (8.5,-0.4)    {\footnotesize $\dots$};
      \node at (8.5,-1)   { \footnotesize $A$}  ;      
      \node at (7.5,-0.4)    {\footnotesize $A_n$}; 
\node at (10.2,-0.3) { $\omega$};
\node[left] at (9.8,4.6) {  $X(\omega)$}; 

\fill [top color=gray!40, bottom color=gray!20] (3,-0.4)--(5,0)--(3,0)--(3,-0.4);
\fill [top color=gray!40, bottom color=gray!20] (7,0.4)--(5,0)--(7,0)--(7,0.4);

\fill [top color=gray!40, bottom color=gray!20] (0,0)--(1,0)--(1,-2.55)-- (0.9,-2.65) -- (0.8,-2.75)  -- (0.7,-2.86) -- (0,-4);

\fill [top color=gray!40, bottom color=gray!20] (10,0)--(9,0)--(9,2.9)-- (9.1,3.02) -- (9.2,3.20)  -- (9.3,3.36) -- (9.6,3.98)  -- (10,5);

\draw [black] (0,-4) .. controls (1,-2.3) .. (4,-0.6);
\draw [black] (6,0.6) .. controls (9,2.5) .. (10,5);
 \path[<-, draw,  thick, dashed] (6.5,0.12)
      to[out=270, in=180](7.5,-4) 
      node[right] { $a_1 (X-c)\id_{(A\cup B)^c}$};
 \path[<-, draw,  thick, dashed] (3.5,-0.12)
      to[out=270, in=180](7.5,-4)  ;

 \path[<-, draw,  thick, dashed] (0.5,-2.7)
      to[out=270, in=180](2.5,-4) 
      node[right] { $  (X-c) \id_{B_1}$};

 \path[<-, draw,  thick, dashed] (9.5,3.3)
      to[out=180, in=0](6.5,3.7) 
      node[left] { $  (X-c) \id_{A_1}$};

\draw[gray,dashed] (1,0) -- (1,-2.55);
\draw[gray,dashed] (2,0) -- (2,-1.75);
\draw[gray,dashed] (3,0) -- (3,-1.2);  
\draw[gray,dashed] (7,0) -- (7,1.2);
\draw[gray,dashed] (8,0) -- (8,1.8); 
\draw[gray,dashed] (9,0) -- (9,2.9);

\draw [black] (4,-0.6) to (6,0.6); 
\draw [dashed] (3,-0.8) to (7,0.8); 
\draw [dashed] (3,-0.4) to (7,0.4); 
\draw [decorate,  decoration = {brace,    raise=1pt,    amplitude=5pt}] (0,0) --  (1,0);
\draw [decorate,  decoration = {brace,    raise=1pt,    amplitude=5pt}] (1,0) --  (2,0);
\draw [decorate,  decoration = {brace,    raise=1pt,    amplitude=5pt}] (2,0) --  (3,0); 
\draw [decorate,  decoration = {brace,    raise=1pt,    amplitude=5pt}] (8,0) --  (7,0);
\draw [decorate,  decoration = {brace,    raise=1pt,    amplitude=5pt}] (9,0) --  (8,0);
\draw [decorate,  decoration = {brace,    raise=1pt,    amplitude=5pt}] (10,0) --  (9,0); 
\draw [decorate,  decoration = {brace,    raise=1pt,    amplitude=5pt}] (10,-0.55) --  (7,-0.55); 
\draw [decorate,  decoration = {brace,    raise=1pt,    amplitude=5pt}] (0,0.55) --  (3,0.55);

\draw[loosely dashed] (0,-2.55) -- (10,-2.55); 
\node[left] at (0,-2.55)  {\footnotesize $Q^-_{1-\alpha_1}(X)$}; 
\draw[loosely dashed] (0,-1.2) -- (10,-1.2); 
\node[left] at (0,-1.2)  {\footnotesize $Q^-_{1-\alpha}(X)$}; 
\draw[loosely dashed] (0,2.9) -- (10,2.9); 
\node[left] at (0,2.9)  {\footnotesize $Q^-_{\alpha_1}(X)$}; 
\draw[loosely dashed] (0,1.2) -- (10,1.2);
\node[left] at (0,1.2)  {\footnotesize $Q^-_{\alpha}(X)$}; 
\node[left] at (0,0)  {\footnotesize $c=Q^-_{1/2}(X)$}; 

\end{tikzpicture}
\end{center}
\end{figure}
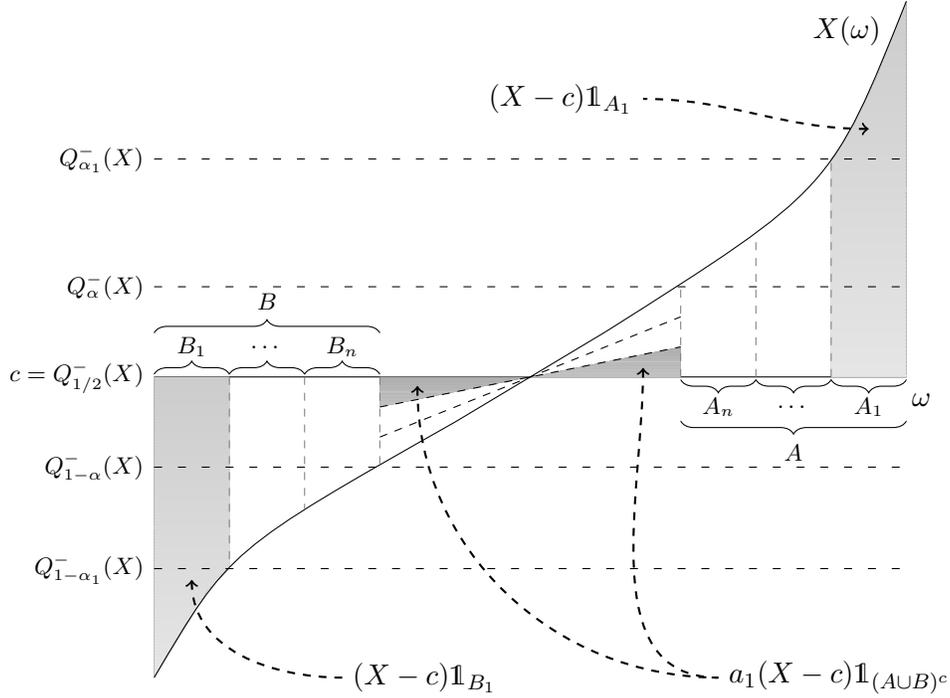

Figure \ref{Fig:Pairwise_counter-monotonicity} illustrates an example of the Pareto-optimal allocation \eqref{eq:IQD_allocation} in Theorem \ref{th:IQD}, part (iii). 
The dependence structure of this allocation warrants some further discussion. 
Without loss of generality, assume $c_1=\dots=c_n=0$   (this implies that a median of $X$ is $c=0$), and assume that $X$ is continuously distributed. 
Note that  (a.s.)~$X > 0$ if event $A$ occurs and $X< 0$ if event $B$ occurs.

First, suppose $\alpha\ge 1/2$ so that $\p({A\cup B})=1$.  
In this case, we have $X_i=X\id_{A_i\cup B_i}$ for $i\in [n]$.
The random vector $(X\id_{A_i}, X \id_{A_j})$  for $i\ne j$
is counter-monotonic because $A_i\cap A_j=\varnothing$ and $X>0$ on $A$.
This implies $(X\id_{A_1},\dots,X\id_{A_n})$  is  \textit{pairwise counter-monotonic}.
From the above analysis,  we can see that 
conditional on $A$, $(X_1,\dots,X_n)$ is pairwise counter-monotonic, 
and so is it conditional on $B$; that is $(X_1,\dots,X_n)$ is a mixture of two 
pairwise counter-monotonic vectors. 
Moreover, $(X_1,\dots,X_n)$ is also
the sum of the two pairwise counter-monotonic vectors $(X\id_{A_1},\dots,X\id_{A_n})$ and $(X\id_{B_1},\dots,X\id_{B_n})$.
We can check that 
$(X_i(\omega)-X_j(\omega))(X_i(\omega')-X_j(\omega'))< 0$ a.s.
for $\omega\in A_i$ and $\omega' \in A_j$, and 
$(X_i(\omega)-X_j(\omega))(X_i(\omega')-X_j(\omega'))> 0$ a.s.
for $\omega\in A_i$ and $\omega' \in B_j$.
Therefore, 
the allocation $(X_1,\dots,X_n)$ is not comonotonic, yet it is not pairwise counter-monotonic either.
This is illustrated by the ``vertical slicing" in Figure   \ref{Fig:Pairwise_counter-monotonicity}, where on $A$ and $B$ pairwise counter-monotonicity holds.  

As discussed above, 
we can describe $(X_1,\dots,X_n)$ as either the sum or the mixture of two pairwise counter-monotonic vectors.
 Pairwise counter-monotonicity is a form of extreme negative dependence that extends the concept of counter-monotonicity to the case of $n\geq 3$ agents; see \cite{PW15} and \cite{LLW23} for more details.  
 This observation is in contrast to the optimal allocations for quantile agents in Theorem 1 of \cite{ELW18}, which are indeed  pairwise counter-monotonic.

 If $0<\alpha<1/2$, then the term $a_i X\id_{(A\cup B)^c}$ appears in the allocation of every agent. 
 Note that conditional on  $(A\cup B)^c$,
 $(X_1,\dots,X_n)$ becomes comonotonic. 
This is illustrated by    ``proportional slicing" in the middle part of Figure \ref{Fig:Pairwise_counter-monotonicity}.
  This local comonotonicity will become crucial in  Section \ref{sec:new5}, where we  
 study the risk sharing problem among several IQD agents and risk-averse agents. 
 The lack of global comonotonicity but having some local comonotonicity is similar to the shape of optimal allocations obtained by \cite{Lieb24}, but in the latter setting this phenomenon is generated by the heterogeneity of the reference measures. A similar pattern appears in \cite{W18}.

 As hinted by Propositions \ref{pr:equivalence} and \ref{pr:como}, solving Pareto-optimal allocations for risk-averse agents requires us to study comonotonic risk sharing, which is the topic of the next section.

\section{Comonotonic risk sharing}\label{sec:4}
We now turn to the important case of comonotonic risk sharing. As before, we first provide theoretical results and then proceed to analyze further the special case of sharing risks with $\mathrm{IQD}$ agents.

\subsection{Pareto optimality, sum optimality, and explicit allocations}\label{sec:41}
The next result is similar to Theorem \ref{th:1}, but for comonotonic risk sharing. We omit its proof because it does not provide new insights.
\begin{proposition}\label{pr:com_PO_SO} 
Suppose that   $h_i(1) \ne 0$ for some $i\in [n]$. 
Then, 
 $(X_1,\dots,X_n)\in \mathbb A^+_n(X)$ is  Pareto optimal in $\mathbb A^+_n(X)$ if and only if $\sum_{i=1}^n \rho_{\tilde h_i}(X_i) =\dboxplus_{{i=1}}^n\rho_{\tilde h_i}(X)$.
\end{proposition}
We now show that $\boldsymbol \lambda$-optimality in $\mathbb A^+_n(X)$ pins down Pareto optimality. This result is stated in a stronger form than 
Proposition \ref{pr:1} for the corresponding notions of optimality in $\mathbb A_n(X)$.
\begin{proposition}\label{pr:2}
Suppose that   $h_i(1) = 0$ for all $i\in [n]$. 
For an allocation $(X_1,\dots,X_n)\in \mathbb A^+_n(X)$, it holds that (i)$\Rightarrow$(ii)$\Rightarrow$(iii):
\begin{enumerate}[(i)]
\item   $\sum_{i=1}^n \lambda_i  \rho_{h_i}(X_i) =\dboxplus_{{i=1}}^n(\lambda_i \rho_{h_i})(X)$ for some $(\lambda_1,\dots,\lambda_n)\in (0,\infty)^n$;
\item $(X_1,\dots,X_n)$ 
 is  Pareto optimal in $\mathbb A^+_n(X)$;
 \item  $\sum_{i=1}^n \lambda_i  \rho_{h_i}(X_i) =\dboxplus_{{i=1}}^n(\lambda_i \rho_{h_i})(X)$ for some $(\lambda_1,\dots,\lambda_n)\in [0,\infty)^n\setminus \{\mathbf 0\}$.
\end{enumerate}
\end{proposition}

    Comonotonicity plays an important role in the proof of Proposition \ref{pr:2}. 
    The risk possibility set of the set of comonotonic allocations is defined as  $S=\{(\rho_{h_1}(X_1), \dots, \rho_{h_n}(X_n)): (X_1,\dots,X_n) \in \mathbb{A}_n^+(X)\}$. The comonotonic additivity of distortion riskmetrics guarantees that the risk possibility set $S$ is a convex set when restricted to $\mathbb{A}_n^+(X)$. This needs not be true on $\mathbb{A}_n(X)$. In this case, we cannot use the Hahn-Banach Theorem to obtain the existence of the Negishi weights $(\lambda_1,\dots,\lambda_n)$, which is the reason why we did not state a ``converse statement" in Proposition \ref{pr:1}.   Propositions \ref{pr:equivalence} and \ref{pr:2} together yield that if all agents have concave distortion functions,
then any Pareto-optimal allocation in $\mathbb A_n(X)$, which yields the same welfare for all agents 
as a Pareto-optimal allocation  in $\mathbb A_n^+(X)$,
must satisfy (iii). 
If their distortion functions are strictly concave, then by Proposition \ref{pr:como}, every Pareto-optimal allocation can be found through an inf-convolution.

We now aim to characterize further the set of Pareto-optimal allocations in $\mathbb A^+_n(X)$. The following result generalizes Proposition 5 of \cite{ELW18} for dual utilities.

\begin{theorem}\label{th:3} 
Suppose that $h_1(1)=\dots=h_n(1)$. Then
 $$ 
\dboxplus_{i=1}^n\rho_{h_i}=\rho_{h_\wedge }, 
$$
where $h_{\wedge}(t)=\min\{ h_1(t),\dots, h_n(t)\}$, and $\rho_{h_\wedge }$ is finite on $\X$.
Moreover, a sum-optimal allocation $(X_1,\dots,X_n)$  in $\mathbb A^+_n(X)$ is given by  $X_i=f_i(X)$, $i=1,\dots,n$, where
\begin{align}\label{eq:f}
f_i(x)=\int_0^xg_i(t)\d t, 
\mbox{~~~and~~~}
g_i(x)= \frac{1}{|M_x|} \id_{\{ i\in M_x\}}, ~~~~x\in \R,
\end{align}
and where
 $M_x=\{j\in [n]:  h_j(\p(X>x))=h_{\wedge } (\p(X>x))\}$. 
The sum-optimal allocation is unique up to constant shifts almost surely if and only if $|M_x|=1$ for $\mu_X$-almost every $x$, where $\mu_X$ is the distribution measure of $X$.
\end{theorem}

A key step in the proof of Theorem \ref{th:3} is the following lemma, which gives a convenient alternative formula for $\rho_h(f(X))$. 
The lemma  generalizes Lemma 2.1 of \cite{CL17} for dual utilities in the context of optimal reinsurance design.

\begin{lemma}\label{lem:change}
For any $h\in \mathcal H^{\rm BV}$, random variable $X$ bounded from below,
and  increasing Lipschitz function $f$  with right-derivative $g $,
we have 
\begin{align}\label{eq:unique-lem}
\rho_{h}(f(X))=\int_0^\infty g(x)  h(\p(X>x))\d x+\int_{-\infty}^0 g(x)  (h(\p(X>x)-h(1))\d x.\end{align}
\end{lemma}

The results in Theorem \ref{th:3} can be extended to domains like $\{X\in L^p: X_-\in L^\infty\}$ for $p\ge 0$ as long as $\rho_{h_1},\dots,\rho_{h_n}$ are finite on this domain. This is because Lemma \ref{lem:change} only requires boundedness from below.  
The next example illustrates the uniqueness statement in Theorem \ref{th:3}, which gives not only unique sum-optimal allocations in $\mathbb A_n^+(X)$, but also unique Pareto-optimal ones, up to constant shifts.

\begin{example}
Suppose that $\rho_{h_1}=\beta_1 \E + \gamma_1 \mathrm{GD}$,
$\rho_{h_2}=\beta_2 \E + \gamma_2 \mathrm{MMD}$ and 
$\rho_{h_3}=\beta_3 \E + \gamma_3 \mathrm{IQD}_\alpha$
for some $\beta_i,\gamma_i>0$, $i=1,2,3$, and $\alpha \in [0,1/2)$.
For any continuously distributed $X\in \X$, the Pareto-optimal allocation in $\mathbb A^+_3(X)$ is unique up to constant shifts. 
To see this, by 
  Proposition \ref{pr:com_PO_SO},  any Pareto-optimal allocation $(X_1,X_2,X_3)$ in $\mathbb A^+_3(X)$ 
satisfies $\sum_{i=1}^3 \rho_{\tilde h_i}(X_i) =\dboxplus_{{i=1}}^3\rho_{\tilde h_i}(X)$. 
Noting that for each $1\le i< j\le 3$, $\tilde h_i(t) = \tilde h_j(t)$
for at most two points $t\in (0,1)$,  
by   Theorem \ref{th:3}, the allocation $(X_1,X_2,X_3)$
is unique up to constant shifts.
\end{example}

By replacing $h_i$ with $\lambda_i h_i$ for some $\lambda_i\ge 0$, we obtain the following corollary, which helps to identify $\boldsymbol \lambda$-optimal allocations in conjunction with Theorem \ref{th:3}. 

\begin{corollary}\label{cor:th3} Let $\boldsymbol \lambda \in \R_+^n\setminus \{\mathbf{0}\}$ be a vector and  $\dboxplus_{i=1}^n\rho_{\lambda_i h_i}$ be finite.  Then
 $ 
\dboxplus_{i=1}^n\rho_{\lambda_i h_i}=\rho_{h_{\boldsymbol \lambda}}, 
$ 
where $h_{\boldsymbol \lambda}(t)=\min\{\lambda_1 h_1(t),\dots,\lambda_n h_n(t)\}$ for $t\in[0,1]$.
\end{corollary}

By Proposition \ref{prop:1}, the inf-convolution $\dboxplus_{i=1}^n\rho_{\lambda_i h_i}$ being finite implies that $\lambda_ih_i(1)$ are equal for  all $i \in [n]$. Corollary \ref{cor:th3} is thus only useful for the case of location-invariant distortion riskmetrics ($h_i(1)=0$, $i\in [n]$), as otherwise we simply normalize $\lambda_i=1$, $i\in[n]$.  Theorem \ref{th:3}'s characterization of  $\boldsymbol \lambda$-optimality in $\mathbb{A}_n^+(X)$ extends to location-invariant distortion riskmetrics by setting $M_x=\{i\in [n]: \lambda_i h_i(\p(X>x))=h_{\boldsymbol \lambda} (\p(X>x))\}$ in   \eqref{eq:f}.

For cash-additive and reverse cash-additive distortion riskmetrics,
Proposition \ref{pr:com_PO_SO} and Theorem \ref{th:3} together yield a full characterization of Pareto-optimal allocations in $\mathbb A_n^+$.  
It remains to  characterize those for location-invariant distortion riskmetrics.
The next proposition gives an answer for a large class of such distortion riskmetrics. 

 \begin{proposition}\label{pr:charac_location_inv}
 Suppose $h_i(1)=0$ and $h_i(t)>0$ for all $i\in [n]$ and all $t \in (0,1)$. Then the allocation
$(X_1,\dots,X_n)\in \mathbb A^+_n(X)$ is Pareto optimal if and only if there exists $K\subseteq  [n]$ and a vector $\boldsymbol \lambda \in  (0,\infty)^{K}$ such that $(X_i)_{i\in K}$ 
solves $\dboxplus_{i\in K}\rho_{\lambda_i h_i}(X)$, and $X_i$, $i\not\in K$ are constants.
 \end{proposition}
 
The assumption that $h_i(t)>0$ for all $i\in [n]$ and all $t \in (0,1)$ is critical for the characterization of Proposition \ref{pr:charac_location_inv}. 
This condition has a natural interpretation, as it  is equivalent to $\rho_{h_i}(X)>0$ for all non-degenerate $X$ 
and it is satisfied by many variability measures; it is part of the definition of deviation measures of \cite{RUZ06}.
 But this assumption rules out IQD, which we study in the next section.

\subsection{IQD agents and positively dependent optimal allocations
}
\label{sec:42}
We start with the comonotonic risk sharing problem among IQD agents. 
The following proposition gives the corresponding statements, parallel to Theorem \ref{th:IQD}, on Pareto optimality and inf-convolution
in this setting. The sum-optimal allocations are given by Theorem \ref{th:3}. 
 \begin{proposition}\label{prop:9}
Consider $X\in \X$ and the IQD risk sharing problem in $\mathbb A_n^+(X)$ with $\alpha_1,\dots,\alpha_n\in [0,1/2)$. 
\begin{enumerate}[(i)]
\item An allocation of $X$
 is Pareto optimal if and only if  
it is sum optimal.
\item 
For  $\lambda_1,\dots,\lambda_n \ge 0$,
$$\dboxplus_{i=1}^n (\lambda_i \mathrm{IQD}_{\alpha_i}) =  \left(\bigwedge_{i=1}^n \lambda_i \right)\mathrm{IQD}_{\bigvee_{i=1}^n \alpha_i}.  $$ 
In particular,
$ \dboxplus_{i=1}^n  \mathrm{IQD}_{\alpha_i}  =  \mathrm{IQD}_{\bigvee_{i=1}^n \alpha_i}.  $ 
\end{enumerate}
\end{proposition}

Comparing Theorem \ref{th:IQD} with Proposition \ref{prop:9}, 
we note that for $\alpha_1,\dots,\alpha_n\in (0,1/2)$,
we have $\sum_{i=1}^n \alpha_i>\bigvee _{i=1}^n \alpha_i$, which implies that
 \begin{align}
 \dboxplus_{i=1}^n (\lambda_i \mathrm{IQD}_{\alpha_i}) (X)-\dsquare_{i=1}^n (\lambda_i \mathrm{IQD}_{\alpha_i}) (X) >0
 \label{eq:compare-IQD}
 \end{align}
for any continuously distributed $X$.
 This further implies that the Pareto-optimal allocations in $\mathbb A_n(X)$ 
 are disjoint from those in $\mathbb A_n^+(X)$. 
The difference  in \eqref{eq:compare-IQD}
can be interpreted as the welfare gain of allowing agents to share risks in non-comonotonic arrangements.
For IQD agents, 
comonotonic allocations are not Pareto optimal in general, and therefore, some form of ``gambling behaviour" in the allocation is beneficial to all agents, although the agents are neither risk averse nor risk seeking in the sense of \cite{RS70}. This is similar to the case of quantile agents, discussed by \cite{ELW18}.

 \section{Several IQD and risk-averse agents}\label{sec:new5}

Combining results established in Sections \ref{sec:3} and \ref{sec:4}, we are now able to tackle the unconstrained risk sharing problem for IQD and risk-averse agents.
We consider agents from the following two sets: the IQD agents, modelled by distortion functions in
$$\mathcal{H}^{\mathrm{IQD}}=\{t\mapsto \id_{\left\{\alpha< t < 1-\alpha\right\}}: \alpha \in [0, 1/2)\}$$
and  the risk-averse agents, modelled by distortion functions in
  $$\mathcal{H}^{\mathrm{C}}=\{ h \in \mathcal{H}^{\mathrm{BV}} \vert ~h(1)=0,~ h \text{ is concave} \}.$$
That is, $\mathcal{H}^{\mathrm{IQD}}$ is the set of all distortion functions for $\mathrm{IQD}$ variability measures and $\mathcal{H}^{\mathrm{C}}$ is the set of location-invariant concave 
distortion functions $h\in \mathcal{H}^{\mathrm{BV}}$.
Notice that each $h\in \mathcal{H}^{\mathrm{C}}$ is increasing in $[0,s]$ and decreasing in $[s,1]$ for some $s\in (0,1)$. Define the mapping 
$G^\alpha_\lambda: \mathcal{H}^{\mathrm{C}} \to \mathcal{H}^{\mathrm{BV}}$ for $\alpha \in [0,1/2)$ and $\lambda \ge 0$ as
 $$ G^\alpha_\lambda(h)(t) = \left(h(t-\alpha) \wedge h(t+\alpha)\wedge \lambda\right)\id_{\{\alpha<t < 1-\alpha\}}~~~\mbox{for} ~t \in [0,1].$$ 
The mapping $G^\alpha_\lambda$ transforms a concave distortion function to another distortion function with value $0$ on  $ [0, \alpha]\cup[1-\alpha,1]$.
See Figure \ref{fig:G} for an illustration of this transform. For $\alpha\ge 1/2$, we define $G^\alpha_\lambda(h)=0$.

\begin{figure}[t]
\caption{An illustration of the transform $G_\lambda^\alpha$}
\label{fig:G}
\begin{subfigure}[b]{0.45\textwidth}
\centering
\begin{tikzpicture}[scale=4.5]
\draw[->] (-0.2,0) --(1.1,0) node[right] {$t$};
\draw[->] (0,-0.2) --(0,0.6) node[above] {$h(t)$};
\draw [black] (0,0) .. controls  (0.2,0.5) .. (1,0);
\node[below] at (-0.03,0) {\scriptsize $0$};
\node[below] at (1,0) {\scriptsize $1$};
\end{tikzpicture}
\caption{$h$}
\end{subfigure}
\begin{subfigure}[b]{0.45\textwidth}
\centering
\begin{tikzpicture}[scale=4.5]
\draw[->] (-0.2,0) --(1.2,0) node[right] {$t$};
\draw[->] (0,-0.2) --(0,0.6) node[above]  {$G_{{\lambda}}^\alpha(h)(t)$};
\draw [black, dashed] (0.1,0) .. controls  (0.3,0.5) .. (1.1,0);
\draw [black, dashed] (-0.1,0) .. controls  (0.1,0.5) .. (0.9,0);
\draw[gray,dashed] (0,0.28) --(0.63,0.28);
\draw [black] (0.43,0.28) .. controls  (0.55,0.225) and (0.7,0.12) .. (0.9,0);
\draw [black]  (0.1,0).. controls (0.12, 0.051) and  (0.18,0.215)  .. (0.23,0.28);
\draw[black] (0.23,0.28) --(0.43,0.28);
\node[below] at (0.1,0) {\scriptsize $\alpha$};
\node[below] at (0.9,0) {\scriptsize $1-\alpha$};
\node[below] at (-0.1,0) {\scriptsize $-\alpha$};
\node[below] at (1.08,0) {\scriptsize $1+\alpha$};
\node[left]  at (0,0.28) {$\lambda$};
\end{tikzpicture}
\caption{$G_\lambda^\alpha(h)$}
\end{subfigure}\\
\end{figure}
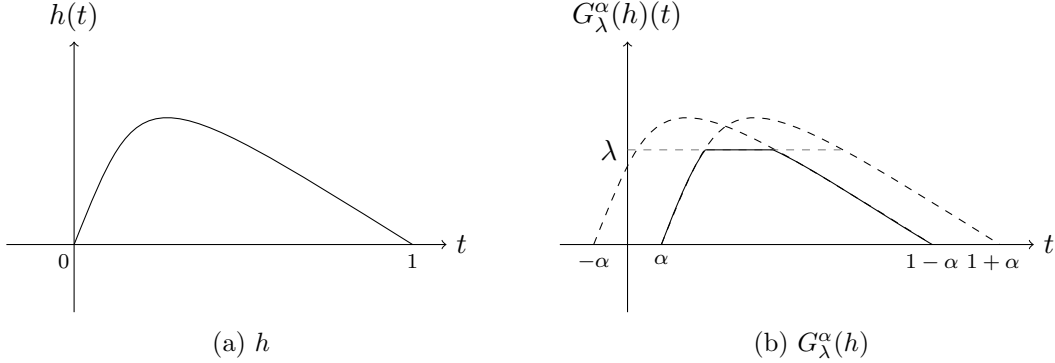

We will see in the next proposition that the function $G^\alpha_\lambda$ 
plays an important role because  of the inf-convolution  of $\lambda \mathrm{IQD}_\alpha$ and $\rho_h$  for $h\in \mathcal{H}^{\mathrm{C}} $ satisfies
$$ (\lambda \mathrm{IQD}_{\alpha} ) \square \rho_h =  \rho_{G^\alpha_\lambda(h)}.$$
This formula is a special case of \eqref{eq:inf_IQD} in Theorem \ref{pr:IQD_h} below. 
\begin{theorem}\label{pr:IQD_h}
Let $N\subseteq [n]$ and $I=[n]\setminus N$. 
Suppose that $h_i \in \mathcal{H}^{\mathrm{C}} $ for $i\in N$ and $h_i\in \mathcal{H}^{\mathrm{IQD}}$ for $i\in I$ with IQD parameter $\alpha_i$.   Denote by  $\alpha=\sum_{i\in I} \alpha_i$.
\begin{enumerate}[(i)]
\item
For   $\lambda_1, \dots, \lambda_n \ge 0$,  denoting by $\lambda=\bigwedge_{i \in I}\lambda_i$ and 
 $h={\bigwedge_{i\in N}({\lambda_i h_i}})$, we have
\begin{equation}\label{eq:inf_IQD}
\dsquare_{i=1}^n (\lambda_i \rho_{h_i})=\rho_{G^\alpha_\lambda(h)}. \end{equation}
\item A Pareto-optimal  allocation is given by
\begin{equation}\label{eq:IQD_allocation2}X_i=(X-c)\id_{A_i\cup B_i}+ Y_i+c_i,
\end{equation}
where, by denoting by $\beta=\alpha\wedge (1/2)$,
\begin{enumerate}[(a)]
\item  $\left \{A_i\right\}_{i=1}^n$ and $\left \{B_i\right\}_{i=1}^n$  are partitions of  a right $\beta$-tail event $A$  and a left $\beta$-tail event $B$ of $X$ with $A,B$ disjoint, respectively,  satisfying $\P(A_i)=\P(B_i)=\alpha_i\beta/\alpha$ for $i \in I$ and $A_i=B_i=\varnothing$ for $i\in N$; 
\item $(Y_1,\dots,Y_n)$ is a Pareto-optimal allocation of $(X-c)\id_{(A\cup B)^c}$ for preferences with distortion functions $ \hat h_1,\dots,  \hat h_n$ 
where $\hat h_i=h_i$ if $i\in N$ and $\hat h_i(t)=\id_{\{t\in (0,1)\}}$ for $i\in I$. 
\item $c\in [Q^-_{1/2}(X),Q^+_{1/2}(X)]$ and $\sum_{i=1}^n c_i=c$.
 \end{enumerate}
 \end{enumerate} 
\end{theorem}

 Similarly to the allocation in Theorem \ref{th:IQD}, the allocation is counter-monotonic when conditioning on tail events. On the tails, the risk is allocated to only one IQD agent at a time. Once again, this is because IQD agents do not care about tail risks. What is left of the risk is then distributed optimally among agents using the techniques introduced in Section \ref{sec:4}. The resulting allocation is quite unusual, but it can be implemented in financial markets through derivatives such as call options and digital options. Section \ref{subsec:IQD_and_RA} further analyzes the allocations found in Theorem \ref{pr:IQD_h} and gives explicit examples of the underlying financial contracts.

\begin{remark}
Let $N\subseteq [n]$ and $I=[n]\setminus N$. 
Suppose that $h_i \in \mathcal{H}^{\mathrm{C}} $ for $i\in N$ and $h_i\in \mathcal{H}^{\mathrm{IQD}}$ for $i\in I$ with IQD parameter $\alpha_i$.   For any $\lambda_1, \dots, \lambda_n \ge 0$ it is
 $ 
\dboxplus_{i=1}^n (\lambda_i \rho_{h_i})=\rho_{h_{\boldsymbol{\lambda}}}, $
where $h_{\boldsymbol{\lambda}}=\bigwedge_{i\in [n]}\lambda_i h_i$. The distortion function $h_{\boldsymbol{\lambda}}$ takes value 0 on $[0, \bigvee_{i\in I} \alpha_i] \cup[\bigvee_{i\in I} \alpha_i,1]$; on the other hand, the   distortion function  $G^\alpha_\lambda(h)$ from Theorem \ref{pr:IQD_h} takes value $0$ on $[0, \sum_{i\in I} \alpha_i] \cup[\sum_{i\in I} \alpha_i,1]$.

\end{remark}

\section{GD, MMD and IQD agents}\label{sec:5}

We now provide examples of the results obtained in Section \ref{sec:3} and \ref{sec:4}. 
Some calculation details are put in Appendix \ref{app:C}. 
The following two subsections come back to the risk sharing problem with several $\mathrm{IQD}$ agents and explains further the allocations found in Section \ref{subsec:Unc_IQD}. The last two subsections analyze the risk sharing problem when agents consider the Gini and mean-median deviations as the relevant statistical measures of risk.

\subsection{Several IQD agents} \label{subsec:IQD_agents} 
The difference between the two sum-optimal allocations found in Theorem \ref{th:IQD} and Proposition \ref{prop:9} is important.

In contrast, Figure \ref{Fig:5} illustrates some comonotonic allocations that are $\boldsymbol \lambda$-optimal  (and also Pareto optimal and sum optimal; see Proposition \ref{prop:9}) when restricted to the subset $\mathbb{A}_n^+(X)$. The solution for $\dboxplus_{i=1}^n (\lambda_i \mathrm{IQD}_{\alpha_i})$ is not unique as $|M_x|$ can be larger than 1. The figure depicts a particular case when simultaneously $\alpha_1<\alpha_2<\alpha_3$ and $\lambda_1<\lambda_2<\lambda_3$. The left panel shows the distortion function of each agent multiplied by the corresponding $\lambda$, and the lower envelope $h_{\boldsymbol \lambda}(t)$. 
Figure \ref{fig:all_com_IQD} presents a sum-optimal allocation where all three agents take non-zero risks.
Comonotonic sum-optimal allocations are not unique, because the allocation where agent 3 takes all risks in the $\alpha_3$-tails and agent 1 takes the rest is also sum optimal.  
As discussed before, comonotonic sum-optimal allocations are generally not sum optimal in $\mathbb A_n(X)$. 

\begin{figure}[ht]
\centering
\caption{Distortion functions and the sum-optimal allocation for $\dboxplus_{i=1}^n \lambda_i \mathrm{IQD}_{\alpha_i}$ } \label{Fig:5}
\begin{subfigure}[b]{0.45\textwidth}
  \centering
\begin{tikzpicture}[scale=0.75]
\draw[->] (-0.2,0) --(9.2,0) node[right] {\small $t$};
\draw[->] (0,-0.2) --(0,6.2) node[above] {\small$\lambda_i h_i(t)$};
\node at (-0.3,-0.3) {\small$0$};
\node at (9,-0.3)  {\small$1$};
\draw (1,2)--(8,2);
\draw[gray,dashed] (1,0) -- (1,2);
\draw[gray,dashed] (8,0) -- (8,2);
\draw[gray,dashed] (0,2) -- (1,2);
\node at (-0.3,2) {\small$\lambda_1$};
\node at (1,-0.3) {\small$\alpha_1$};
\node at (8,-0.3)  {\small$1-\alpha_1$};
\draw (2.5,4)--(6.5,4);
\draw[gray,dashed] (2.5,0) -- (2.5,4);
\draw[gray,dashed] (6.5,0) -- (6.5,4);
\draw[gray,dashed] (0,4) -- (2.5,4);
\node at (-0.3,4) {\small$\lambda_2$};
\node at (2.5,-0.3) {\small$\alpha_2$};
\node at (6.5,-0.3)  {\small$1-\alpha_2$};
\draw (4,6)--(5,6);
\draw[gray,dashed] (4,0) -- (4,6);
\draw[gray,dashed] (5,0) -- (5,6);
\draw[gray,dashed] (0,6) -- (4,6);
\node at (-0.3,6) {\small$\lambda_3$};
\node at (4,-0.3) {\small$\alpha_3$};
\node at (5,-0.3)  {\small$1-\alpha_3$};

\draw [red,very thick] (0,0)--(4,0);
\draw [red, very thick] (4,2)--(5,2);
\draw [red, very thick] (5,0)--(9,0);
\end{tikzpicture}
\caption{Distortion functions for $\lambda_i \mathrm{IQD}_{\alpha_i}$, $i=1,2,3$}\label{fig:one_com_IQD}
 \end{subfigure}
~
\begin{subfigure}[b]{0.45\textwidth}
\centering
\begin{tikzpicture}[scale=0.7]
\draw[->] (-0.2,0) --(9.2,0) node[right] {\footnotesize$x$};
\draw[->] (0,-0.2) --(0,6) node[above] {\footnotesize$f_i(x)$};
\node at (-0.3,-0.3) {\footnotesize$0$};

\draw[blue,domain =0:2] plot (\x ,{\x});
\draw[blue,domain =2:7] plot (\x ,{2});
\draw[blue,domain =7:9] plot (\x ,{\x-5}) node[right] {\footnotesize$X_2$};

\draw[red,domain =2:3.5] plot (\x ,{\x-2});
\draw[red,domain =3.5:5.5] plot (\x ,{1.5});
\draw[red,domain =5.5:7] plot (\x ,{\x-4});
\draw[red,domain =7:9] plot (\x ,{3}) node[right] {\footnotesize$X_3$};

\draw[black,domain =3.5:5.5] plot (\x ,{\x-3.5});
\draw[black,domain =5.5:9] plot (\x ,{2}) node[right] {\footnotesize$X_1$};

\draw[gray,dashed] (2,0) -- (2,2);
\node at (1.5,-0.3) {\footnotesize$Q^-_{1-\alpha_2}(X)$};
\draw[gray,dashed] (3.5,0) -- (3.5,1.5);
\node at (3.7,-0.3) {\footnotesize$Q^-_{1-\alpha_3}(X)$};
\draw[gray,dashed] (5.5,0) -- (5.5,2);
\node at (6,-0.3) {\footnotesize$Q^-_{\alpha_3}(X)$};
\draw[gray,dashed] (7,0) -- (7,3);
\node at (7.8,-0.3) {\footnotesize$Q^-_{\alpha_2}(X)$};
\end{tikzpicture}
\caption{Allocation for $\dboxplus_{i=1}^n \lambda_i \mathrm{IQD}_{\alpha_i}$ }\label{fig:all_com_IQD}
\end{subfigure}
\end{figure}
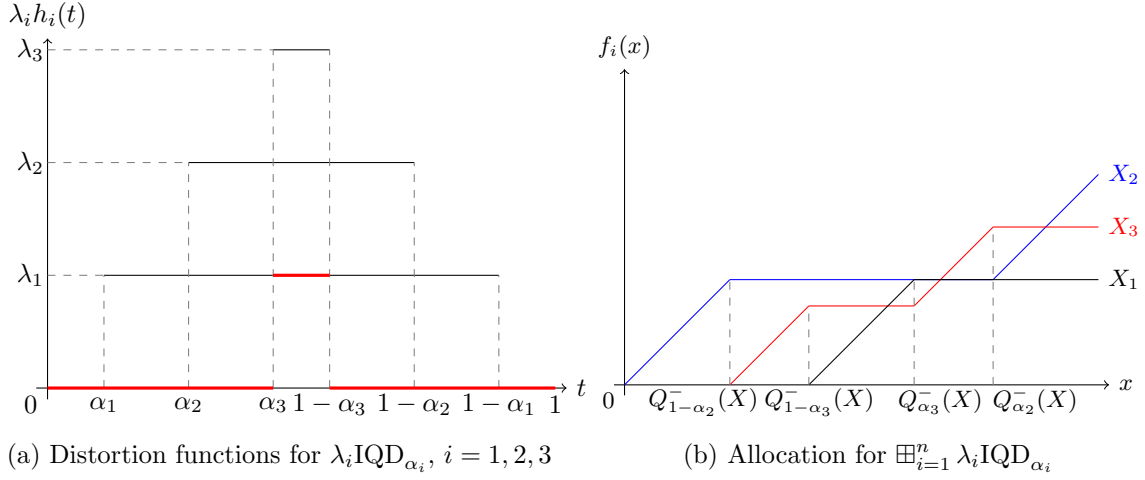

\subsection{The GD, MMD, and IQD problem}\label{subsec:IQD_and_RA} 

We now turn to the allocations characterized by Theorem \ref{pr:IQD_h}. Consider the problem of sharing risk between Anne, Bob and Carole, i.e., the case when there is only one $\mathrm{GD}$ agent, one $\mathrm{MMD}$ and $\mathrm{IQD}$ agent. Let $\alpha<1/2$ and $\lambda_1,\lambda_2, \lambda_3>0$ and consider the inf-convolution $$\inf_{(X_1,X_2, X_3)\in \mathbb{A}(X)} \left\{ \lambda_1\mathrm{GD}(X_1)+\lambda_2\mathrm{MMD}(X_2)+\lambda_3\mathrm{IQD}_\alpha(X_3)\right\}.$$
Without loss of generality we assume $Q_{1/2}^-(X)=0$ for the convenience of presentation,  so that $c$ in Theorem \ref{pr:IQD_h} is taken as $0$.

Let $A$ be a right $\alpha$-tail event and $B\subsetneq A^c$ be a left $\alpha$-tail event of $X$, where $A$ and $B$ are disjoint sets. All the $\alpha$-tail risks must go to the $\mathrm{IQD}$ agent. That is, every sum-optimal allocation requires that the IQD agent takes the whole risk on $A\cup B$.  

It remains to share risk ``in the middle", that is, on the event $(A\cup B)^c$. 
We denote by $Y=X\id_{(A\cup B)^c}$, which has an optimal allocation $(Y_1,\dots,Y_n)$ in Theorem \ref{pr:IQD_h} which is comonotonic on $(A\cup B)^c$. 
This is done in the same fashion as we do later for comonotonic risk sharing, with the caveat that the $\mathrm{IQD}$ agent might take on some risk depending on the weights $\lambda_1,\, \lambda_2$ and $\lambda_3$. Define
 $c_1=1/2-\sqrt{1/4-\lambda_3/\lambda_1}+\alpha$,  $c_2=\lambda_3/\lambda_2+\alpha$ and $c_3=1-\lambda_2/\lambda_1+\alpha$.
 If $c_1\in (\alpha,1/2)$, then $\lambda_1 h_{\mathrm{GD}}(t)$  and $\lambda_3 h_{\mathrm{IQD}}(t)$ cross twice on $(0,1)$, once at $c_1-\alpha$ and then once again at $1-c_1+\alpha$. If $c_2\in (\alpha,1/2)$, then $\lambda_2 h_{\mathrm{MMD}}(t)$ and $\lambda_3 h_{\mathrm{IQD}}(t)$ cross twice on $(0,1)$, once at $c_2-\alpha$ and then once again at $1-c_2+\alpha$.
 Similarly, if $c_3\in (\alpha,1/2)$ then $\lambda_1 h_{\mathrm{GD}}(t)$ and $\lambda_2 h_{\mathrm{MMD}}(t)$ cross at $c_3-\alpha$ and $1-c_3+\alpha$.  Note that $c_2>\alpha$ and $\alpha \le c_1\le 1/2+\alpha$ whenever $1/4\ge \lambda_3/\lambda_1$.

 We have six cases to handle; the details can be found in Appendix \ref{app:C}. Figure \ref{Fig:H_six_cases} plots the function $G_{{\lambda}}^\alpha(h)$ for $h=\min\{ \lambda_1 h_{\mathrm{GD}}, \lambda_2 h_{\mathrm{MMD}}\}$. The  red, blue, and black parts represent the risk that goes to  the $\mathrm{GD}$ agent, the $\mathrm{MMD}$ agent, and the $\mathrm{IQD}$ agent, respectively.

\begin{figure}[h!]
\caption{The function $G_{{\lambda}}^\alpha(h)$}
\label{Fig:H_six_cases}
\begin{subfigure}[b]{0.3\textwidth}
\centering
\begin{tikzpicture}[scale=4.5]
\draw[->] (0.15,0) --(0.9,0) node[below] {$t$};
\draw[->] (0.2,-0.1) --(0.2,0.3) node[above] {$G_{{\lambda}}^\alpha(h)(t)$};
\draw[red,domain =0.25:0.5] plot (\x ,{(\x-0.25)-(\x-0.25)*(\x-0.25});
\draw[red,domain =0.5:0.75] plot (\x ,{(\x+0.25)-(\x+0.25)*(\x+0.25});
\node[below] at (0.25,0) {$\alpha$};
\node[below] at (0.75,0) {$1-\alpha$};
\end{tikzpicture}
\caption{Case 1}
\end{subfigure}~~~
\begin{subfigure}[b]{0.3\textwidth}
\centering
\begin{tikzpicture}[scale=4.5]
\draw[->] (0.15,0) --(0.9,0) node[below] {$t$};
\draw[->] (0.2,-0.1) --(0.2,0.3) node[above] {$G_{{\lambda}}^\alpha(h)(t)$};
\draw[blue,domain =0.25:0.5] plot (\x ,{(\x-0.25)});
\draw[blue,domain =0.5:0.75] plot (\x ,{(0.75-\x)});
\node[below] at (0.25,0) {$\alpha$};
\node[below] at (0.75,0) {$1-\alpha$};
\end{tikzpicture}
\caption{Case 2}
\end{subfigure}~~~
\begin{subfigure}[b]{0.3\textwidth}
\centering
\begin{tikzpicture}[scale=4.5]
\draw[->] (0.15,0) --(0.9,0) node[below] {$t$};
\draw[->] (0.2,-0.1) --(0.2,0.3) node[above] {$G_{{\lambda}}^\alpha(h)(t)$};
\draw[red,domain =0.25:0.37] plot (\x ,{(\x-0.25)-(\x-0.25)*(\x-0.25)});
\draw[gray,dashed] (0.37,0) --(0.37,0.1);
\draw[red,domain =0.63:0.75] plot (\x ,{(\x+0.25)-(\x+0.25)*(\x+0.25)});
\draw[black,domain =0.37:0.63] plot (\x ,{0.1});
\node[below] at (0.25,0) {$\alpha$};
\node[below] at (0.75,0) {$1-\alpha$};
\node[below] at (0.37,0) {$c_1$};
\end{tikzpicture}
\caption{Case 3}
\end{subfigure}\\
\begin{subfigure}[b]{0.3\textwidth}
\centering
\begin{tikzpicture}[scale=4.5]
\draw[->] (0.15,0) --(0.9,0) node[below] {$t$};
\draw[->] (0.2,-0.1) --(0.2,0.3) node[above] {$G_{{\lambda}}^\alpha(h)(t)$};
\draw[blue,domain =0.25:0.35] plot (\x ,{(\x-0.25)});
\draw[blue,domain =0.65:0.75] plot (\x ,{(0.75-\x)});
\draw[black,domain =0.35:0.65] plot (\x ,{0.1});
\node[below] at (0.25,0) {$\alpha$};
\node[below] at (0.75,0) {$1-\alpha$};
\draw[gray,dashed] (0.35,0) --(0.35,0.1);
\node[below] at (0.35,0) {$c_2$};
\end{tikzpicture}
\caption{Case 4}
\end{subfigure}~~~
\begin{subfigure}[b]{0.3\textwidth}
\centering
\begin{tikzpicture}[scale=4.5]
\draw[->] (0.15,0) --(0.9,0) node[below] {$t$};
\draw[->] (0.2,-0.1) --(0.2,0.3) node[above] {$G_{{\lambda}}^\alpha(h)(t)$};
\draw[blue,domain =0.25:1/11+0.25] plot (\x ,{(\x-0.25)});
\draw[blue,domain =0.75-1/11:0.75] plot (\x ,{(0.75-\x)});
\draw[red,domain =0.25+1/11:0.5] plot (\x ,{1.1*(\x-0.25)-1.1*(\x-0.25)*(\x-0.25)});
\draw[red,domain =0.5:0.75-1/11] plot (\x ,{1.1*(\x+0.25)-1.1*(\x+0.25)*(\x+0.25)});
\node[below] at (1/11+0.25,0) {$c_3$};
\draw[gray,dashed] (1/11+0.25,0) --(1/11+0.25,0.1);
;
\node[below] at (0.25,0) {$\alpha$};
\node[below] at (0.75,0) {$1-\alpha$};
\end{tikzpicture}
\caption{Case 5}
\end{subfigure}~~~
\begin{subfigure}[b]{0.3\textwidth}
\centering
\begin{tikzpicture}[scale=4.5]
\draw[->] (0.15,0) --(0.9,0) node[below] {$t$};
\draw[->] (0.2,-0.1) --(0.2,0.3) node[above]  {$G_{{\lambda}}^\alpha(h)(t)$};
\draw[blue,domain =0.25:1/11+0.25] plot (\x ,{(\x-0.25)});
\draw[blue,domain =0.75-1/11:0.75] plot (\x ,{(0.75-\x)});
\draw[red,domain =0.25+1/11:0.41] plot (\x ,{1.1*(\x-0.25)-1.1*(\x-0.25)*(\x-0.25)});
\draw[red,domain =0.59:0.75-1/11] plot (\x ,{1.1*(\x+0.25)-1.1*(\x+0.25)*(\x+0.25)});
\draw[black,domain =0.41:0.59] plot (\x ,{0.15});
\node[below] at (0.25,0) {$\alpha$};
\node[below] at (0.75,0) {$1-\alpha$};
\end{tikzpicture}
\caption{Case 6}
\end{subfigure}
\end{figure}
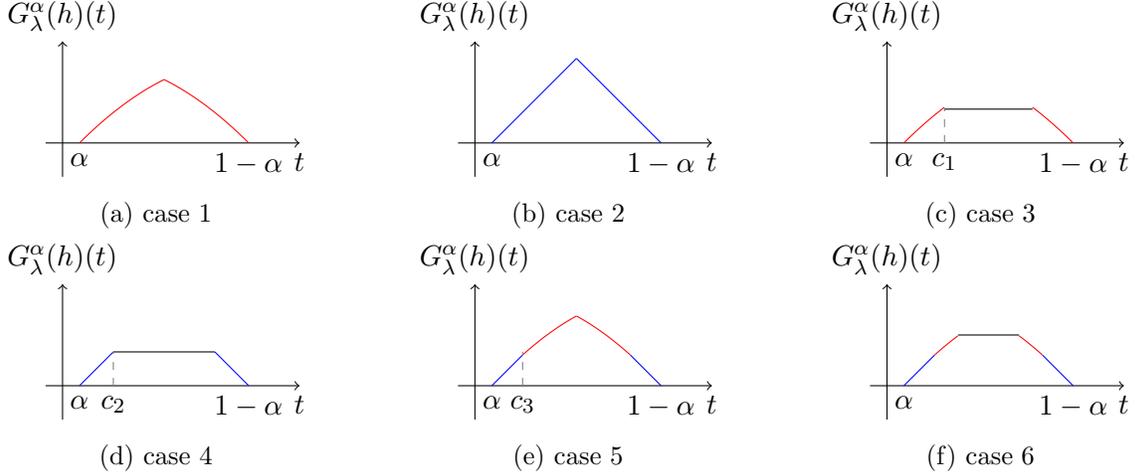

We present the Pareto-optimal allocations $(X_1,X_2,X_3)$ in the six cases below. 
These allocations are generally not comonotonic, 
but they are comonotonic on  the event $(A\cup B)^c$.
 Recall that $Y$ stands for $X\id_{(A\cup B)^c}$. 

\noindent\underline{\textbf{Case 1},
 $c_1\ge 1/2$ and $c_3\le \alpha$}:
$ X_1=Y,~X_2=0~\mathrm{and}~X_3=X\id_{A\cup B}.$

\noindent\underline{\textbf{Case 2}, $c_2\ge 1/2$ and $c_3\ge 1/2$}: 
$X_1=0, ~X_2=Y ~\mathrm{and}~X_3=X\id_{A\cup B}.$

\noindent \underline{\textbf{Case 3}, 
 $c_3\le \alpha<c_1<1/2$}:   
$X_1=X -X_3, ~X_2=0 $ and 
$X_3=X\id_{A\cup B}+Y \wedge Q^-_{c_1}(X)- Y \wedge Q^-_{1-c_1}(X).$
\noindent \underline{\textbf{Case 4}, 
either 
$\alpha<c_1<c_3<1/2$ or
$\alpha<c_2<1/2<c_3$}: 
$X_1=0,~ X_2=X-X_3 $ {and}
$$X_3=X\id_{A\cup B}+ Y \wedge Q^-_{c_2}(X)- Y \wedge Q^-_{1-c_2}(X).$$ 
\noindent \underline{\textbf{Case 5},  $\alpha<c_3<1/2< c_1$}: 
$X_2=X-X_1-X_3,~X_3=X\id_{A\cup B}$ 
and 
$X_1= Y \wedge Q^-_{c_3}(X)- Y \wedge Q^-_{1-c_3}(X).$
\noindent \underline{\textbf{Case 6}, $\alpha<c_3\le c_1<1/2$}:
$X_1= Y \wedge Q^-_{c_3}(X)- Y \wedge Q^-_{c_1}(X)+Y \wedge Q^-_{1-c_1}(X)- Y \wedge Q^-_{1-c_3}(X),$
$$X_2=Y-Y \wedge Q^-_{c_3}(X)+ Y \wedge Q^-_{1-c_3}(X)
~~\mathrm{and}~~X_3=X\id_{A\cup B} + Y \wedge Q^-_{c_1}(X)- Y \wedge Q^-_{1-c_1}(X).$$

The allocation in Case 6 shows a particularly rich feature, and we depict it in Figure \ref{fig:4}. 

\begin{figure}[t]
\caption{Interpreting the Pareto-optimal allocation in Case 6}
\begin{subfigure}[b]{0.99\textwidth}
\begin{center}
\caption{A Pareto-optimal allocation for Anne, Bob and Carole in Case 6, where the red, blue and gray areas represent  the allocations to Anne (GD), Bob (MMD) and Carole (IQD) respectively, up to constant shifts} \label{fig:4}
\vspace{-0cm}
\begin{tikzpicture}[scale=1]

\fill [top color=gray!20, bottom color=gray!20] (0,0)--(5,0)--(3.65,-0.8)-- (1,-0.8)--(1,-2.55)-- (0.9,-2.65) -- (0.8,-2.75)  -- (0.7,-2.86) -- (0,-4);

\fill [top color=gray!20, bottom color=gray!20] (10,0)--(5,0)--(6.33,0.8)-- (9,0.8) --(9,2.9)-- (9.1,3.02) -- (9.2,3.20)  -- (9.3,3.36) -- (9.6,3.98)  -- (10,5);

\fill [top color=red!20, bottom color=red!20] (9,0.8)--(6.33,0.8)--(7,1.24)--(7.5,1.57)--(7.76,1.75)--(9,1.75)--(9,0.8);

\fill [top color=red!20, bottom color=red!20] (1,-0.8)--(3.65,-0.8)--(2.25,-1.6)--(1,-1.6)--(1,-0.8);

\fill [top color=blue!20, bottom color=blue!20] (1,-1.6)--(2.25,-1.6)--(2.1,-1.7)--(2,-1.76)--(1.8,-1.9)--(1.6,-2.02)--(1.4,-2.17)--(1.2,-2.34)--(1.1,-2.43)--(1,-2.55)--(1,-1.6);

\fill [top color=blue!20, bottom color=blue!20] (9,1.75)--(7.76,1.75)--(7.8,1.76)--(7.9,1.85)--(8,1.92)--(8.2,2.06)--(8.4,2.23)--(8.6,2.43)--(8.8,2.63)--(8.9,2.75)--(9,2.9);
 
\draw[gray,dashed] (1,0) -- (1,-2.55);
\draw[gray,dashed] (9,0) -- (9,2.9);  

\draw[loosely dashed] (0,-2.55) -- (10,-2.55); 
\node[left] at (0,-2.55)  {\footnotesize $Q^-_{1-\alpha}(X)$}; 
\draw[loosely dashed] (0,2.9) -- (10,2.9); 
\node[left] at (0,2.9)  {\footnotesize $Q^-_{\alpha}(X)$}; 
\node[left] at (0,0)  {\footnotesize $Q^-_{1/2}(X)$}; 

\draw [black] (0,0) -- (10,0);
   \node at (0.5,0.45)   { \footnotesize $B$}  ;              
      \node at (9.5,-0.45)    {\footnotesize $A$};

\node at (10.2,-0.3) { $\omega$};
\node[left] at (9.7,4.3) {  $X(\omega)$}; 
\draw [black] (0,-4) .. controls (1,-2.3) .. (4,-0.6);
\draw [black] (6,0.6) .. controls (9,2.5) .. (10,5);
\draw [black] (4,-0.6) to (6,0.6); 

\draw [decorate,  decoration = {brace,    raise=1pt,    amplitude=5pt}] (0,0) --  (1,0);
\draw [decorate,  decoration = {brace,    raise=1pt,    amplitude=5pt}] (10,0) --  (9,0); 

\draw [dashed] (1,-0.8) to (3.6,-0.8); 
\draw [dashed] (1,-1.6) to (2.2,-1.6); 
\draw [dashed] (9,0.8) to (6.4,0.8); 
\draw [dashed] (9,1.75) to (7.8,1.75);

\end{tikzpicture}
\end{center}

\end{subfigure}\\
\begin{subfigure}[b]{0.99\textwidth}
\begin{center}
\caption{Payoff of Anne from the derivative contracts }
\label{fig:anna}
\begin{tikzpicture}[scale=1.5]
\draw[->] (-0.2,0) --(6.5,0) node[right] {$P$};
\draw[->] (0,-0.2) --(0,2.2) node[above] {Payoff};

\node[below] at (0.5,0) {\footnotesize{$Q^-_{1-\alpha}(P)$}};
\node[below] at (1.5,0) {\footnotesize{$Q^-_{1-c_3}(P)$}};
\node[below] at (2.5,0) {\footnotesize{$Q^-_{1-c_1}(P)$}};
\node[below] at (3.5,0) {\footnotesize{$Q^-_{c_1}(P)$}};
\node[below] at (4.5,0) {\footnotesize{$Q^-_{c_3}(P)$}};
\node[below] at (5.5,0) {\footnotesize{$Q^-_{\alpha}(P)$}};

\draw[red, domain =0:0.5] plot (\x ,{1});
\draw[red, domain =0.5:1.5] plot (\x ,{0});
\draw[red,domain =1.5:2.5] plot (\x ,{\x-1.5});
\draw[red,domain =2.5:3.5] plot (\x ,{1}); 
\draw[red,domain =3.5:4.5] plot (\x ,{\x-2.5});
\draw[red,domain =4.5:5.5] plot (\x ,{2});
\draw[red,domain =5.5:6] plot (\x ,{1});

 \draw[dash pattern={on 0.84pt off 2.51pt}] (0.5,1)--(0.5,0);
  \draw[dash pattern={on 0.84pt off 2.51pt}] (3.5,1)--(3.5,0);
 \draw[dash pattern={on 0.84pt off 2.51pt}] (5.5,2)--(5.5,0);
\end{tikzpicture}
\end{center}
\end{subfigure}
\end{figure}

\begin{remark}[Interpreting the Pareto-optimal allocation by derivative contracts]
We consider Case 6, the most sophisticated case.
Let $P=-X>0$ be the future (after one period) price of an asset to be shared by the three agents. 
Due to symmetry of the preferences, $(-X_1,-X_2,-X_3)$ in Case 6 is a Pareto-optimal allocation for $P$.
The corresponding sharing strategy can be obtained by trading bull call spreads (differences of call options) and digital options.
We will describe the position for Anne first. 
Anne purchases two bull call spreads (a and b) and  a digital put option (c) from Carole, and sells a digital call option (d) to Carole. 
\begin{enumerate}[a:]
\item A bull call spread  with  long call strike price $Q_{c_1}^{-}(P)$ and short call strike  $Q_{c_3}^{-}(P)$.
\item A bull call spread with long call strike price $Q_{1-c_3}^{-}(P)$ and short call strike  $Q_{1-c_1}^{-}(P)$.
\item $Q_{1-c_1}^-(P)-Q_{1-c_3}^-(P)$ units of a digital put with strike price $Q_{1-\alpha}^-(P)$. 
\item $Q_{c_3}^{-}(P)-Q_{c_1}^{-}(P)$ units of a digital put with strike price $Q_{\alpha}^{-}(P)$.
\end{enumerate}
We can check that the above derivative portfolio precisely gives the payoff $-X_1+y$ in Case 6 with $P=-X$, where $y$ is a constant (see Figure \ref{fig:anna}).  
We can construct Bob's strategy using bull call spreads and digital options similarly. Carole writes the above derivative contracts to Anne and Bob and keeps the asset to herself. 
The prices of these contracts can be negotiated among the agents, and they do not affect the optimality of the allocation because of the translation invariance of variability measures.

\end{remark}

\subsection{Insurance between two GD and MMD agents}
We next solve the insurance example (Example \ref{ex:intro}) presented in the introduction.  Consider two individuals, Anne and Bob, who evaluate their risk with GD and MMD, respectively.  That is, set $h_1=h_{\mathrm{GD}}$ and $h_2=h_{\mathrm{MMD}}$.  (Or, they could use  $\E+\lambda_1 \mathrm{GD}$ and $\E+\lambda_2\mathrm{MMD}$, which would not change our analysis.)  This setting is simpler than the three-agent problem in Section \ref{subsec:IQD_and_RA}, and it offers a clearer visualization of the Pareo-optimal allocation.

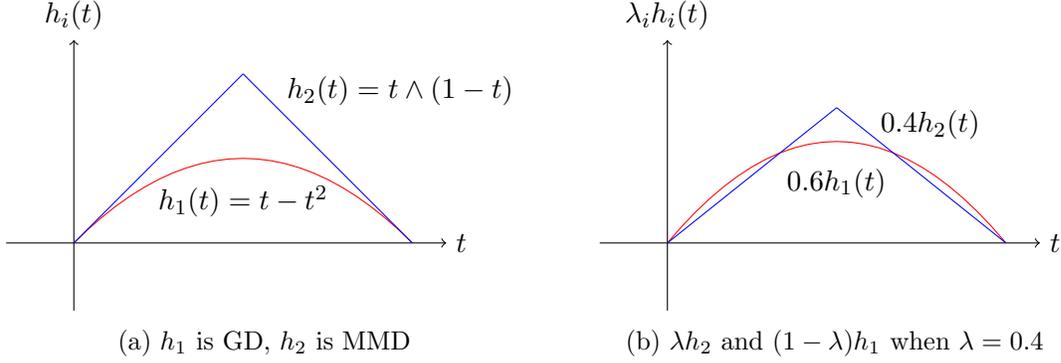
\begin{figure}[t]
\caption{Distortion functions of  GD  and MMD agents}
\label{fig:distortion_MMD_vs_GD}
\begin{subfigure}[b]{0.45\textwidth}
\centering
\begin{tikzpicture}[scale=4.5]
\draw[->] (-0.2,0) --(1.1,0) node[right] {$t$};
\draw[->] (0,-0.2) --(0,0.6) node[above] {$h_i(t)$};
\draw[red,domain =0:1] plot (\x ,{\x-\x*\x});
\draw[blue,domain =0:0.5,] plot (\x ,{\x});
\draw[blue,domain =0.5:1] plot (\x ,{1-\x});
\node[below] at (0.5,0.2) {$h_1(t)=t-t^2$};
\node[right] at (0.6,0.45) {$h_2(t)=t\wedge (1-t)$};
\end{tikzpicture}
\caption{ $h_1$ is GD, $h_2$ is MMD}
\end{subfigure}
\begin{subfigure}[b]{0.45\textwidth}
\centering
\begin{tikzpicture}[scale=4.5]
\draw[->] (-0.2,0) --(1.1,0) node[right] {$t$};
\draw[->] (0,-0.2) --(0,0.6) node[above] {$\lambda_i h_i(t)$};
\draw[red,domain =0:1] plot (\x ,{1.2*\x-1.2*\x*\x});
\draw[blue,domain =0:0.5,] plot (\x ,{0.8*\x});
\draw[blue,domain =0.5:1] plot (\x ,{0.8-0.8*\x});
\node[below] at (0.5,0.25) {$0.6h_1(t)$};
\node[right] at (0.6,0.35) {$0.4h_2(t)$};
\end{tikzpicture}
\caption{$\lambda h_2$ and $(1-\lambda) h_1$ when $\lambda =0.4$}
\end{subfigure}
\end{figure}

Both $h_1$ and $h_2$ are strictly concave, and, by Proposition \ref{pr:como}, any Pareto-optimal allocation in $\mathbb A_n(X)$ is comonotonic. By Proposition \ref{pr:2}, each Pareto-optimal allocation can be found by solving the inf-convolution $\dboxplus_{i=1}^2 (\lambda_i\rho_{h_i})$ for some Negishi weights $(\lambda_1,\lambda_2)\in [0,\infty)^2\setminus \{\mathbf 0\}$. Consider the normalized ones $\lambda_1=\lambda \in [0,1]$ and $\lambda_2=1-\lambda$. Figure \ref{fig:distortion_MMD_vs_GD} depicts the functions $h_i(t)$ and $\lambda_i h_i(t)$.

By positive homogeneity it is $\lambda \rho_{h_1}(X_1)= \rho_{\lambda h_1}(X_1)$ for Anne,
and similarly for Bob. By Corollary \ref{cor:th3}, we have 
$ \dboxplus_{i=1}^2\rho_{\lambda_i h_i}=\rho_{h_{\boldsymbol \lambda}}, $
where $h_{\boldsymbol \lambda}(t)=\min\{\lambda h_1(t), (1-\lambda) h_2(t)\}$. That is, the sum-optimal allocation gives all the marginal $t$-quantile risk to the individual with the lowest $\lambda_i h_i(t)$. 

The condition of Theorem \ref{th:3}   is satisfied, and so the $(\lambda_1,\lambda_2)$-optimal allocation is unique up to constant shifts. Any Pareto-optimal allocation takes the form
$$X_1=X\wedge Q^{-}_{c}(X) - X\wedge Q^{-}_{1-c}(X)+k ~~\mathrm{and}~~ X_2=X-X_1,$$
where $c \in [0,1/2]$ and $k\in \R$ is a constant. We can interpret this as a situation where the GD individual insures the potential losses $X$ of the MMD one. The ceded risk is the random variable $X_1$, while its price is $k$, the latter which needs to be negotiated between the two agents. Next, we argue that the mapping $\lambda \mapsto c$ is surjective.

(i) If $\lambda<1/2$ we have $\lambda h_1<(1-\lambda) h_2$ everywhere and so $c=0$. That is, the ceded risk is $X$ and the GD agent provides full insurance. (ii) Similarly, if $\lambda>2/3$ we have $\lambda h_1>(1-\lambda) h_2$ everywhere and so $c=1/2$. It is the MMD agent that retains all the risk and no insurance is provided. Finally, (iii) if $1/2 < \lambda< 2/3$ then $\lambda h_1>(1-\lambda) h_2$ on both $(0, (2\lambda-1)/\lambda)$ and $((1-\lambda)/\lambda,1)$ and $\lambda h_1<(1-\lambda) h_2$ on $((2\lambda-1)/\lambda, (1-\lambda)/\lambda)$.  Hence, $c=(2\lambda-1)/\lambda$ and the ceded risk has a simple deductible $Q^{-}_{1-c}(X)$ and an upper limit $Q^{-}_{c}(X)$. This type of allocation is depicted in Figure \ref{fig:allocation_MMD_vs_GD}.

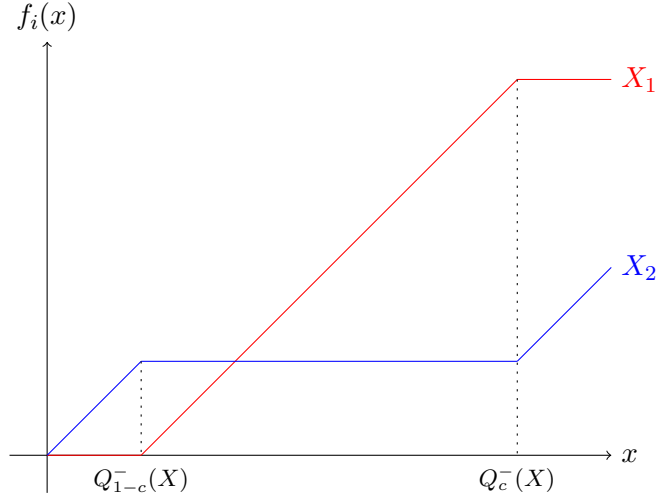
\begin{figure}[h!]
\begin{center}
\caption{A Pareto-optimal allocation for the MMD and GD pair}
\label{fig:allocation_MMD_vs_GD}
\begin{tikzpicture}[scale=2.5]
\draw[->] (-0.2,0) --(3,0) node[right] {$x$};
\draw[->] (0,-0.2) --(0,2.2) node[above] {$f_i(x)$};
\draw[red,domain =0:0.5] plot (\x ,{0});
\draw[red,domain =0.5:2.5] plot (\x ,{\x-0.5});
\draw[red,domain =2.5:3] plot (\x ,{2}) node[right] {$X_1$};

\draw[blue,domain =0:0.5] plot (\x ,{\x});
\draw[blue,domain =0.5:2.5] plot (\x ,{0.5});
\draw[blue,domain =2.5:3] plot (\x ,{\x-2}) node[right] {$X_2$};

\node[below] at (0.5,0) {\footnotesize{$Q^-_{1-c}(X)$}};
\node[below] at (2.5,0) {\footnotesize{$Q^-_{c}(X)$}};

\draw[dash pattern={on 0.84pt off 2.51pt}] (0.5,0.5)--(0.5,0);
\draw[dash pattern={on 0.84pt off 2.51pt}] (2.5,2)--(2.5,0);
\end{tikzpicture}
\end{center}
\end{figure}

The constant $k$ can take any value because by location invariance, for any $k\in \R$, we have  $\rho_{h_i}(X_i+k)=\rho_{h_i}(X_i)+h_i(1)k=\rho_{h_i}(X_i)$ 
and the price of the insurance does not affect Pareto optimality.  
This observation   remains true if agents use $\E+\lambda_i\rho_{h_i}$ instead of $\rho_{h_i}$.

\subsection{Risk sharing with several mixed GD-MMD agents}
We conclude with the problem of sharing risk among many agents $i\in[n]$ evaluating their risks with the variability measure
$$\rho_{h_i}(X_i)=\int X_i \d \left(\left(a_i h_{\mathrm{GD}} + (1-a_i)h_{\mathrm{MMD}}\right )\circ \P\right)=a_i \mathrm{GD}(X_i) +(1-a_i)\mathrm{MMD}(X_i), $$
$a_i\in [0,1]$. It is easily verified that for every $i\in[n]$ the distortion function $h_i= a_i h_{\mathrm{GD}} + (1-a_i)h_{\mathrm{MMD}}$ is strictly concave and satisfies $h_i(1)=0$.
We can therefore invoke Theorem \ref{th:3}, Corollary \ref{cor:th3} and Proposition \ref{pr:charac_location_inv} to characterize the set of Pareto-optimal allocations. Consider the usual normalization of the Negishi weights $\sum_{i=1}^n\lambda_i =1$ with $\lambda_i>0$ and notice that
$$ 
\dboxplus_{i=1}^n\rho_{\lambda_i h_i}=\rho_{h_{\boldsymbol \lambda}},
$$
where $h_{\boldsymbol \lambda}(t)=\min\{\lambda_1 h_1(t),\dots,\lambda_n h_n(t)\}$.

Deriving every agent's allocation (contract) in a closed-form solution is a bit more cumbersome. Yet, Theorem \ref{th:3} and Corollary \ref{cor:th3} still fully pin down the shape of the optimal allocation, and we can visualize it easily. Consider the case when $0<\lambda_1 a_1<\lambda_2 a_2 <\dots< \lambda_n a_n$ and set $M_x=\{i\in [n]: \lambda_i h_i(\p(X>x))=h_{\boldsymbol \lambda} (\p(X>x))\}$ as before. We have that $\vert M_x\vert = 1$ $\mu_X$-almost surely, so the sum-optimal allocation is unique up to constant shifts for any $\boldsymbol \lambda$. Figure \ref{Fig:3} shows an example with three agents.
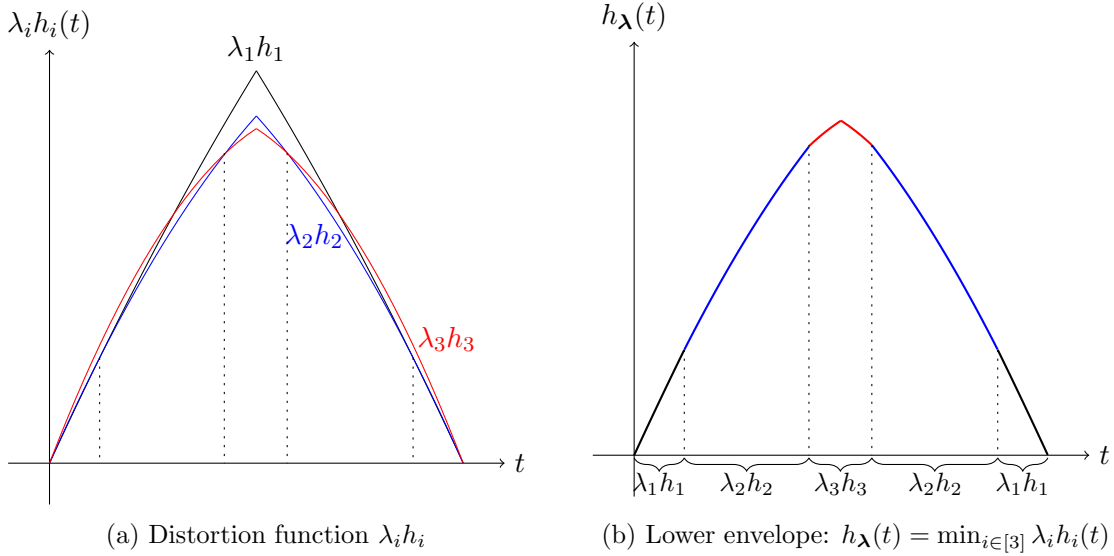
\begin{figure}[h!]
\caption{Distortion functions for mixed GD-MMD agents, where $a_1=1/4, ~a_2=1/2,~a_3=3/4$ and $\boldsymbol \lambda=(0.31,0.32,0.37)$}\label{Fig:3}
\centering
\begin{subfigure}[b]{0.45\textwidth}
\centering
\begin{tikzpicture}[scale=5.5]
\draw[->] (-0.1,0) --(1.1,0) node[right] {$t$};
\draw[->] (0,-0.1) --(0,1) node[above] {$\lambda_ih_i(t)$};

\draw[black,domain =0:0.5] plot (\x ,{2.17*(0.25*(\x-\x*\x)+0.75*(\x))}) node[above] {$\lambda_1h_1$};
\draw[black,domain =0.5:1] plot (\x ,{2.17*(0.25*(\x-\x*\x)+0.75*(1-\x))});
\draw[blue,domain =0:0.5] plot (\x ,{2.24*(0.5*(\x-\x*\x)+0.5*(\x))});
\node[blue] at (0.64,0.55) {$\lambda_2h_2$};
\draw[blue,domain =0.5:1] plot (\x ,{2.24*(0.5*(\x-\x*\x)+0.5*(1-\x))});
\draw[red,domain =0:0.5] plot (\x ,{2.59*(0.75*(\x-\x*\x)+0.25*(\x))});
\draw[red,domain =0.5:1] plot (\x ,{2.59*(0.75*(\x-\x*\x)+0.25*(1-\x))});
\node[red] at (0.96,0.3) {$\lambda_3h_3$};

\draw[dash pattern={on 0.84pt off 2.51pt}] (0.1212,0.2551)--(0.1212,0);
\draw[dash pattern={on 0.84pt off 2.51pt}] (0.4225,0.7504)--(0.4225,0);
\draw[dash pattern={on 0.84pt off 2.51pt}] (0.5745,0.7504)--(0.5745,0);
\draw[dash pattern={on 0.84pt off 2.51pt}] (0.8788,0.2551)--(0.8788,0);

\end{tikzpicture}
\caption{Distortion function $\lambda_ih_i$}
\end{subfigure}
~
\begin{subfigure}[b]{0.45\textwidth}
\centering
\begin{tikzpicture}[scale=5.5]
\draw[->] (-0.1,0) --(1.1,0) node[right] {$t$};
\draw[->] (0,-0.1) --(0,1) node[above] {$h_{\boldsymbol \lambda }(t)$};
\draw[black,domain =0:0.1212, thick] plot (\x ,{2.17*(0.25*(\x-\x*\x)+0.75*(\x))});
\draw[black,domain =0.8788:1, thick] plot (\x ,{2.17*(0.25*(\x-\x*\x)+0.75*(1-\x))});
\draw[blue,domain =0.1212:0.4225, thick] plot (\x ,{2.24*(0.5*(\x-\x*\x)+0.5*(\x))});
\draw[blue,domain =0.5745:0.8788, thick] plot (\x ,{2.24*(0.5*(\x-\x*\x)+0.5*(1-\x))});
\draw[red,domain =0.4225:0.5, thick] plot (\x ,{2.59*(0.75*(\x-\x*\x)+0.25*(\x))});
\draw[red,domain =0.5:0.5745, thick] plot (\x ,{2.59*(0.75*(\x-\x*\x)+0.25*(1-\x))});

\draw[dash pattern={on 0.84pt off 2.51pt}] (0.1212,0.2551)--(0.1212,0);
\draw[dash pattern={on 0.84pt off 2.51pt}] (0.4225,0.7504)--(0.4225,0);
\draw[dash pattern={on 0.84pt off 2.51pt}] (0.5745,0.7504)--(0.5745,0);
\draw[dash pattern={on 0.84pt off 2.51pt}] (0.8788,0.2551)--(0.8788,0);

\draw [decorate,  decoration = {brace,    raise=1pt,    amplitude=5pt}] (0.1212,0) --  (0,0);
\draw [decorate,  decoration = {brace,    raise=1pt,    amplitude=5pt}] (0.4225,0) --  (0.1212,0);
\draw [decorate,  decoration = {brace,    raise=1pt,    amplitude=5pt}] (0.5745,0) --  (0.4225,0);
\draw [decorate,  decoration = {brace,    raise=1pt,    amplitude=5pt}] (0.8788,0) --  (0.5745,0);
\draw [decorate,  decoration = {brace,    raise=1pt,    amplitude=5pt}] (1,0) --  (0.8788,0);

\node at (0.0606,-0.07)   {\small $\lambda_1h_1$}; 
\node at (0.27185,-0.07)   {\small $\lambda_2h_2$}; 
\node at (0.5,-0.07)   {\small $\lambda_3h_3$}; 
\node at (0.72815,-0.07)   {\small $\lambda_2h_2$}; 
\node at (0.9394,-0.07)   { \small$\lambda_1h_1$}; 
\end{tikzpicture}
\caption{Lower envelope:  $h_{\boldsymbol \lambda }(t)= \min_{i \in [3]}\lambda_ih_i (t)$}
\end{subfigure}
\end{figure}

As we obtained in the previous application, $h_{\boldsymbol \lambda}$ induces a partition of the state space on which only one agent takes the full marginal risk. That is, the Pareto-optimal allocation's shape is similar to the payoff obtained with a collection of straight deductibles insurance contracts with upper limits. For instance, the part of the risk that goes to agent 2 is 
$$X_2 = X\wedge b-X\wedge a + X \wedge d -X\wedge c$$
for $0<a<b<c<d<\infty$ implicitly defined through the lower envelope $h_{\boldsymbol \lambda}(t)$.

\section{Heterogeneous beliefs in comonotonic risk sharing}
  \label{sec:7}
   
We considered throughout an atomless probability space $(\Omega, \mathcal F,\p)$. This assumption entails that every individual $i\in [n]$ agrees on the fundamentals of the risk to be shared. 
We now show that all our results on comonotonic risk sharing can be extended to incorporate heterogeneous beliefs with almost no extra effort; this is not true for the unconstrained setting of risk sharing in Section \ref{sec:3}.  
Our characterization of comonotonic risk sharing extends the main results of \cite{L20}, which focus on dual utilities. 
  See also  
 \cite{ELMW20}, \cite{BG20}  and  \cite{Lieb24} for risk sharing with risk measures and heterogeneous beliefs. 

Let $(\Omega,\mathcal F)$ be a measurable space that allows for atomless probability measures and denote by $\P_i$ the atomless probability measure that agent $i\in[n]$ considers. That is, every individual $i\in [n]$ believes the probability space $(\Omega,\mathcal F,\P_i)$ is the true one.
Let $\mathcal{P}$ be the set of atomless probability measures on the measurable space $(\Omega, \mathcal F)$ and let $\ll$ denote absolute continuity.
As before, every individual evaluates their risk with the distortion riskmetric $$\rho_{h_i}^{\P_i}(X)=\int X \d \left ( h_i\circ \P_i\right).$$
For a probability measure $\p$, we define the corresponding  left quantile as $Q^{\p}_t(X)=\inf\{x\in \R: \p(X\le x)\ge 1-t \}$.

The next lemma is instrumental in proving our last result:
\begin{lemma}\label{lemma:absolute_continuity}
Let $\P_0,\p \in \mathcal{P}$ be such that $\P_0 \ll \p$ and $h\in \mathcal{H}^{\mathrm{BV}}$, and suppose that $X \in \mathcal X$ is continuously distributed under $\p$. The function  $g(t)=h(\P_0(X>Q^{\p}_{t}(X)))$, $t \in [0,1]$,  satisfies  $\rho_h^{\P_0}(f(X))= \rho_g^\p(f(X))$ for any increasing functions $f:\R \to \R$.
\end{lemma}

Lemma \ref{lemma:absolute_continuity} states that if a belief $\P_0$ is absolutely continuous with respect to a probability measure $\P$ and if a random variable $X$ is continuously distributed under $\P$, then we can always find a distortion function $g$ such that the two distortion riskmetrics $\rho_h^{\P_0}$ and $\rho_g^\p$ are exactly the same for every random variable $Y=f(X)$ comonotonic with $X$. 

Our last proposition states that when every belief is sufficiently ``well-behaved", then the comonotonic risk sharing problem with heterogeneous beliefs is equivalent to a comonotonic risk sharing problem with homogeneous belief $\P$.
\begin{proposition}\label{pr:heterogeneous_beliefs}
Let $\P_1, \dots, \P_n \in \mathcal{P}$, $h_1, \dots, h_n \in \mathcal{H}^{\mathrm{BV}}$ be given and let $X\in \mathcal X$ admit a density under all $\P_1, \dots, \P_n$. There exist a probability measure $\P \in \mathcal{P}$ and a collection of distortion functions $g_1, \dots, g_n \in \mathcal{H}^{\mathrm{BV}}$ such that the allocation $(X_1, \dots, X_n) \in \mathbb{A}_n^+(X)$ is Pareto optimal for $(\rho_{h_1}^{\P_1}, \dots, \rho_{h_n}^{\P_n})$ if and only if it is Pareto optimal for $(\rho_{g_1}^{\P}, \dots, \rho_{g_n}^{\P})$.
\end{proposition}

The proof follows immediately by noticing that we can find a probability measure $\P$ such that $X$ admits a density under $\P$ and for which $\P_i\ll \P$, $i\in [n]$, and then invoking Lemma \ref{lemma:absolute_continuity}. The key insight is that restricting our allocation to be (globally) comonotonic eliminates, by construction, any potential ``side-bets" originating from the heterogeneity of the beliefs of the agents, as seen in \cite{Lieb24}. Characterizing the set of Pareto-optimal allocations then simply boils down to making suitable changes of measure and/or distortion functions and solving the simplified problem with homogeneous beliefs.

\section{Conclusion}\label{sec:8}
We summarize the paper with a few remarks on the results that we obtained. {While we focused on the case $\X=L^\infty$, all the results of this article generalize to larger spaces provided all the decision criteria $\rho_{h_i}$ and the inf-convolution $\dsquare_{i=1}^n\rho_{h_i}(X)$ are finite on the larger space. We emphasized when key results can be readily generalized, but this finiteness property must often be verified case-by-case. For example, the results on risk sharing with IQD agents can be extended to $L^0$ because the IQD is bounded from below. This property does not generalize to other functionals.}

The unconstrained risk sharing problem for non-concave distortion functions 
typically leads to non-comonotonic sum-optimal allocations without  explicit forms, and they can be difficult to analyze. 
Although we obtained several results on  necessary or sufficient conditions for Pareto and sum optimality (Theorem \ref{th:1} and Propositions \ref{prop:1}-\ref{pr:como}), 
 a full characterization of the Pareto-optimal or sum-optimal allocations for arbitrary distortion riskmetrics
 is beyond the current techniques. 
 
The case of IQD agents is, nevertheless, special, although they do not have concave distortion functions. For this setting, we can fully characterize all Pareto-optimal allocations via sum-optimal ones, and the inf-convolution  for such distortion riskmetrics admit concise formulas (Theorem \ref{th:IQD} and Proposition \ref{prop:9}): 
$$
  \dsquare_{i=1}^n (\lambda_i \mathrm{IQD}_{\alpha_i}) =  \left(\bigwedge_{i=1}^n \lambda_i \right) \mathrm{IQD}_{\sum_{i=1}^n \alpha_i} 
~~
 \mbox{and}~~\dboxplus_{i=1}^n (\lambda_i \mathrm{IQD}_{\alpha_i}) =  \left(\bigwedge_{i=1}^n \lambda_i \right)\mathrm{IQD}_{\bigvee_{i=1}^n \alpha_i}, $$
 and their particular instances  
$$
  \dsquare_{i=1}^n  \mathrm{IQD}_{\alpha_i} =  \mathrm{IQD}_{\sum_{i=1}^n \alpha_i} 
~~
 \mbox{and}~~\dboxplus_{i=1}^n  \mathrm{IQD}_{\alpha_i} =  \mathrm{IQD}_{\bigvee_{i=1}^n \alpha_i}. $$
 These formulas may be compared with the quantile inf-convolutions formulas  obtained by \cite{ELW18} and \cite{LMWW22} 
$$
  \dsquare_{i=1}^n  Q_{\alpha_i}^{-}=  Q^{-}_{\sum_{i=1}^n \alpha_i} ,~~  \dsquare_{i=1}^n  Q_{\alpha_i}^{+}=  Q^{+}_{\sum_{i=1}^n \alpha_i},~~\dboxplus_{i=1}^n Q_{\alpha_i}^{-} = Q^{-}_{\bigvee_{i=1}^n \alpha_i}~~ \mbox{and}~~\dboxplus_{i=1}^n Q_{\alpha_i}^{+} = Q^{+}_{\bigvee_{i=1}^n \alpha_i}. 
 $$

 These results show that 
the representative agent  (using the inf-convolution as its reference) for risk sharing among several IQD agents  is  again an IQD agent, and similarly, the representative agent among several quantile agents is again a quantile agent.

When the distortion functions are concave, or, when we constrain ourselves to the set of comonotonic allocations,
the risk sharing problem becomes much more tractable, and we obtain explicit allocations which are Pareto optimal or sum optimal (Theorem \ref{th:3}). This builds on the comonotonic improvement \textit{\`a la} \cite{LM94}, when the distortion riskmetrics are convex order consistent. 
A high-level summary is that all results that were established for increasing distortion riskmetrics, in particular, \cite{Y87}'s dual utilities, can be extended in parallel to non-increasing ones without extra efforts 
(these results are summarized in Propositions \ref{pr:com_PO_SO}-\ref{pr:charac_location_inv}). 
This opens up various application areas where risks are traditionally studied with only increasing distortion riskmetrics.

Combining the results for IQD agents and for risk-averse agents, we are able to solve risk sharing problems among these agents, whose Pareto-optimal allocations are found explicitly (Theorem \ref{pr:IQD_h}).  Various examples of risk sharing among these agents are presented in Section \ref{sec:5}. 

It remains unclear to us whether our analysis can be generalized to 
 distortion riskmetrics other than IQD, which are not convex (i.e., with non-concave distortion functions),
 and how large   the class of such tractable risk functionals is. 
 As far as we are aware, the unconstrained risk sharing problems for non-convex risk measures and variability measures have very limited explicit results (e.g., \cite{ELW18}, \cite{W18} and \cite{LMWW22}), and further investigation is needed for a better understanding of the challenges and their solutions.

\subsection*{Acknowledgments}
We thank the Editor, an Associate Editor, and two anonymous referees for helpful comments on the paper. 
Ruodu Wang is supported by the Natural Sciences and Engineering Research Council of Canada (CRC-2022-00141, RGPIN-2024-03728).



\appendix

\section{Proofs of results in   Section \ref{sec:3}}
\begin{proof}[Proof of Proposition \ref{prop:1}] 
(i)  Let $(X_1,\dots, X_n)$ be a Pareto-optimal allocation in $\mathbb A_n(X)$.  We will show, without loss of generality, that any of the three following hypotheses leads to a contradiction of the Pareto optimality of $(X_1, \dots, X_n)$: (1) if simultaneously $h_1(1)=0$ and $h_2(1)>0$;  (2) if simultaneously $h_1(1)<0$ and $h_2(1)>0$ and (3) if simultaneously $h_1(1)=0$ and $h_2(1)<0$. 

Consider the allocation $(X_1+c, X_2-c, X_3, \dots, X_n)$. Clearly, the allocation belongs to  $\mathbb A_n(X)$. 
Recall that by translation invariance it is $\rho_{h_1}(X_1+c)=\rho_{h_1}(X_1)+c h_1(1)$ and $\rho_{h_2}(X_2-c)=\rho_{h_2}(X_2)-c h_2(1)$.

Suppose (1) first so that $h_1(1)=0$ and $h_2(1)>0$. Setting $c>0$ we have that $\rho_{h_1}(X_1+c)=\rho_{h_1}(X_1)$ and $\rho_{h_2}(X_2-c)<\rho_{h_2}(X_2)$  contradicting the Pareto optimality of $(X_1,\dots,X_n)$. 
For (2), we have $\rho_{h_1}(X_1+c)<\rho_{h_1}(X_1)$ and $\rho_{h_2}(X_2-c)<\rho_{h_2}(X_2)$ as $h_1(1)<0$ and $h_2(1)>0$.
For (3), the case when $h_1(1)=0$ and $h_2(1)<0$, we can choose  $c<0$, which leads to a similar contradiction of the Pareto optimality of $(X_1,\dots,X_n)$.

The case when $(X_1,\dots, X_n)$ is Pareto optimal in $\mathbb{A}^+_n(X)$ is identical, and we conclude that $h_i(1)$ are either all zero, all positive, or all negative.

(ii) We show that if there exist $i,j \in [n]$ such that $h_i(1)\neq h_j(1)$, then $\dboxplus_{i=1}^n \rho_{h_i}(X)=-\infty$ for any $X \in \mathcal{X}$.
Without loss of generality,  let $h_1(1)<h_2(1)$ and consider a $c>0$. Given $X \in \X$, for any allocation $(X_1, \dots, X_n) \in \mathbb{A}_n^+(X)$ we have that  
$$\rho_{h_1}(X_1+c)+\rho_{h_2}(X_2-c)= \rho_{h_1}(X_1)+\rho_{h_2}(X_2)+c( h_1(1)-h_2(1)).$$
Consider now the allocation $(X_1+c, X_2-c, X_3, \dots, X_n)$. Taking the limit $c \to \infty$ we have $\sum_{i=1}^n \rho_{h_i}(X_i)=-\infty$ and so $\dboxplus_{i=1}^n \rho_{h_i}(X)=-\infty$.
\end{proof}

\begin{proof}[Proof of Theorem \ref{th:1}] 
For the ``if" part, it is  clear that $(X_1,\dots, X_n)$ is Pareto optimal for the normalized preferences $\rho_{\tilde h_1},\dots,  \rho_{\tilde h_n}$. Hence, $(X_1,\dots, X_n)$ is also Pareto optimal for agents using $\rho_{ h_1}, \dots, \rho_{ h_n}$ as their preferences, as the normalization does not change the preferences.

Next, we show the ``only if" part. 
Let $(X_1,\dots,X_n)\in \mathbb A_n(X)$ be a Pareto-optimal allocation in $\mathbb{A}_n(X)$. 
By Proposition \ref{prop:1}, we have $h_i(1)$, $i \in [n]$, are either all positive or all negative; that is $\tilde h_i(1)$, $i\in [n]$, are all $1$ or $-1$.  We first consider the case where $\tilde h_i(1)=1$ for $i\in [n]$. Assume by contradiction that $\sum_{i=1}^n \rho_{\tilde h_i}(X_i) >\dsquare_{{i=1}}^n\rho_{\tilde h_i}(X)$. There exists an allocation $(Y_1,\dots,Y_n)\in \mathbb A_n(X)$ such that $\sum_{i=1}^n\rho_{\tilde h_i}(Y_i)<\sum_{i=1}^n\rho_{\tilde h_i}(X_i)$. Set $c_i= \rho_{\tilde h_i}(X_i) -\rho_{\tilde h_i}(Y_i)$, $i=1,\dots,n$ and notice that $c=\sum_{i=1}^n c_i>0$. Hence, $$(Y_1+c_1-c/n,\dots, Y_n +c_n -c/n) \in  \mathbb A_n(X)$$ and by translation invariance for every $i\in [n]$ it is $$\rho_{\tilde h_i}(Y_i +c_i -c/n)=\rho_{\tilde h_i}(Y_i +c_i)-c/n<\rho_{\tilde h_i}(Y_i+c_i)=\rho_{\tilde h_i}(X_i),$$
contradicting the Pareto optimality of $(X_1,\dots,X_n)$. The case $\tilde h_i(1)=-1$, $i\in [n]$, is analogous.
\end{proof}

\begin{proof}[Proof of Lemma \ref{lem:rw1}]
 The implications (i)$\Rightarrow$(ii)$\Rightarrow$(iii)$\Rightarrow$(iv)  are all straightforward, where (iii)$\Rightarrow$(iv) follows from the fact that $X\le_{\rm cx} Y$ is equivalent to   $\rho_h(X)\le \rho_h(Y)$ holding for all concave $h\in \H^{\rm BV}$   by Theorem 2 of \cite{WWW20b}.

 We next show (iv)$\Rightarrow$(i). 
 Define
$$
L_X(t) =\int_0^t Q^-_s(X) \d s\mbox{~~~for $X\in \X$ and $t \in [0,1]$}.
$$ 
Let $h \in \mathcal H^{\rm BV}$ be strictly concave, $X\le_{\rm cx} Y$, and $\rho_h(X)=\rho_h(Y)$.
The convex order relation 
implies 
$ L_X(t) \le L_Y(t) 
$ for all $t\in [0,1]$  with equality at $t=0$ and $t=1$ (Theorem 3.A.5 of \cite{SS07}).
Note that $h$ 
is continuous  and almost everywhere differentiable on $(0,1)$ with possible discontinuity at $0$ and $1$.
Let $h'$ be the left derivative of $h$ on $(0,1)$, and let $a=\lim_{x\downarrow 0} h(x)-h(0)$ and $b=\lim_{x\uparrow 1}h(x)-h(1)$; both $a$ and $b$ are nonnegative due to concavity of $h$.

By decomposing $h$ into two possible jumps at $0,1$ and a continuous part, we  can write, 
using Lemma \ref{lem:qr}, 
\begin{align*}
\rho_h(X)= \int_{(0,1)}  Q_t^{-}(X) \d h(t) - b Q_1^+(X) + a Q_0^-(X) = \int_0^1 h'(t) Q_t^{-}(X) \d t - b Q_1^+(X) + a Q_0^-(X) .
\end{align*}
Since $h$ is concave, the convex order condition implies 
$
 \int_{(0,1)}  Q_t^{-}(X) \d h(t) 
 \le 
 \int_{(0,1)}  Q_t^{-}(Y) \d h(t) 
  $ (Theorem 2 of \cite{WWW20b}).
  Moreover, $Q_1^+(X)\ge Q_1^+(Y)$ and $Q_0^-(X)\le Q_0^-(Y)$.
  Therefore, in order for $\rho_h(X)=\rho_h(Y)$,   
 all three inequalities above are equalities. 
 In particular we have 
 $
  \int_0^1 h'(t) Q_t^{-}(X) \d t =  \int_0^1 h'(t) Q_t^{-}(Y) \d t.
 $
 It follows that
 \begin{align} \label{eq:R1-L2-1} 
  \int_0^1 h'(t) Q_t^{-}(X) \d t 
  =   \int_0^1 h'(t)  \d L_X(t)  =   \int_0^1 h'(t)  \d L_Y(t).
\end{align} 
Let us show $(L_X(t) -L_Y(t)) h'(t)\to 0$ as $t\downarrow 0$ or $t\uparrow 1$. 
Since $X$ and $Y$ are bounded, $L_X(t)-L_Y(t)$ is Lipschitz as $t\downarrow 0$ and $t\uparrow 1$.
Since $L_X(0)=L_Y(0)$ and $L_X(1)=L_Y(1)$, the Lipschitz property gives 
$
|L_X(t)-L_Y(t)| \le c ((1-t)\wedge t) 
$
for some $c>0$.  Note that $h'$ is integrable and strictly decreasing. Since $h'$ is decreasing, as $t\downarrow 0$, $h'(t)$  
        either has a finite limit or has the limit $\infty$.
        In the first case, clearly $  th'(t) \to 0.$
        In the second case, for $t>0$ small enough, we have $h'(t)>0$ and $\int_0^t h'(t) \d t \ge t h'(t)$. Integrability of $h'$ implies $\int_0^t h'(t) \d t\to 0$, which implies  $  th'(t)\to 0.$
       Similarly, $(1-t)h' (t) \to 0$ as $t\uparrow 1$.
Hence, we have the desired limits. 

Using integration by parts,  
\eqref{eq:R1-L2-1} and the above limits yield
$
\int_0^1   (L_X(t)- L_Y(t)  ) \d (-h')(t) = 0. $
Since $L_X(t)\le L_Y(t)$ for all $t\in (0,1)$ and $-h'$ is strictly increasing on $(0,1)$, we conclude that 
$L_X(t)=L_Y(t)$ for almost every $t\in (0,1)$. 
Since $L_X$ and $L_Y$ are continuous, this further implies that $L_X=L_Y$.
This is sufficient for $X\laweq Y$.
 \end{proof}

 \begin{proof}[Proof of Proposition \ref{pr:como}]
(i) follows from Corollary \ref{cor:rw1} observing that comonotonic improvements strictly improve welfare. 
For (ii), the ``only if" part is directly shown by (i). We only show the ``if" part. As the normalization of $h_i$, $i\in [n]$, will not change the preferences, we only consider the case when $a_i=a_j=a$ for all $i,j\in[n]$.
Let $(X_1,\dots,X_n)\in \mathbb A_n^+(X)$. By comonotonic additivity and positive homogeneity it is $\sum_{i=1}^n \rho_{a_i h_1}(X_i)= a\rho_{h_1}(X)$. 
Let $(Y_1,\dots, Y_n) \in \mathbb A_n(X)$. By subadditivity we have $
     \sum_{i=1}^n \rho_{a_i h_1}(Y_i)\geq a\rho_{h_1}\left(\sum_{i=1}^n Y_i\right)= a\rho_{h_1}(X).$
Hence, a comonotonic allocation $(X_1, \dots, X_n)$ always solves $\dsquare_{{i=1}}^n\rho_{a_ih}(X)$, and thus it is Pareto optimal.
\end{proof}

\begin{proof} [Proof of Theorem \ref{th:IQD}]
We first prove part (ii)  and then use it to prove part (i). 
Let us first verify $\dsquare_{i=1}^n (\lambda_i \mathrm{IQD}_{\alpha_i})  \ge \lambda \mathrm{IQD}_{\alpha}$.
Using \eqref{eq:IQD-1} and the fact that an IQD is nonnegative  if $\alpha <1/2$, then for $X\in \X$, 
\begin{align*}
\dsquare_{i=1}^n (\lambda_i \mathrm{IQD}_{\alpha_i}) &\ge 
\lambda \dsquare_{i=1}^n \mathrm{IQD}_{\alpha_i}(X) \\ &= 
\lambda  \inf\left\{\sum_{i=1}^n Q_{\alpha_i}^-(X_i)  + \sum_{i=1}^n Q_{\alpha_i}^- (-X_i): (X_1,\dots,X_n)\in \mathbb{A}_n(X) \right\}  
\\&\ge \lambda  \dsquare_{i=1}^n   Q_{\alpha_i}^- (X) + \lambda  \dsquare_{i=1}^n    Q_{{\alpha}_i}^- 
(-X) 
\\& =   \lambda  Q_{\sum_{i=1}^n \alpha_i }^- (X) +  \lambda     Q_{\sum_{i=1}^n \alpha_i }^- 
(-X) =  \lambda \mathrm{IQD}_{\alpha}(X),
  \end{align*}
  where   the second-last equality is due to Corollary 2 of \cite{ELW18}.
  If $\alpha \ge 1/2$, then $\dsquare_{i=1}^n (\lambda_i \mathrm{IQD}_{\alpha_i})  \ge 0 = \lambda \mathrm{IQD}_{\alpha}$ holds automatically. 
  
  Next, we verify  $\dsquare_{i=1}^n (\lambda_i \mathrm{IQD}_{\alpha_i})  \le \lambda \mathrm{IQD}_{\alpha}$ by showing that the construction of the allocation $(X_1,\dots,X_n)$ of $X\in \X$ in \eqref{eq:IQD_allocation} satisfies $\sum_{i=1}^n \lambda_i \mathrm{IQD}_{\alpha_i}(X_i) = \lambda \mathrm{IQD}_{\alpha}(X)$. 
  This will prove part (ii) as well as Remark \ref{rem:IQD}.
First, it is straightforward to verify $(X_1,\dots,X_n)\in \mathbb A_n(X)$.
Since IQD is location invariant, we can, without loss of generality, assume $c=c_1=\dots=c_n=0$; i.e., $0$ is a median of $X$. 
Note that this leads to the simplified form
$$ 
    X_i=  X\id_{A_i\cup B_i}+ a_i  X \left(1-\id_{A\cup B}\right),~~~~i\in [n].
$$
  
 If $\alpha\ge 1/2$, then $\beta = 1/2$ and $\p(A^c \cup B^c) = 0$. Thus, $\p(X_i > 0) = \p(A_i) \leq \alpha_i$ and symmetrically, $\p(X_i<0)\le \p(B_i)\le \alpha_i$. Hence, $\mathrm{IQD}_{\alpha_i}(X_i)=0$.
  
Next, assume $\alpha<1/2$. 
We have 
$ \p( \{ X > Q^{-}_\alpha(X)  \}\cap A ^c  )=0 $ by Lemma A.3 of \cite{WZ21}. 
For $i\in [n]$, if $a_i >0$, we can compute 
 \begin{align*}  
\p( X_i   >    a_i  Q^-_{\alpha}(X)   )  & \le   \p(A_i) + \p( \{ X_i >  a_i Q^-_{\alpha}(X) \} \setminus A_i)
=  \alpha_i + \p( \{ X > Q^{-}_\alpha(X)  \}\cap A ^c  )
=   \alpha_i.
\end{align*}
 If $a_i=0$, $X_i=X\id_{A_i\cup B_i}$. Hence,
$
\p( X_i   >    a_i  Q^-_{\alpha}(X)   )  =\p(X_i>0)\le \p(A_i)=  \alpha_i$. In conclusion, $\p( X_i   >    a_i  Q^-_{\alpha}(X)   ) \le \alpha_i$ for $i \in [n]$, which implies    
 $ 
 Q^-_{\alpha_i}  ( X_i)  \le  a_i Q^-_{\alpha}(X).
 $ 
Using a symmetric argument, we get 
 $
 Q^+_{1-\alpha_i}  ( X_i) \ge  a_i Q^+_{1-\alpha}(X).
 $
It follows that 
 $$
 \mathrm{IQD}_{\alpha_i}(X_i) \le a_i Q^-_{\alpha}(X)- a_i Q^+_{1-\alpha}(X)= a_i \mathrm{IQD}_{\alpha}(X).
 $$ 
 Therefore, $\sum_{i=1}^n \lambda_i \mathrm{IQD}_{\alpha_i}(X_i) \le  \sum_{i=1}^n \lambda_i a_i \mathrm{IQD}_{\alpha}(X)$.
 Taking $a_i=0$ for all $i\in [n]$ with $\lambda_i>\lambda$ gives the desired 
 inequality $\sum_{i=1}^n \lambda_i \mathrm{IQD}_{\alpha_i}(X_i) \le    \lambda  \mathrm{IQD}_{\alpha}(X)$. 
 
 Putting the above arguments together, 
 we prove (ii), that is, $\dsquare_{i=1}^n \lambda_i \mathrm{IQD}_{\alpha_i}(X) =   \lambda  \mathrm{IQD}_{\alpha}(X)$.
In particular, 
  \begin{align}\label{eq:IQD-check}
 \mathrm{IQD}_{\alpha_i}(X_i) = a_i \mathrm{IQD}_{\alpha}(X) ~~~\mbox{and} ~~~
 \sum_{i=1}^n  \mathrm{IQD}_{\alpha_i}(X_i) =  \mathrm{IQD}_{\alpha}(X)
 =   \dsquare_{i=1}^n  \mathrm{IQD}_{\alpha_i}(X),    \end{align}
 and thus $(X_1,\dots,X_n)$ is sum optimal.

Next, we show part (i).
The ``if" statement follows from Proposition \ref{pr:1}, and we will show the ``only if" statement.
Take any Pareto-optimal allocation $(Y_1,\dots,Y_n)$ of $X$. 
Write $x=\mathrm{IQD}_\alpha(X)$, $y_i= \mathrm{IQD}_{\alpha_i}(Y_i)$ for $i\in [n]$,
and $y=\sum_{i=1}^n y_i$. 
It suffices to show $y=x$.   By \eqref{eq:IQD-check}, we have that $x \le y$ always holds.
If $y=0$, there is nothing to show;
next we assume $y >0$.
 Let  $(X_1,\dots,X_n)$ be an allocation in \eqref{eq:IQD_allocation} with $a_i=y_i/y$  for $i\in [n]$, which sums up to $1$. By \eqref{eq:IQD-check},
we have $\mathrm{IQD}_{\alpha_i}(X_i) = a_i \mathrm{IQD}_\alpha(X) =a_i x$.
If $x<y$, then $\mathrm{IQD}_{\alpha_i}(X_i)= x y_i/y \le y_i=\mathrm{IQD}_{\alpha_i}(Y_i) $ for $i\in [n]$, and strict inequality holds 
as soon as $y_i>0$.  In this case, $(X_1, \dots, X_n)$ Pareto dominates $(Y_1, \dots, Y_n)$,  conflicting Pareto optimality of $(Y_1,\dots,Y_n)$.
Hence, we obtain $x=y$.

Finally, part (iii) on Pareto optimality of $(X_1,\dots,X_n)$ follows by combining (i) and \eqref{eq:IQD-check}.
\end{proof}

\section{Proofs of results in Section \ref{sec:4}}
\begin{proof}[Proof of Proposition \ref{pr:2}]
(i) $\Rightarrow$ (ii) is analogous to Theorem \ref{th:1}.

(ii) $\Rightarrow$ (iii) 
Let $S=\{(\rho_{h_1}(X_1), \dots, \rho_{h_n}(X_n)): (X_1,\dots,X_n) \in \mathbb{A}_n^+(X)\}$ be the risk possibility set of the set of comonotonic allocations.
By comonotonic additivity and positive homogeneity, $S$ is a convex set. Let $(X_1, \dots, X_n)$ be a Pareto-optimal allocation and $\mathbf x=(\rho_{h_1}(X_1),\dots,\rho_{h_n}(X_n))$.
 Notice now that $\mathbf x$ always belongs to the boundary of $S$.
Let $V=\{(v_1, \dots, v_n): v_i \le \rho_{h_i}(X_i) \text{ for } i \in [n]\}$ where $(X_1, \dots, X_n)$ is Pareto optimal. It is clear that  $V \cap S=\{\mathbf x\}$.

Therefore, by the Separating Hyperplane Theorem, there exists $(\lambda_1,\dots, \lambda_n) \in \R^n \setminus \mathbf{0}$ such that $\sum_{i=1}^n \lambda_i\rho_{h_i}(X_i)=\inf_{\mathbf u \in S} \sum_{i=1}^n\lambda_iu_i=\dboxplus_{i=1}^n {\lambda_i\rho_{h_i}(X)}$.

We are left to show that $\lambda_i\ge 0$ for every $i \in [n]$. Let $\mathbf v=\mathbf x-(1,0,\dots,0)$. We have $\mathbf v \in V$. Hence,
we have $\lambda_1\ge 0$ as $\sum_{i=1}^n \lambda_i v_i\le\sum_{i=1}^n \lambda_i\rho_{h_i}(X_i)$. Similarly, we obtain $\lambda_i\ge 0$ for all $i \in [n]$.
\end{proof}

\begin{proof}[Proof of Theorem \ref{th:3}]

Let $h_\wedge(t)=\min\{h_1(t),\dots,h_n(t)\}$. We first show that 
$
\dboxplus_{i=1}^n\rho_{h_i}= \rho_{h_\wedge}.
$  Comonotonic additivity of $\rho_{h_\wedge}$ implies that 
$\dboxplus_{i=1}^n\rho_{h_i}\ge \rho_{h_\wedge}.$
Conversely, notice that for every $i\in [n]$ the function $f_i$ in \eqref{eq:f} is Lipschitz continuous and non-decreasing because $g_i$ is nonnegative and bounded.
Using Lemma \ref{lem:change}, we get 
\begin{align}\label{eq:como-unique}
\rho_{h_i}(f_i(X))=\int_0^\infty g_i(s)  h_i(\p(X>s))\d s+\int_{-\infty}^0 g_i(s)  (h_i(\p(X>s)-h_i(1))\d s.\end{align}
It follows that
    \begin{align*}\sum_{i=1}^n \rho_{h_i}(f_i(X))&=\sum_{i=1}^n \int_0^\infty g_i(s)  h_i(\p(X>s))\d s+\int_{-\infty}^0 g_i(s)  (h_i(\p(X>s)-h_i(1))\d s\\
    &= \int_0^\infty \sum_{i=1}^n g_i(s)  h_i(\p(X>s))\d s+\int_{-\infty}^0 \sum_{i=1}^n g_i(s)  (h_i(\p(X>s)-h_i(1))\d s\\
    &= \int_0^\infty  {h_\wedge}(\p(X>s))\d s+\int_{-\infty}^0   ({h_\wedge}(\p(X>s)-{h_\wedge}(1))\d s\\
    &=\rho_{h_\wedge}(X) \ge \dboxplus_{i=1}^n\rho_{h_i} (X).
    \end{align*}
Hence, $\dboxplus_{i=1}^n\rho_{h_i}= \rho_{h_\wedge}$.

Next, we show that the solution is unique up to constant shifts almost surely if and only if $|M_x|=1$ for $\mu_X$-almost every $x$, where $\mu_X$ is the distribution measure of $X$.

 Since the above argument of $\sum_{i=1}^n \rho_{h_i}(f_i(X))= \dboxplus_{i=1}^n\rho_{h_i} (X) $  only requires $\sum_{i\in M_x} g_i(x)=1$ for almost every $x$, 
any allocation $(f_1(X),\dots,f_n(X))$ in \eqref{eq:f}  with $g_i$ replaced by
$$
g_i(x)=   \id_{\{ i = \min M_x\}} \mbox{ or } g_i(x)=   \id_{\{ i = \max M_x\}}, ~~~~x\in \R, 
 $$
also satisfies  sum optimality. 
Therefore, if $|M_x|=1$ does not hold $\mu_X$-almost surely, there are multiple optimal allocations that are not constant shifts from each other.

Conversely, we show that if $|M_x|=1$ for $\mu_X$-almost every $x$ then every sum-optimal allocation is almost surely equal to the one in \eqref{eq:f}. 

For any increasing and Lipschitz function $k$ with right-derivative $w$ and two distortion functions with $h\ge g$ and $h(1)=g(1)$, we have, by Lemma \ref{lem:change}, 
 $$\rho_h(k(X)) - \rho_g (k(X)) 
= \int_{-\infty}^ {\infty} w(s) (h(\p(X>s)) -g(\p(X>s)) )\d s.$$ 
This means $\rho_h(k(X)) = \rho_g (k(X))$  if and only if $w(s)=0$ almost surely for $s$ such that $h(\p(X>s))>g(\p(X>s))$.
Note that  if $(k_1(X),\dots,k_n(X))\in \mathbb A^+_n(X)$ is sum optimal, then 
$$\sum_{i=1}^n \rho_{h_i}(k_i(X))=\rho_{h_\wedge}(X)=\sum_{i=1}^m \rho_{{h_\wedge}}(k_i(X)).$$
This implies that $w_i(x)=0$ as soon as $h_i(\p(X>x)) >h_\wedge(\p(X>x)) $, where $w_i$ is the right-derivative of $k_i$.
Moreover, $w_i(x)=1$ if $h_i(\p(X>x)) = h_\wedge(\p(X>x)) $ since $\sum_{j=1}^n w_j(x)=1$ for almost every $x$. 
Thus, $w_i$ is uniquely determined $\mu_X$-a.s., implying that $k_i$ is unique $\mu_X$-a.s.~up to a constant shift.
\end{proof}

\begin{proof}[Proof of Lemma \ref{lem:change}]
Without loss of generality we assume $X\ge 0$ and $f(X)\ge 0$.
Denote by $\nu=h\circ \p$.
We have 
\begin{align*} 
\rho_{h}(f(X)) - \int_0^\infty g(x)  h(\p(X>x))\d x
&= \int_0^\infty  \nu  (f(X)>y) \d y  - \int_0^\infty g(x) \nu(X>x)\d x
\\
&= \int_0^\infty g(x)  \nu (f(X)>f(x) )\d x  - \int_0^\infty g(x)   \nu (X>x)\d x 
\\
&= \int_0^\infty g(x)   ( \nu (f(X)>f(x) ) - \nu(X>x))\d x.
\end{align*}
Note that  $\p(f(X)>f(x))\le  \p(X>x)$ for all $x$.
If  $\p(f(X)>f(x))< \p(X>x)$, then there exists $z>x$ such that $f(z)=f(x)$.
This implies that $g(x)=0$ for any point $x$ with  $\nu (f(X)>f(x) ) - \nu(X>x) \ne 0$. Therefore, 
$$
\rho_{h}(f(X)) - \int_0^\infty g(x)  h(\p(X>x))\d x=0.$$
The case of general $X$  bounded from below can be obtained by constant shifts on both $X$ and $f$.
\end{proof}

 \begin{proof}[Proof of Proposition \ref{pr:charac_location_inv}]
It is clear that since $h_i(1)=0$ and $h_i(t)>0$ for all $i\in [n]$ and all $t \in (0,1)$, we have that $\rho_{h_i}(X)\ge 0$ for all $i\in [n]$, with equality only if $X$ is a constant. 
We first show the ``if" statement. Suppose, by contradiction, that $(X_1,\dots,X_n)\in \mathbb A^+_n(X)$ is not Pareto optimal but that it solves $\dboxplus_{i\in K}\rho_{\lambda_i h_i}(X-\sum_{i\notin K} X_i)$ for $K=\{i\in[n]: X_i\notin \mathbb R\}$ and $\boldsymbol \lambda \in  (0,\infty)^{K}$. Our contradiction hypothesis implies that there exists a $(Y_1, \dots, Y_n) \in  \mathbb A^+_n(X)$ such that simultaneously $\rho_{h_i}(Y_i)\le \rho_{h_i}(X_i)$ for every  $i \in [n]$ and $\rho_{h_j}(Y_j)< \rho_{h_j}(X_j)$ for some $j \in [n]$. Notice that if $i\notin K$ it is $$0\leq\rho_{h_i}(Y_i)\leq \rho_{h_i}(X_i)=0$$ and so it must be the case that $\rho_{h_i}(Y_i)<\rho_{h_i}(X_i)$ for some $i\in K$, a contradiction with the hypothesis that $(X_i)_{i\in K}$ solves $\dboxplus_{i\in K}\rho_{\lambda_i h_i}(X-\sum_{i\notin K} X_i)=\dboxplus_{i\in K}\rho_{\lambda_i h_i}(X)$, where the equality follows because of location invariance of $\dboxplus_{i\in K}\rho_{\lambda_i h_i}$.

Conversely, let $(X_1,\dots,X_n)\in \mathbb A^+_n(X)$ be Pareto optimal and define $K=\{i\in [n]: X_i\notin \R\}$; this gives that $\sum_{i\not \in K} X_i$ is a constant. Recall that $\rho_{h_i}(X_i)=0$ for every $i\notin K$, and $\rho_{h_i}(X_i)>0$ for every $i \in K$. 
It is clear that $(X_i)_{i \in K}$ is a Pareto-optimal allocation of $X-\sum_{i\notin K} X_i$ for the collection $(\rho_{h_i})_{i\in K}$. By Proposition \ref{pr:2}, there exists a $\boldsymbol \lambda \in [0, \infty)^{K}\setminus \{\mathbf 0\}$ such that $\sum_{i\in K} \lambda_i  \rho_{h_i}(X_i) =\dboxplus_{{i\in K}}(\lambda_i \rho_{h_i})(X-\sum_{j\notin K} X_j)=\dboxplus_{i\in K}\rho_{\lambda_i h_i}(X)$. As $\rho_{h_i}(X_i)>0$ for $i\in K$, we have $\dboxplus_{{i\in K}}(\lambda_i \rho_{h_i})(X)>0$.
 It must be the case that $\lambda_i>0$ for all $i \in K$, as otherwise, we have $\dboxplus_{{i\in K}}(\lambda_i \rho_{h_i})(X)=0$, a contradiction. 
\end{proof}

\begin{proof}[Proof of Proposition \ref{prop:9}]
Part (ii) follows directly from Corollary \ref{cor:th3}, so it remains to show part (i).
Suppose that $(X_1,\dots,X_n)\in \mathbb{A}^+_n(X)$ is Pareto optimal. Then there exists $(\lambda_1,\dots,\lambda_n)\in [0,\infty)^n$, 
with $\lambda= \bigwedge_{i=1}^n \lambda_i>0$, such that 
$$\sum_{i=1}^n (\lambda_i \mathrm{IQD}_{\alpha_i}) (X_i) =  \lambda \mathrm{IQD}_{\bigvee_{i=1}^n \alpha_i} (X).  $$ 
Using the fact that an IQD is nonnegative and part (ii), we get 
$$\lambda \dboxplus_{i=1}^n  \mathrm{IQD}_{\alpha_i} (X)\le \sum_{i=1}^n (\lambda \mathrm{IQD}_{\alpha_i}) (X_i)\le \sum_{i=1}^n (\lambda_i \mathrm{IQD}_{\alpha_i}) (X_i) =  \lambda \mathrm{IQD}_{\bigvee_{i=1}^n \alpha_i} (X)=   \lambda \dboxplus_{i=1}^n  \mathrm{IQD}_{\alpha_i} (X), $$
and so $(X_1,\dots,X_n)$ is sum optimal.   
\end{proof}

\section{Proofs of results in Section \ref{sec:new5}}
 
We first present a lemma that we will use in the proof of  Theorem \ref{pr:IQD_h}.
\begin{lemma}\label{lem:IQD}
    For $\alpha \in [0, 1/2)$, $\lambda>0$ and $h\in \mathcal{H}^{\mathrm{C}}$ it is
  \begin{equation}\label{eq:IQDMMD}
  (\lambda \mathrm{IQD}_{\alpha} ) \square \rho_h =  \rho_{G^\alpha_\lambda(h)}.
\end{equation}
\end{lemma} 
\begin{proof}[Proof of Lemma \ref{lem:IQD}]

We first verify that $\lambda\mathrm {IQD}_\alpha(X_1)+\rho_h(X_2)\ge \rho_{G^\alpha_\lambda(h)}(X)$ for any $(X_1, X_2) \in \mathbb A_2(X)$.
As both $\mathrm{IQD}_\alpha$ and $\rho_h$ are location invariant, we can, without loss of generality, assume the allocation $(X_1, X_2)$ satisfies $Q_{1/2}^-(X_1)=0$.
Let $A$ be a right $\alpha$-tail event of $X_1$  and $B\subseteq A^c$ be a left $\alpha$-tail event of $X_1$. 
Hence, $\p(A)=\p(B)=\alpha$ and $X_1(\omega_B)\le X_1(\omega)\le X_1(\omega_A)$ a.s.~for $\omega_A \in A$, $\omega_B \in B$ and $\omega \in (A \cup B)^c$. Let $X^*_1=X_1\id_{(A \cup B)^c}$ and $h^*=h\wedge \lambda$.
Recall that $\mathrm{IQD}_0=Q_0^--Q_1^+$ is the range functional.
It is straightforward to verify that $\mathrm{IQD}_\alpha(X_1)=\mathrm{IQD}_0(X^*_1)$ and that $h^*$ is concave. 
Further, notice that $\lambda\mathrm{IQD}_0\ge \rho_{h^*}$, $\rho_{h}\ge \rho_{h^*}$ and $\rho_{h^*}$ is subadditive.
Therefore,
\begin{align*}
\lambda\mathrm{IQD}_{\alpha}(X_1)+\rho_h(X_2)= \lambda\mathrm{IQD}_0(X^*_1)+\rho_h(X_2)\ge \rho_{h^*}(X^*_1)+\rho_{h^*}( X_2)\ge \rho_{h^*}\left( X^*_1+X_2\right).
\end{align*}

 As $ Q_{1/2}^-(X_1)=0$, we   have, in the a.s.~sense, $X_1 \ge 0$ on $A$ and $X_1 \le 0$  on $B$; that is,  $X^*_1 +X_2 =X $  on $(A\cup B)^c $,  $X^*_1 +X_2 \ge X $  on $B$, and  $X^*_1  +X_2 \le  X $  on $A $. For any $x \in \R$, we have
\begin{align*}\p\left( X^*_1+X_2>x\right)&\ge \p\left(X>x,\, (A\cup B)^c\right)+\p\left(X>x,\, B\right)\ge \p\left(X>x\right)-\p\left(A\right)=\p\left(X>x\right)-\alpha,
\end{align*}
and similarly,
$\p\left(X^*_1+X_2\le x\right)\ge \p\left(X\le x\right)-\alpha.
$
Therefore, 
 $$ \p\left(X> x\right)-\alpha\le \p\left( X^*_1+X_2>x\right)\le \p\left(X> x\right)+\alpha. $$
Let $s\in \R$ be such that 
$x\mapsto h^*(\p(X^*_1+X_2>x ))$
is increasing on $(-\infty,s]$
and decreasing on $[s,\infty)$. Such $s $ exists since $h^*$ is first increasing and then decreasing. By treating $h^*(t)=0$ if $t$ is outside $[0,1]$, we have 
\begin{align*}
\rho_{h^*} (X_1^*+X_2)
&=
\int_{-\infty}^s  h^*(\p(X^*_1+X_2>x )  )\d x
+ \int_{s}^\infty  h^*(\p(X^*_1+X_2>x  ) )\d x
\\ &\ge \int_{-\infty}^s  h^*(\p(X>x ) -\alpha )\d x
+ \int_{s}^\infty  h^*(\p( X >x  )+\alpha, 1)\d x
\\& \ge \int_{-\infty}^\infty \min\left \{  h^*\left(\p(X >  x)  +\alpha\right),
   h^*(\p(X > x) -\alpha) \right\} \d x
  = \rho_{G_{\lambda}^\alpha(h) }(X),
\end{align*} 
where the last equality follows because 
$$
G_{\lambda}^\alpha(h)(t)
= (h(t-\alpha)\wedge h(t+\alpha) \wedge \lambda)
\id_{\{\alpha<t<1-\alpha\}}
=h^*(t-\alpha)\wedge h^*(t+\alpha)
$$
and $G_{\lambda}^\alpha(h)(1)=0$. 
Therefore, we have $\lambda\mathrm {IQD}_\alpha(X_1)+\rho_h(X_2)\ge \rho_{G^\alpha_\lambda(h)}(X)$.

Next, we give an allocation $(X_1, X_2)\in\mathbb A_2(X)$ that   attains the lower bound  $\rho_{G^\alpha_\lambda(h)}(X)$. 
Define the function $f(s) = h^*(\p(X>x)+\alpha) - h^*(\p(X>x)-\alpha)$ where $h^*(t)=0$ if $t$ is outside $[0,1]$.
Since $h^*$ is concave,
the function $s\mapsto f(s)$
is increasing on the set of $s$ with $\p(X>s)\in [\alpha,1-\alpha]$.
Moreover, $f(s)\le 0$ for $s\le  Q_{1-\alpha}^-(X) $  and $f(s)\ge 0$ for $s\ge Q_{\alpha}^+(X).$ 
Hence, 
there exists $s^* \in [  Q_{1-\alpha}^-(X),  Q_{\alpha}^+(X)]$ such that $f(s)\ge 0$ for $s < s^*$
and $f(s)\le 0$ for $s>s^*$.

Let $A$ be a right $\alpha$-tail event of $X$  and $B\subseteq A^c$ be a left $\alpha$-tail event of $X$.
Write $T=A\cup B$.
Let $(Y_1, Y_2) \in \mathbb{A}^+_2(X\id_{ T^c}+s^*\id_{T})$ be a $(\lambda,1)$-optimal allocation for $(\mathrm{IQD}_{\alpha}, \rho_h)$.
Define $X_1=(X-s^*)\id_{ T }+Y_1$ and $X_2=Y_2$; clearly $(X_1,X_2)\in \mathbb A_2(X)$.  By Theorem \ref{th:3}, we have
$$\lambda\mathrm{IQD}_{\alpha}(X_1)+\rho_h(X_2)= \lambda\mathrm{IQD}_0(Y_1)+\rho_h(Y_2)= \rho_{h^*}(X\id_{T^c}+s^*\id_{T}).$$
Note that 
\begin{align*}
\rho_{h^*}(X\id_{T^c}+s^*\id_{T}) 
&=
\int_{Q_{1-\alpha}^-(X)}^{Q_{\alpha}^+(X)}
h^*(\p(X\id_{T^c}+s^*\id_{T}>x) ) \d x
\\&= \int_{Q_{1-\alpha}^-(X)}^{s^*}
h^*(\p(X>x,~T^c) + 2\alpha ) \d x + \int^{Q_{\alpha}^+(X)}_{s^*}
h^*(\p(X>x,~T^c)) \d x
\\&= \int_{Q_{1-\alpha}^-(X)}^{s^*}
h^*(\p(X>x ) + \alpha ) \d x + \int^{Q_{\alpha}^+(X)}_{s^*}
h^*(\p(X>x ) - \alpha ) \d x 
\\&= \int_{Q_{1-\alpha}^-(X)}^{Q_{\alpha}^+(X)}
\min\{ h^*(\p(X>x ) + \alpha ), 
h^*(\p(X>x ) - \alpha )\} \d x  
=\rho_{G^\alpha_\lambda(h)}(X),
\end{align*} 
where the second-last equality is due to the definition of $s^*$.
Therefore, the lower bound $\rho_{G^\alpha_\lambda(h)}(X)$ can be attained. 
Thus, $(\lambda \mathrm{IQD}_{\alpha} ) \square \rho_h =\rho_{G^\alpha_\lambda(h)}(X)$. 
\end{proof}

\begin{proof}[Proof of Theorem \ref{pr:IQD_h}]
As the cases $I=[n]$ and $N=[n]$ follow from  Theorems \ref{th:IQD} and    \ref{th:3} respectively, we assume that the sets $I$ and $N$ are non-empty.

(i) The equality  $\dsquare_{i=1}^n (\lambda_i \rho_{h_i}) = \rho_{G^\alpha_\lambda(h)}$ follows from Lemma \ref{lem:IQD}, Theorems \ref{th:IQD} and \ref{th:3}, and
the fact that the inf-convolution is associative (Lemma 2 of \cite{LWW20}), which together yield
$$
\dsquare_{i=1}^n (\lambda_i \rho_{h_i})
= \left(\dsquare_{i\in I}  (\lambda_i \rho_{h_i})\right) \square 
\left(\dsquare_{i\in N} (\lambda_i \rho_{h_i}) \right)
= (\lambda \mathrm{IQD}_\alpha) \square \rho_{h} = \rho_{G^\alpha_\lambda(h)}.
$$

(ii)   Without loss of generality, we assume $c=c_1=\dots=c_n=0$ and let $Y=X\id_{(A\cup B)^c}$. If $\alpha\ge 1/2$, it is straightforward to check that $(X_1, \dots, X_n)$ is Pareto optimal as $\rho_{h_i}(X_i)=0$ for $i\in [n]$. Now, we assume $\alpha<1/2$.

We first show that $\rho_{h_i}(X_i)\le \rho_{\hat h_i}(Y_i)$ for all $i\in [n]$. Note that  $\rho_{h_i}(X_i)=\rho_{h_i}(Y_i)=\rho_{\hat h_i}(Y_i)$ for  all $i\in N$. We are left to show
$\mathrm{IQD}_{\alpha_i}(X_i)\le \mathrm{IQD}_0(Y_i)$ 
for all $i\in I$. As $X(\omega)\le 0$ a.s.~for $\omega \in B_i$,  
$$\p(X_i \le Q^-_0(Y_i))=\p(X\id_{A_i\cup B_i}+Y_i\le Q^-_0(Y_i))\ge \p(B_i)+\p((A_i\cup B_i)^c)=\alpha_i+1-2\alpha_i=1-\alpha_i.$$ That is, $Q^-_{\alpha_i}(X_i) \le Q^-_0(Y_i)$. Similarly,  $Q^+_{1-\alpha_i}(X_i) \ge Q^+_1(Y_i)$. Hence,  $\rho_{h_i}(X)=\mathrm{IQD}_{\alpha_i}(X_i)\le \mathrm{IQD}_0(Y_i)=\rho_{\hat h_i}(Y_i)$ 
for all $i\in I$.

 Let $(Y_1',\dots,Y_n')$ be a 
  comonotonic improvement of $(Y_1,\dots,Y_n)$.
  The definition of comonotonic improvement and 
  Pareto optimality of $(Y_1,\dots,Y_n)$   imply that $\rho_{\hat h_i}(Y_i)=\rho_{\hat h_i}(Y'_i)$ for all $i\in [n]$. 
  First, if there exist some $i\in N$ such that 
  $h_i(t)=0$ on $[0,1]$, then 
  Pareto optimality of $(Y'_1,\dots,Y'_n)$ implies 
  that $\rho_{\hat h_i}(Y'_i)=0$ for each $i\in [n]$.
  This in turn implies that $\rho_{h_i}(X_i)=0$ for each $i\in [n]$, and hence $(X_1,\dots,X_n)$ is Pareto optimal. Below, we assume for each $i\in N$, $h_i(t)>0$
  for some $t\in (0,1)$, which gives that $h_i(t)>0$ for all $t\in (0,1)$ due to concavity.

As $\hat h_i(1)=0$ and $\hat h_i(t)>0$ for all $i\in [n]$ and $t\in (0,1)$, by Proposition \ref{pr:charac_location_inv}, Pareto optimality  of $(Y_1',\dots,Y_n
')$    implies that there exist $K\subseteq  [n]$ and a vector $\boldsymbol \lambda \in  (0,\infty)^{K}$ such that $(Y'_i)_{i\in K}$ 
solves $\dboxplus_{i\in K}\rho_{\lambda_i \hat h_i}(Y)$, and $Y'_i$, $i\not\in K$ are constants.
Denote by $h^*= \bigwedge_{i\in N \cap K} (\lambda_i h_i)$ and $\lambda^*= \bigwedge_{i\in I\cap K}  \lambda_i>0$; here, we set $\inf \varnothing =\infty$. 
Putting together several observations above, we get 
\begin{align}
\sum_{i\in K} \lambda_i\rho_{h_i}(X_i)\le \sum_{i\in K} \rho_{\lambda_i  \hat h_i}(Y_i)=\sum_{i\in K} \rho_{\lambda_i \hat h_i}(Y'_i)=\dboxplus_{i\in K}\rho_{\lambda_i \hat h_i}(Y)=\rho_{h^* \wedge \lambda^*}(Y),
\label{eq:pf-th4}
\end{align}
where the first inequality holds because $\rho_{h_i}(X_i)\le \rho_{\hat h_i}(Y_i)$ for all $i\in [n]$, the first equality holds because $\rho_{\hat h_i}(Y_i)=\rho_{\hat h_i}(Y'_i)$ for all $i\in [n]$,
the second equality is due to $\boldsymbol \lambda$-optimality of $(Y_i')_{i\in K}$ whose component-wise sum is $Y$ plus a constant, and the last equality is due to Theorem \ref{th:3}.
Furthermore, for $i\notin K$, we have $ 0\le \rho_{h_i}(X_i)\le \rho_{\hat h_i}(c_i)=0$; that is $\rho_{h_i}(X_i)=0$. 
Note that 
\begin{align}
\rho_{h^* \wedge \lambda^*}(Y)= 
\rho_{h^* \wedge \lambda^*}(X\id_{(A\cup B)^c})= \rho_{G_{\lambda^*}^\alpha(h^*)}(X).\label{eq:pf-th4-2}
\end{align}
Take $\beta\ge \lambda^*$.
If $i\in N\setminus K$, then 
$X_i=Y_i'$ is a constant.
Write $Z=\sum_{i\in I\cup K} X_i$.
Using \eqref{eq:pf-th4} and \eqref{eq:pf-th4-2}, we get 
\begin{align}
\sum_{i\in K} \lambda_i\rho_{h_i}(X_i) + 
\sum_{i\in I\setminus K} \beta\rho_{h_i}(X_i)&\le \rho_{G_{\lambda^*}^\alpha(h^*)}(X)  =\rho_{G_{\lambda^*}^\alpha(h^*)}\left(X-\sum_{i\in N\setminus K} X_i\right)
=\rho_{G_{\lambda^*}^\alpha(h^*)}(Z).
\label{eq:pf-th4-3}
\end{align}
Using part (i), we have 
$$\left(\dsquare_{i\in K} ( \lambda_i\rho_{h_i} )\right) \square \left(\dsquare_{i\in I\setminus K} ( \beta \rho_{h_i} )\right) =\rho_{G_{\lambda^*}^\alpha (h^*)}.
$$
Therefore, \eqref{eq:pf-th4-3} implies  that $(X_i)_{i\in I\cup K}\in \mathbb A_n(Z) $ minimizes  $\sum_{i\in K} \lambda_i\rho_{h_i}(X_i) + 
\sum_{i\in I\setminus K} \beta\rho_{h_i}(X_i)$. 
Since also $\rho(X_i)=0$ for $i\notin K$, 
we conclude that $(X_1,\dots,X_n)$ is Pareto optimal.
\end{proof}

\section{Proof of of results in Section \ref{sec:7}}

\begin{proof}[Proof of Lemma \ref{lemma:absolute_continuity}]
Let $g(t)=h(\P_0(X>Q^{\p}_{t}(X)))$ for $t \in [0,1]$, where $Q^{\p}_{t}(X)$ is the  left quantile under the measure $\p$. We first show that $g(\p(X>x))=h(\P_0(X>x))$ for all $x\in \R$. 
It is clear that $
    g(\p(X>x))=h(\P_0(X>Q^{\p}_{\p(X> x)}(X))).
$
By the definition of $Q_{t}^\p$, we have 
$Q^{\p}_{\p(X> x)}(X)\le x$. For $x\in \R$, if $Q^{\p}_{\p(X> x)}(X)=x$, then it is clear that  $g(\p(X>x))=h(\P_0(X>x))$. If $Q^{\p}_{\p(X> x)}(X)<x$, we have $\p(Q^{\p}_{\p(X> x)}(X)<X\le x)=0$.
As $\P_0 \ll \p$,  we have $\P_0(Q^{\p}_{\p(X> x)}(X)<X\le x)=0$.
Hence, $$h(\P_0(X>Q^{\p}_{\p(X> x)}(X)))=h\left(\P_0(x \ge X>Q^{\p}_{\p(X> x)}(X))+\P_0(X>x)\right)=h(\P_0(X>x)).$$
Taking $t \uparrow 1$, we obtain $g(1)=h(1)$.

Next, we will show $\rho_h^{\P_0}(f(X))= \rho_g^\p(f(X))$ for any increasing function $f:\R \to \R$. Denote by $f^{-1}(x)=\inf\{y: f(y)>x\}$ the pseudo-inverse function of $f$. As $\p(X=x)=0$ for all $x\in \R$ and $\P_0\ll \p$, we have $\p(X=x)=\P_0(X=x)$ for all $x\in \R$. Hence,
$$ 
\begin{aligned}
\rho_h^{\P_0}(f(X))=&\int_0^\infty h(\P_0(f(X) > x))\d x + \int_{-\infty}^{0} (h(\P_0(f(X) > x))-h(1) )\d x \\
&=\int_0^\infty h\left(\P_0(X > f^{-1}(x))+\P_0(X = f^{-1}(x))\id_{\{f(f^{-1}(x))>x\}}\right)\d x \\
&\quad + \int_{-\infty}^{0} \left(h\left(\P_0(X > f^{-1}(x))+\P_0(X = f^{-1}(x))\id_{\{f(f^{-1}(x))>x\}}\right)-h(1) \right)\d x \\
&=\int_0^\infty h\left(\P_0(X > f^{-1}(x))\right)\d x + \int_{-\infty}^{0} (h(\P_0(X > f^{-1}(x)))-h(1) )\d x \\
&=\int_0^\infty g(\p(X > f^{-1}(x)))\d x + \int_{-\infty}^{0} (g(\p(X > f^{-1}(x)))-h(1) )\d x = \rho_g^\p(f(X)),
\end{aligned}$$
as desired.
\end{proof}

\begin{proof}[Proof of Proposition \ref{pr:heterogeneous_beliefs}]
Let $\P=1/n \sum_{i=1}^n \P_i$ and $g_i(t)=h_i(\P_i(X>Q^{\P}_{t}(X)))$ for $t \in [0,1]$. It is clear that $X$ also has a density function under $\P$ and $\rho_{h_i}^{\P_i}(f(X))=\rho_{g_i}^{\P}(f(X))$ for increasing functions $f$ and $i\in [n]$ by Lemma \ref{lemma:absolute_continuity}. Hence, $(\rho_{h_1}^{\P_1}, \dots, \rho_{h_n}^{\P_n})$ and $(\rho_{g_1}^{\P}, \dots, \rho_{g_n}^{\P})$ have the same class of Pareto-optimal allocations.
\end{proof}

\section{Omitted details  in Section \ref{sec:5}} \label{app:C}
We present the functions  $G_{{\lambda}}^\alpha(h)$ for Cases 1 to 6 in Section \ref{subsec:IQD_and_RA} which yield the allocations that we present in that section.

\noindent\textbf{Case 1:}
When $c_1\ge 1/2$ and $c_3\ge 1/2$ it is
$$G_{{\lambda}}^\alpha(h)(t)=
\lambda_2\left((t-\alpha)\wedge(1-t-\alpha)  \right)\id_{\{\alpha < t < 1-\alpha\}}.$$ 

\noindent\textbf{Case 2:} 
When $c_2\ge 1/2$ and $c_3\le \alpha$ it is
$$G_{{\lambda}}^\alpha(h)(t)=\lambda_3\left([(t-\alpha)(1+\alpha-t)] \wedge[(t+\alpha)(1-\alpha-t)]  \right)\id_{\{\alpha < t< 1-\alpha\}}.$$ 

\noindent \textbf{Case 3:} 
When either 
$\alpha<c_2<c_3<1/2$ or
$\alpha<c_1<1/2<c_3$
it is 
$$G_{{\lambda}}^\alpha(h)(t)=\left(\lambda_2[(t-\alpha)\wedge(1-t-\alpha)] \wedge\lambda_1 \right)\id_{\{\alpha < t< 1-\alpha\}}.$$ 

\noindent \textbf{Case 4:} 
When $c_3\le \alpha<c_2<1/2$  it is
$$G_{{\lambda}}^\alpha(h)(t)=\left(\lambda_3[(t-\alpha)(1+\alpha-t)] \wedge[(t+\alpha)(1-\alpha-t)]\wedge\lambda_1  \right)\id_{\{\alpha < t < 1-\alpha\}}.$$ 

\noindent \textbf{Case 5:} When $\alpha<c_3<1/2< c_2$, it is
$$G_{{\lambda}}^\alpha(h)(t)=\left\{
\begin{aligned}
&0, & t\in [0, \alpha]\cup[1-\alpha,1],\\
&\lambda_2(t-\alpha), & t \in (\alpha, c_3),\\
&\lambda_3(t-\alpha)(1-t+\alpha), & t \in [ c_3, 1/2),\\
&\lambda_3(t+\alpha)(1-t-\alpha), & t \in [ 1/2, 1-c_3),\\
&\lambda_2(1-\alpha-t), & t \in [ 1-c_3, 1-\alpha).
\end{aligned}
\right.$$

\noindent \textbf{Case 6:} 
When $\alpha<c_3\le c_2<1/2$ it is

$$G_{{\lambda}}^\alpha(h)(t)=\left\{
\begin{aligned}
&0, & t\in [0, \alpha]\cup[1-\alpha,1],\\
&\lambda_2(t-\alpha), & t \in (\alpha, c_3),\\
&\lambda_3(t-\alpha)(1-t+\alpha), & t \in [ c_3, c_2),\\
&\lambda_1, & t \in [c_2,1-c_2),\\
&\lambda_3(t+\alpha)(1-t-\alpha), & t \in [1-c_2,1-c_3),\\
&\lambda_2(1-t-\alpha), & t \in [1-c_3, 1-\alpha).\\
\end{aligned}
\right.$$


{
\small
\begin{thebibliography}{10}


\bibitem[\protect\citeauthoryear{Acerbi}{Acerbi}{2002}]{A02}
{Acerbi, C.} (2002).  Spectral measures of risk: A coherent representation of subjective risk aversion. \emph{Journal of Banking and Finance}, \textbf{26}(7), 1505--1518.

\bibitem[\protect\citeauthoryear{Amarante}{Amarante}{2022}]{A22}
Amarante, M. (2022). A unified framework for Bayesian and non-Bayesian decision making and inference. \emph{Mathematics of Operations Research}, \textbf{47}(4), 2721--2742.


\bibitem[\protect\citeauthoryear{Artzner et al.}{Artzner et al.}{1999}]{ADEH99}
{Artzner, P., Delbaen, F., Eber, J.-M. and Heath, D.} (1999). Coherent measures of risk. \emph{Mathematical Finance}, \textbf{9}(3), 203--228.

\bibitem[\protect\citeauthoryear{Barrieu and El Karoui}{Barrieu and El Karoui}{2005}]{BE05}
Barrieu, P. and El Karoui, N. (2005). Inf-convolution of risk measures and optimal risk transfer.
\emph{Finance and Stochastics},  \textbf{9}, 269--298.


\bibitem[\protect\citeauthoryear{Bellini et al.}{2022}]{BFWW22}
Bellini, F., Fadina, T., Wang, R. and Wei, Y. (2022). Parametric measures of variability induced by risk measures. \emph{Insurance: Mathematics and Economics}, \textbf{106}, 270--284. 


\bibitem[\protect\citeauthoryear{Boonen and Ghossoub}{2020}]{BG20}
Boonen, T. J. and Ghossoub, M. (2020).
Bilateral risk sharing with heterogeneous beliefs and exposure constraints.
\emph{ASTIN Bulletin}, \textbf{50}(1), 
293--323.




\bibitem[\protect\citeauthoryear{Carlier and Dana}{Carlier  and Dana}{2003}]{CD03}
 Carlier, G. and Dana, R.-A.  (2003). Core of convex distortions of a probability. \emph{Journal of Economic Theory}, \textbf{113}, 199--222.



\bibitem[\protect\citeauthoryear{Carlier et al.}{Carlier et al.}{2012}]{CDG12}
 Carlier, G.,  Dana, R.-A. and Galichon, A. (2012). Pareto efficiency for the concave order and multivariate
comonotonicity. \emph{Journal of Economic Theory}, \textbf{147}, 207--229.

\bibitem[\protect\citeauthoryear{Cerreia-Vioglio et al.}{2011}]{CMMM11}
Cerreia-Vioglio, S.,  Maccheroni, F., Marinacci, M. and Montrucchio, L. (2011). Risk measures: Rationality and diversification. \emph{Mathematical Finance},  {\bf 21}(4), 743--774.

\bibitem[\protect\citeauthoryear{Cerreia-Vioglio et al.}{Cerreia-Vioglio et al.}{2014}]{CMMR14}
Cerreia-Vioglio, S., Maccheroni, F., Marinacci, M. and Rustichini, A. (2014). Niveloids and their extensions: Risk measures on small domains. \emph{Journal of Mathematical Analysis and Applications}, \textbf{413(1)}, 343-360.


\bibitem[\protect\citeauthoryear{Cheung and Lo}{Cheung and Lo}{2017}]{CL17}{Cheung, K. C. and Lo, A.} (2017). Characterizations of optimal reinsurance treaties: a cost-benefit approach. \emph{Scandinavian Actuarial Journal}, \textbf{2017}(1), 1--28.

\bibitem[\protect\citeauthoryear{Denneberg}{1990}]{D90}
Denneberg, D. (1990). Premium calculation: Why standard deviation should be replaced by absolute deviation. \emph{ASTIN Bulletin}, \textbf{20}(2), 181--190.

\bibitem[\protect\citeauthoryear{Denneberg}{1994}]{D94}
Denneberg, D. (1994). \emph{Non-additive Measures and Integral}. Kluwer, Dordrecht.

 
\bibitem[\protect\citeauthoryear{Dhaene et al.}{Dhaene et al.}{2012}]{DKLT12}
{Dhaene, J. and Kukush, A., Linders, D. and Tang, Q.} (2012). Remarks on quantiles and distortion risk measures. \emph{European Actuarial Journal}, \textbf{2}(2), 319--328.

\bibitem[\protect\citeauthoryear{Embrechts et al.}{2018}]{ELW18}
Embrechts, P., Liu, H. and Wang, R. (2018). Quantile-based risk sharing. \emph{Operations Research}, \textbf{66}(4), 936--949.

\bibitem[\protect\citeauthoryear{Embrechts et al.}{2020}]{ELMW20}
Embrechts, P., Liu, H., Mao, T. and Wang, R. (2020). Quantile-based risk sharing with heterogeneous beliefs. \emph{Mathematical Programming Series B}, \textbf{181}(2), 319--347.

\bibitem[\protect\citeauthoryear{Embrechts et al.}{Embrechts et al.}{2002}]{EMS02}
Embrechts, P., McNeil, A. and Straumann, D. (2002). Correlation and dependence in risk management: properties and pitfalls. In \emph{Risk Management: Value at Risk and Beyond} (Ed. Dempster, M. A. H.), 176--223, Cambridge University Press. 



\bibitem[\protect\citeauthoryear{Filipovic and Svindland}{Filipovi{\'c} and Svindland}{2008}]{FS08} Filipovic, D and Svindland, G. (2008).  Optimal capital and risk allocations for law- and cash-invariant convex functions. \emph{Finance and Stochastics}, \textbf{12}(3), 423--439.


 \bibitem[\protect\citeauthoryear{F\"ollmer and Schied}{F\"ollmer and Schied}{2016}]{FS16} F\"ollmer, H.~and Schied, A.~(2016). \emph{Stochastic Finance. An Introduction in Discrete Time}. Fourth Edition.  {Walter de Gruyter, Berlin}.




\bibitem[\protect\citeauthoryear{Furman et al.}{Furman et al.}{2017}]{FWZ17}
{Furman, E., Wang, R. and Zitikis, R.} (2017). Gini-type measures of risk and variability: Gini shortfall, capital allocation and heavy-tailed risks. \emph{Journal of Banking and Finance}, \textbf{83}, 70--84.



\bibitem[\protect\citeauthoryear{Grechuk, Molyboha and Zabarankin}{Grechuk et al.}{2009}]{GMZ09}
Grechuk, B., Molyboha, A. and Zabarankin, M. (2009). Maximum entropy principle with general deviation measures. \emph{Mathematics of Operations Research}, \textbf{34}(2), 445--467. 




\bibitem[\protect\citeauthoryear{Jouini et al.}{2008}]{JST08}
Jouini, E., Schachermayer, W. and Touzi, N. (2008). Optimal risk sharing for law invariant monetary utility functions. \emph{Mathematical Finance}, \textbf{18}(2), 269--292.

\bibitem[\protect\citeauthoryear{Konno and Yamazak}{1991}]{KY91}
Konno, H. and Yamazaki, H. (1991). Mean-absolute deviation portfolio optimization model and its applications to Tokyo stock market. \emph{Management Science}, \textbf{37}(5), 519--531.


\bibitem[\protect\citeauthoryear{Landsberger and Meilijson}{Landsberger and Meilijson}{1994}]{LM94}
Landsberger, M. and  Meilijson, I. (1994). Co-monotone allocations, Bickel-Lehmann dispersion and the Arrow-Pratt measure of risk aversion. \textit{Annals of Operations Research}, \textbf{52}(2), 97--106.


\bibitem[\protect\citeauthoryear{Lauzier et al.}{Lauzier et al.}{2023}]{LLW23}
Lauzier, J.G., Lin, L. and Wang, R. (2023). Pairwise counter-monotonicity. \textit{Insurance: Mathematics and Economics}, \textbf{111}, 279--287.


\bibitem[\protect\citeauthoryear{Liu}{2020}]{L20}
Liu, H. (2020). Weighted comonotonic risk sharing under heterogeneous beliefs. \emph{ASTIN Bulletin}, \textbf{50}(2), 647--673.


\bibitem[\protect\citeauthoryear{Liu et al.}{2022}]{LMWW22}
Liu, F., Mao, T., Wang, R. and Wei, L. (2022). Inf-convolution, optimal allocations, and model uncertainty for tail risk measures. \emph{Mathematics of Operations Research}, \textbf{47}(3), 2494--2519.



\bibitem[\protect\citeauthoryear{Liebrich}{2024}]{Lieb24} Liebrich, FB. (2024). Risk sharing under heterogeneous beliefs without convexity. \emph{Finance and Stochastics} \textbf{28}, 999--1033.


\bibitem[\protect\citeauthoryear{Ludkovski and R\"uschendorf}{2008}]{LM08}
Ludkovski, M. and R\"uschendorf, L. (2008). On comonotonicity of Pareto optimal risk sharing. \emph{Statistics and Probability Letters}, \textbf{78}(10), 1181--1188.

\bibitem[\protect\citeauthoryear{Liu et al.}{2020}]{LWW20}
Liu, P., Wang, R. and Wei, L. (2020). Is the inf-convolution of law-invariant preferences law-invariant? \emph{Insurance: Mathematics and Economics}, \textbf{91}, 144--154.

\bibitem[\protect\citeauthoryear{Markowitz}{1952}]{MARK52}
Markowitz, H. (1952). Portfolio selection. \emph{The Journal of Finance}, \textbf{7}(1), 77--91.

\bibitem[\protect\citeauthoryear{Markowitz}{2014}]{M14}
Markowitz, H. (2014). Mean–variance approximations to expected utility. \emph{European Journal of Operational Research}, \textbf{234}(2), 346--355.


\bibitem[\protect\citeauthoryear{Maccheroni et al.}{2013}]{MMR13}
Maccheroni, F., Marinacci, M. and Ruffino, D. (2013). Alpha as ambiguity: Robust mean‐variance portfolio analysis. \emph{Econometrica}, \textbf{81}(3), 1075--1113.


\bibitem[\protect\citeauthoryear{Marinacci and Montrucchio}{Marinacci and Montrucchio}{2004}]{MM04}
{Marinacci, M. and Montrucchio L.} (2004). Introduction to the mathematics of ambiguity.  In \emph{Uncertainty in Economic Theory}, I. Gilboa, Ed.  Routledge, New York, NY, USA. 46--107.




\bibitem[\protect\citeauthoryear{Negishi}{Negishi}{1960}]{N60}
Negishi, T. (1960). Welfare economics and existence of an equilibrium for a competitive economy. \emph{Metroeconomica}, \textbf{12}(2‐3), 92--97.


\bibitem[\protect\citeauthoryear{Kusuoka}{Kusuoka}{2001}]{K01}
{Kusuoka, S.} (2001). On law invariant coherent risk measures. \emph{Advances in Mathematical Economics}, \textbf{3}, 83--95.

 \bibitem[\protect\citeauthoryear{Puccetti and Wang}{Puccetti and
  Wang}{2015}]{PW15}
Puccetti, G. and Wang R.  (2015).
Extremal dependence concepts.
 \emph{Statistical Science},  \textbf{30}(4), 485--517.

      \bibitem[\protect\citeauthoryear{Ravanelli and Svindland}{Ravanelli and Svindland}{2014}]{RS14}
Ravanelli, C. and Svindland, G. (2014). Comonotone Pareto optimal allocations for law invariant robust
utilities on $L^1$. \emph{Finance and Stochastics}, \textbf{18}(1), 249--269.
 
\bibitem[\protect\citeauthoryear{Rockafellar et al.}{Rockafellar et al.}{2006}]{RUZ06}
Rockafellar, R. T., Uryasev, S. and Zabarankin, M. (2006). Generalized deviation in risk analysis. \emph{Finance and Stochastics}, \textbf{10}, 51--74.
 
\bibitem[\protect\citeauthoryear{Rostek}{Rostek}{2010}]{R10}
 Rostek, M. (2010).
Quantile maximization in
decision theory. \emph{Review of Economic Studies}, \textbf{77}, 339--371.


 

 
\bibitem[\protect\citeauthoryear{Rothschild and Stiglitz}{1970}]{RS70}
Rothschild, M. and Stiglitz, J. E. (1970). Increasing risk: I. A definition. \emph{Journal of Economic Theory}, \textbf{2}(3), 225--243.






\bibitem[\protect\citeauthoryear{R{\"u}schendorf}{R{\"u}schendorf}{2013}]{R13}
R{\"u}schendorf, L. (2013).
  {\em Mathematical Risk Analysis. Dependence, Risk Bounds, Optimal
  Allocations and Portfolios}.
  Springer, Heidelberg.
  
\bibitem[\protect\citeauthoryear{Schmeidler}{Schmeidler}{1989}]{S89}
{Schmeidler, D.} (1989). Subjective probability and expected utility without additivity.
\emph{Econometrica}, \textbf{57}(3), 571--587.


  
 \bibitem[\protect\citeauthoryear{Shaked and Shanthikumar}{Shaked and Shanthikumar}{2007}]{SS07} Shaked, M. and Shanthikumar, J. G. (2007). \emph{Stochastic Orders}. Springer Series in Statistics.

 \bibitem[\protect\citeauthoryear{Shalit and Yitzhaki}{1984}]{SY84}
Shalit, H. and Yitzhaki, S. (1984). Mean-Gini, portfolio theory, and the pricing of risky assets. \emph{The Journal of Finance}, \textbf{39}(5), 1449--1468. 

\bibitem[\protect\citeauthoryear{Wakker}{Wakker}{2010}]{W10}
Wakker, P. P. (2010). \emph{Prospect Theory: For Risk and Ambiguity}. Cambridge University Press.

\bibitem[\protect\citeauthoryear{Wang et al.}{Wang et al.}{2020a}]{WWW20a}
Wang, Q., Wang, R. and Wei, Y.  (2020a). Distortion riskmetrics on general spaces. \emph{ASTIN Bulletin}, \textbf{50}(4), 827--851. 


\bibitem[\protect\citeauthoryear{Wang}{Wang}{1996}]{W96}
Wang, S. (1996). Premium calculation by transforming the layer premium density. \emph{ASTIN Bulletin}, \textbf{26}, 71--92.


\bibitem[\protect\citeauthoryear{Wang et al.}{Wang et al.}{2020b}]{WWW20b}
Wang, R., Wei, Y. and Willmot, G. E. (2020b). Characterization, robustness and aggregation of signed Choquet integrals. \emph{Mathematics of Operations Research}, \textbf{45}(3), 993--1015.

\bibitem[\protect\citeauthoryear{Wang and Zitikis}{2021}]{WZ21}
 Wang, R. and Zitikis, R. (2021). An axiomatic foundation for the Expected Shortfall. \emph{Management Science}, \textbf{67}, 1413--1429.
 
\bibitem[\protect\citeauthoryear{Wang and Wei}{Wang and Wei}{2020}]{WW20} Wang, R. and Wei, Y. (2020). Risk functionals with convex level sets. \emph{Mathematical Finance},  \textbf{30}(4), 1337--1367.

\bibitem[\protect\citeauthoryear{Weber}{Weber}{2018}]{W18} Weber, S. (2018). Solvency II, or how to sweep the downside risk under the carpet. \emph{Insurance: Mathematics and Economics}, \textbf{82}, 191--200.

\bibitem[\protect\citeauthoryear{Yaari}{Yaari}{1987}]{Y87}
{Yaari, M. E.} (1987). The dual theory of choice under risk. \emph{Econometrica}, \textbf{55}(1), 95--115.


\bibitem[\protect\citeauthoryear{Yule}{Yule}{1911}]{Y11}
Yule, G. U. (1911). \emph{An Introduction to the Theory of Statistics}. Charles Griffin and Company.





\end{thebibliography}
}
\end{document}